\documentclass[aps,epsf,amsmath,superscriptaddress,citeautoscript,showpacs,floatfix,nofootinbib,prl,twocolumn]{revtex4-2}
\usepackage[T1]{fontenc}
\usepackage{txfonts}

\usepackage[dvipsnames]{xcolor}
\definecolor{red}{rgb}{0.9, 0,0}
\definecolor{cerulean}{rgb}{0., 0.42,0.9}
\definecolor{navy}{rgb}{0.05, 0.05,0.8}

\usepackage{setspace}
\usepackage[colorlinks]{hyperref}
\hypersetup{
    colorlinks = true,
    citecolor  = red,
	linkcolor  = navy
}

\usepackage{bm}

\usepackage{multirow}
\usepackage{slashed}
\usepackage{graphics}
\usepackage{amsmath}
\usepackage{amssymb}
\usepackage{latexsym}
\usepackage{mathtools}
\usepackage{dsfont}
\usepackage{cleveref}
\usepackage{hyperref}
\usepackage{amsfonts}
\usepackage{enumitem}
\usepackage{xstring} 
\usepackage{xspace} 
\usepackage[normalem]{ulem}
\usepackage{fontawesome}
\usepackage{braket}
\usepackage{mathrsfs}

\usepackage{amsthm}

\newtheorem{lemma}{Lemma}
\newtheorem{theorem}{Theorem}
\newtheorem{remark}{Remark}

\newcommand{\be}{\begin{equation}}
\newcommand{\ee}{\end{equation}}
\newcommand{\bs}{\begin{split}}
\newcommand{\es}{\end{split}}

\newcommand{\simgt}{\lower.5ex\hbox{$\; \buildrel > \over \sim \;$}}
\newcommand{\simlt}{\lower.5ex\hbox{$\; \buildrel < \over \sim \;$}}

\newcommand{\comment}[1]{}

\usepackage{subfiles} 

\begin{document}

\title{Universal Bound for Entanglement Generation}
\author{Alfred Li}
\email{ali2@caltech.edu}
\affiliation{Burke Institute for Theoretical Physics, California Institute of Technology, Pasadena, California 91125, USA}

\author{Daisuke Miki}
\email{dmiki@caltech.edu}
\affiliation{Burke Institute for Theoretical Physics, California Institute of Technology, Pasadena, California 91125, USA}

\author{Yanbei Chen}
\email{yanbei@caltech.edu}
\affiliation{Burke Institute for Theoretical Physics, California Institute of Technology, Pasadena, California 91125, USA}

\date{\today}

\begin{abstract}
    We derive a universal condition for entanglement generation under general bilinear interactions in the presence of white thermal noise. While various protocols have been proposed to enhance the amount of generated entanglement, it remains unclear whether they can also relax the threshold for entanglement generation itself. Using a Gorini-Kossakowski-Sudarshan-Lindblad description, we analyze general multimode systems and derive a separability-preserving condition for bilinear interactions under white thermal noise. As an application to gravity-induced entanglement, we show that the gravitational interaction must dominate over thermal noise for entanglement to arise. In particular, this bound cannot be relaxed by changing the initial state or by introducing mediator systems, although such ingredients may enhance the amount of entanglement once it is generated. These results establish a general limitation on entanglement-generation protocols in thermal environments.
\end{abstract}

\maketitle

\textit{Introduction---}
Recently, gravity-induced entanglement has emerged as a particularly promising route toward experimentally probing the quantum nature of gravity \cite{bose2017,marletto2017, bose2025rmp, marletto2025}:
if two initially separable systems become entangled solely through their mutual gravitational interaction, then gravity cannot be described as a channel implementable by local operations and classical communication (LOCC). In this sense, entanglement generation provides an operational witness of nonclassicality in gravity and has motivated extensive studies from several perspectives, including gravitons and quantum field theory, the consistency of complementarity and causality \cite{mari2016,belenchia2018,danielson2022,sugiyama2023,sugiyama2024}, and the dynamical quantization of gravity under unitarity and Lorentz invariance \cite{carney2022}. 
{Related operational tests have also been proposed that do not rely directly on entanglement generation, including tests of LOCC-incompatibility of gravitational dynamics \cite{lami2024} and of whether gravity can act as a nonclassical channel for quantum information \cite{mari2025}.}

A major challenge, however, is that gravity is extremely weak. In Ref.~\cite{miao2020}, 
the generation of gravity-induced entanglement in two-mode Gaussian systems was shown to require a condition that the gravitational interaction overcome thermal decoherence, 
$\hbar Gm/d^3>\gamma k_BT$, where $G$ is the gravitational constant, $m$ is the mass, $d$ is the separation between two masses, $\gamma$ is the mechanical dissipation rate, and $T$ is the environmental temperature. This condition is extremely demanding,
requiring a stringent combination of large
mass, small separation, low temperature, low damping, and strong
decoherence suppression \cite{krisnanda2020observable,qvarfort2020mesoscopic,datta2021signatures,miki2024feasible,plato2023enhanced,tilly2021qudits,schut2022resilience,rijavec2021decoherence}. Even taking $m/d^3$ to be of order a
dense solid density, $\rho\sim 10^4\,{\rm kg/m^3}$, gives
$\gamma T \lesssim 10^{-17}\,{\rm K/s}$. 
%
Recent optimization of oscillator shapes have lead to an enhancement, but only within one order of magnitude~\cite{tang2025formfactor}. 
Existing experimental platforms demonstrate this difficulty: milligram-scale pendula and torsional sensors can reach large quality
factors, about $Q\sim10^6$ for milligram pendula and up to $Q\sim10^5$ for torsional pendula, but Hz-scale devices with such $Q$ would still require temperatures as low as $T\sim10^{-11}\,{\rm K}$ \cite{schmole2016micromechanical,matsumoto2019mg,catano2020highq,komori2020atto,agafonova2026milligram}. Very-low-frequency torsional designs partially relax this constraint: TOBA's Phase III design will have $\omega\sim 2\pi \times 10^{-3}\,{\rm Hz}$ and $Q\sim 10^8$; corresponding to $\gamma\sim 10^{-10}\,{\rm s}^{-1}$, and $T\sim10^{-7}\,$K~\cite{Oshima:2023Du}. Magnetically levitated systems
offer cryogenic operation and quality factors above $10^7$, and have enabled precision tests of weak gravity with mesoscopic test masses, but 
must address the
suppression of residual magnetic interactions below the gravitational signal \cite{vinante2020ultralow,hofer2023highq,timberlake2021modified,fuchs2024measuring}. Finally, optically levitated nanoparticles have achieved motional
ground-state cooling and coherent quantum control in the $10^5$ Hz range, but the currently used  small
masses ($\sim10^{-18}\rm{kg}$) and total motional heating rates ($\Gamma_{\text{heat}}/(2\pi)\sim 10^{3}-10^{4} \ \rm{Hz}$) 
currently makes gravitational interaction too
weak compared with environmental decoherence
\cite{delic2020cooling,tebbenjohanns2021quantum,piotrowski2023simultaneous}. 

%

In response to this, several approaches have been proposed to enhance gravity-induced entanglement, including the use of a massive mediator \cite{pedernales2022}, higher-order interactions in optomechanical systems \cite{kaku2023}, inverted harmonic potentials \cite{fujita2025,shiomatsu2025}, input optical squeezing \cite{hatakeyama2026}, conditional momentum squeezing of massive objects \cite{fukuzumi2026}, and pulsed optomechanical protocols \cite{miki2026}. 
In the massive mediator proposal \cite{pedernales2022}, the mediator can enhance the effective gravitational interaction, but thermal noise has not been fully incorporated into the analysis. In the other approaches listed above, the amount of entanglement can be enhanced once it is generated, but the onset of entanglement remains constrained by thermal noise and is still governed by the same condition.  At the same time, related analyses of mediated systems have shown that the introduction of an additional oscillator can also suppress long-time-averaged bipartite entanglement in certain regimes \cite{christopher2025}. These results raise a basic and, to our knowledge, still unresolved question: can any of these methods relax the fundamental threshold for entanglement generation itself, or do they merely amplify entanglement after the separability-breaking threshold has already been crossed?

In this work, we address this question in a general setting. We derive a universal separability-preserving condition for bilinearly coupled systems subject to white thermal noise, without assuming any particular initial state. The result establishes a state-independent threshold for entanglement generation, set by the competition between coherent interaction and thermal decoherence. Applied to gravity-induced entanglement, this yields a specific quantitative requirement: the gravitational coupling must exceed the corresponding thermal-noise bound before entanglement can arise. Different initial states or mediator systems may enhance the entanglement once generated, but they cannot lower this threshold. In this sense, our result extends the criterion of Ref.~\cite{miao2020} from its original setting to a broad, state-independent separability-preservation bound.
%
%
Closely related bounds have been discussed in Refs.~\cite{kafri2013,kafri2014,diosi2017,fabiano2026} for semiclassical gravity models: a Newtonian interaction mediated by an LOCC or otherwise non-entangling channel must be accompanied by sufficient force noise. While the noise arises from a different mechanism, the underlying balance is the same—coherent gravitational interaction versus decoherence. 

In the rest of this paper, we first derive the bound in the covariance-matrix formalism for multimode Gaussian systems, generalizing the two-mode condition of Refs.~\cite{kafri2013,kafri2014}. We then extend it to general systems using a GKSL description, following Ref.~\cite{diosi2017}, and construct an explicit LOCC representation that yields the separability condition. Finally, we show that mediator degrees of freedom can be absorbed into either side of a fixed bipartition; mediators may reshape the dynamics and/or enhance the amount of entanglement, but cannot evade the universal threshold.

\textit{General Gaussian Entanglement---}We first consider entanglement between two multimode Gaussian systems subject to a bilinear interaction and white thermal noise, extending the two-mode treatment of Refs.~\cite{kafri2013,kafri2014}. We define the canonical column vector $\hat{\bm{\xi}} = (\hat{\bm{\xi}}_{A}^T,\hat{\bm{\xi}}_{B}^T)^T$, where $\hat{\bm{\xi}}_{i} = (\hat{x}^{i}_1,\hat{p}^{i}_1,\ldots,\hat{x}^{i}_{n},\hat{p}^{i}_{n})^T$ for $i\in\{A,B\}$, with canonical commutation relation
\begin{equation}
    \left[\hat{\bm{\xi}},\hat{\bm{\xi}}^T\right] = i\bm{\Omega} \equiv i(\bm{\eta}_{A}\oplus\bm{\eta}_{B}), \hspace{0.3cm} \bm{\eta}_{i} \equiv \bigoplus_{j=1}^{n}\begin{pmatrix}
        0 & 1\\
        -1 & 0
    \end{pmatrix}.
\end{equation}
We work in units with $\hbar=1$ and assume $n_A=n_B=n$ for simplicity. For Gaussian states, separability can be tested at the level of the covariance matrix $\bm V(t) \equiv \tfrac{1}{2}\langle\{\hat{\bm \xi}(t)-\langle\hat{\bm \xi}(t)\rangle,(\hat{\bm \xi}(t)-\langle\hat{\bm \xi}(t)\rangle)^T\}\rangle$ via the PPT criterion~\cite{peres1996,horodecki1996,duan2000,simon2000}.

We consider a general bilinear system Hamiltonian coupled to a Gaussian thermal bath,
\begin{equation}
    \hat{H} =\frac{1}{2}\hat{\bm{\xi}}^{T}\textbf{G}\hat{\bm{\xi}}+ \frac{1}{2}(\hat{\bm{\xi}}^{\text{bath}})^{T}\textbf{H}\hat{\bm{\xi}}^{\text{bath}}
        +\hat{\bm{\xi}}^{T}\hat{\bm F}^{\text{th,bath}},
    \label{general hamiltonian}
\end{equation}
specialized for now to a rank-1 coupling with matched thermal noise on each subsystem,
\begin{align}
\label{rank 1 interaction hamiltonian}
    \bm G=\left(\begin{array}{cc}
    \bm \eta^{-1}\bm M_A & K_g \bm u_A \bm u_B^{T} \\[4pt]
    K_g \bm u_B\bm u_A^{T} & \bm \eta^{-1}\bm M_B
    \end{array}\right),\quad
    \hat{\bm F}^{\rm th,bath}=
    \left(\begin{array}{c}
    \hat{F}_A^{\rm th} \bm u_A \\
    \hat{F}_B^{\rm th} \bm u_B
    \end{array}\right),
\end{align}
where $\bm M_i$ encodes the free dynamics of subsystem $i$ and $\bm u_{A,B}$ are real vectors selecting which quadratures couple. The choice $\bm u_A=\bm u_B=(1,0,\ldots,0)^T$ recovers the familiar gravitational interaction $K_g\hat{x}_A\hat{x}_B$ (see Fig.~\ref{fig:rank 1}). $\hat{\bm{F}}^{\text{th},\text{bath}}$ represents the thermal force vector, which has white-noise spectrum $\frac{1}{2}\langle\{\hat{F}_{i,I}^{\rm th}(t),\hat{F}_{j,I}^{\rm th}(t')\}\rangle=\delta_{ij}\delta(t-t')S^{\rm th}_{i}$ (a subscript $I$ denotes the interaction picture, a subscript $H$, the Heisenberg picture, and if omitted, the Schrodinger picture.)

As a concrete realization, consider the Newtonian gravitational interaction between two equal masses $m$ at separation $d$: $K_g=2Gm^2/d^3$. The thermal noise on each subsystem is set by the system-bath coupling and bath temperature; in the high-temperature limit with equal mechanical dissipation rate $\gamma$, the spectra read $S^{\rm th}_A = S^{\rm th}_B = 2\gamma m k_B T$. We will use these scales to interpret the entanglement-generation bound below.

\begin{figure}
    \centering
    \includegraphics[width=0.63\linewidth]{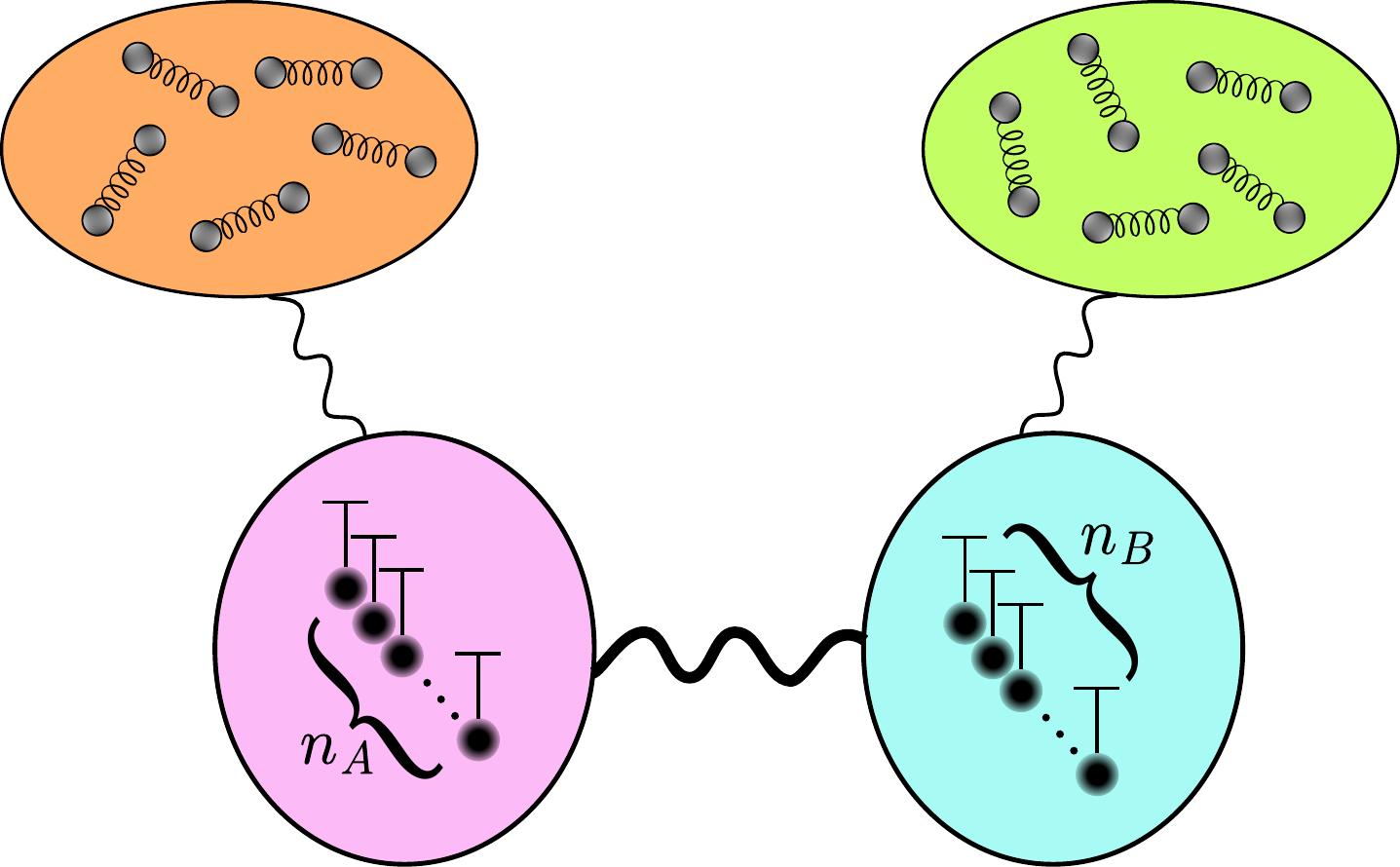}
    \caption{Configuration of rank-1 bilinearly coupled systems: $n$ oscillators in subsystem $A$ (pink) couple to $n$ oscillators in subsystem $B$ (blue) via a rank-1 bilinear interaction $K_{g}(\hat{\bm{\xi}}^{T}_{A}\textbf{u}_{A})(\hat{\bm{\xi}}^{T}_{B}\textbf{u}_{B})$ (thick line). Each subsystem couples linearly to its own thermal bath (orange and green, thin lines).}
    \label{fig:rank 1}
\end{figure}

The Heisenberg equations of motion give a simple Langevin equation $\frac{d}{dt}\hat{\bm{\xi}}_H(t) = (\bm M+K_g\bm C)\hat{\bm{\xi}}_H(t)+\hat{\bm n}_H(t)$, where $\bm M=\bm M_A\oplus\bm M_B$, $\bm C$ has off-diagonal blocks $\bm C_{AB}=\bm \eta\bm u_A\bm u_B^T$, $\bm C_{BA}=\bm \eta\bm u_B\bm u_A^T$, and $\hat{\bm n}_{i,H}=\bm \eta\bm u_i\hat F_{i,H}^{\rm th}$. Neglecting second-order corrections, valid on timescales short compared to dissipation (see Section `Discussion on $\hat{F}_{H}(t)$ vs $\hat{F}_{I}(t)$' of Supplemental Material~\cite{SM}), the covariance matrix evolves as
\begin{equation}
        \bm{V}(t) = \bm \Phi(t)\bm{V}(0)\bm \Phi^T(t) + \int_{0}^{t}du \ \bm \Phi(u)\bm{F}\bm \Phi^T(u),
\end{equation}
with symplectic propagator $\bm{\Phi}(t)\equiv e^{(\bm M+K_g\bm C)t}$ and $\bm F=(\bm{\eta}\bm u_A\bm u_A^{T}\bm{\eta}^{T}S^{\rm th}_{A})\oplus(\bm{\eta}\bm u_B\bm u_B^{T}\bm{\eta}^{T}S^{\rm th}_{B})$.

\begin{theorem}
\label{pert-theorem}
For a separable pure Gaussian initial state, in the perturbative regime $K_{g}t,\ S^{\text{th}}_{A}t,\ S^{\text{th}}_{B}t\ll 1$, the dynamics of \eqref{rank 1 interaction hamiltonian} preserve separability provided
\begin{equation}
    \label{sufficient cond}
    S^{\text{th}}_{A}S^{\text{th}}_{B} \geq K_{g}^{2}.
\end{equation}
\end{theorem}
\begin{remark}
\label{pert-correlated}
    If in addition, the noise is correlated so that $\frac{1}{2}\langle\{\hat{F}_{A,I}^{\rm th}(t),\hat{F}_{B,I}^{\rm th}(t')\}\rangle=\delta(t-t')S^{\rm th}_{AB}$, then separability is preserved provided that:
\begin{equation}
     S^{\text{th}}_{A}S^{\text{th}}_{B} \geq K_{g}^{2}+(S^{\text{th}}_{AB})^2.
\end{equation}
\end{remark}

The full proof of Theorem \ref{pert-theorem} and Remark \ref{pert-correlated} are given in Section `Full Proof of Gaussian Calculation' and subsection `Including Correlated Noise' therein of the Supplemental Material~\cite{SM}. The strategy uses the multimode separability criterion~\cite{serafini2017}: a Gaussian state is separable if and only if $\bm V\succeq \bm \sigma_A\oplus\bm \sigma_B$ for some valid local covariance matrices $\bm \sigma_{A,B}$ (i.e., satisfying $\bm \sigma_j+\tfrac{i}{2}\bm \eta_j\succeq 0$). It therefore suffices to exhibit a decomposition
\begin{equation}
    \bm V(t) = (\bm \sigma_A\oplus\bm \sigma_B) + \bm N, \qquad \bm N\succeq 0.
\end{equation}
We construct such $\bm \sigma_A,\bm \sigma_B,\bm N$ explicitly to first order in $K_g t$ and $S^{\rm th}_{A,B}t$, and the positivity $\bm N\succeq 0$ reduces to \eqref{sufficient cond}. Under the stringent condition where $\textbf{u}_{A}$ and $\textbf{u}_{B}$ are parallel, and are eigenvectors under the symplectic propagator $\bm{\Phi}(t)$, it is possible to boost this sufficiency condition to a necessary one, more details in subsection `Stringent Necessary and Sufficient' of the Supplemental Material~\cite{SM}.

Inserting the gravitational and thermal scales introduced above, and restoring $\hbar$, the separability bound \eqref{sufficient cond} becomes $\gamma k_B T \ge \frac{\hbar Gm}{d^3}$, recovering the entanglement-generation threshold of Ref.~\cite{miao2020}. The multimode generalization thus shows that adding modes with identical couplings cannot relax this constraint; more general interactions are treated below.

\textit{General GKSL Argument---}While the argument of the previous section is illustrative in explicitly detailing the time-dependence of all contributions to the covariance matrix, it is limited to only Gaussian input states and Gaussian dynamics, and is only valid in the perturbative regime. Since we are interested in deriving a separability condition, another approach is to consider under what conditions the dynamics of \eqref{general hamiltonian} can be modeled by an LOCC protocol because these are by definition, separability preserving operations. The basics of this idea were  laid out in \cite{diosi2017}, and in this section we utilize it to derive a sufficient separability condition that applies to any separable input state and also extend the result to a general interaction between subsystems.

We first consider systems with `rank 1' interactions, i.e. those of the same form as \eqref{rank 1 interaction hamiltonian}. To verify if dynamics are re-creatable through LOCC schemes, we must work with the density matrix itself, thus requiring the Born-Markov equation \cite{schlosshauer2007decoherence}. Assuming the same physical system as \eqref{rank 1 interaction hamiltonian}, the master equation of the system density matrix $\hat{\bm{\rho}}$ is (derivation found in Section `Derivation of Rank-1 Interaction Born Markov Equation' of Supplemental Material \cite{SM}):
\begin{equation}
\label{rank 1 BM}
    \begin{split}
        \frac{d}{dt}\hat{\bm{\rho}}(t)  
        =-i[\hat{H}_{HO},\hat{\bm{\rho}}(t)]&-iK_{g}[\hat{\bm{\xi}}^{T}_{A}\textbf{u}_{A}\textbf{u}_{B}^{T}\hat{\bm{\xi}}_{B},\hat{\bm{\rho}}(t)]\\
        &+\sum_{\alpha=A,B}S^{\text{th}}_{\alpha}\mathcal{D}[\hat{\bm{\xi}}^{T}_{\alpha}\textbf{u}_{\alpha}]\hat{\bm{\rho}}(t)
    \end{split}
\end{equation}
where $\hat{H}_{HO} \equiv \frac{1}{2}\hat{\bm{\xi}}^{T}\bm{\Omega}^{-1}\textbf{M}\hat{\bm{\xi}}$ and we have defined the dissipator:
\begin{equation}
    \mathcal{D}[\hat{\bm{\xi}}^{T}_{\alpha}\textbf{u}_{\alpha}]\hat{\bm{\rho}}(t) = \hat{\bm{\xi}}^{T}_{\alpha}\textbf{u}_{\alpha}\hat{\bm{\rho}}(t)\hat{\bm{\xi}}^{T}_{\alpha}\textbf{u}_{\alpha}-\frac{1}{2}\{(\hat{\bm{\xi}}^{T}_{\alpha}\textbf{u}_{\alpha})^2,\hat{\bm{\rho}}(t)\}
\end{equation}
Two key assumptions used in deriving this equation beyond the usual Born-Markov approximations are a white noise spectrum for $\hat{F}^{\text{th}}_{A,I}(t),\hat{F}^{\text{th}}_{B,I}(t)$ and an omission of the damping term that usually takes the form of a mixed commutator. Practically, this was necessary for our LOCC argument to work and physically, this means we are working on timescales much smaller than the damping/relaxation timescale. In certain special cases, the separability condition can be extended to include damping. However, doing so requires additional model-dependent constraints relating the damping and dissipator terms, thus greatly reducing the generality of our results. Further discussion of the damped case is given in Section `Discussion on Including Damping/Relaxation Term in GKSL Arguments' of Supplemental Material~\cite{SM}.

\begin{theorem} 
\label{rank 1 theorem}
For multi-mode systems, in the regime where the Born-Markov approximation holds, and damping/relaxation may be ignored, and assuming a separable state input, a sufficient condition that  a subsystem coupling interaction of the form $K_{g}(\hat{\bm{\xi}}^{T}_{A}\textbf{u}_{A})(\hat{\bm{\xi}}^{T}_{B}\textbf{u}_{B})$ preserve separability in the presence of thermal noise is:
\begin{equation}
    S^{\text{th}}_{A}S^{\text{th}}_{B} \geq K_{g}^{2}
\end{equation}
where $K_{g}$ is the strength of the bilinear coupling and $S^{\text{th}}_{A},S^{\text{th}}_{B}$ are the white noise spectra of the uncorrelated external baths.
\end{theorem}
\begin{remark}
    \label{rank 1 correlated}
    If in addition, the noise between baths is also correlated and described by $S^{\text{th}}_{AB}$, then separability is preserved provided that:
    \begin{equation}
        S^{\text{th}}_{A}S^{\text{th}}_{B} \geq K_{g}^2+(S^{\text{th}}_{AB})^2
    \end{equation}
\end{remark}

The full details of the proofs can be found in Sections `Full GKSL Proof for Rank 1 Interaction' and Subsection `Including Correlated Noise' therein of Supplemental Material~\cite{SM}.

The goal is to devise an LOCC protocol whose Lindbladian exactly reproduces \eqref{rank 1 BM}. Since the term $H_{HO}$ involves only local terms and can be implemented using unitaries, we can ignore it and only focus on the dissipator and entanglement generating terms. 

Now consider a symmetric measurement-and-feedback protocol where Alice and Bob individually implements local Gaussian weak measurements of their respective local operators $\hat{X}_{A}$ and $\hat{X}_{B}$ with measurement strengths $\gamma_{A},\gamma_{B} > 0$ and obtain classical measurement outcomes $y_{A},y_{B}\in\mathbb{R}$. Given the classical outcomes, Alice and Bob can apply local unitaries $\hat{U}_{A}(y_{B}), \hat{U}_{B}(y_{A})$ conditioned on the other party's outcome with feedback strengths $\lambda_{A},\lambda_{B}\in\mathbb{R}$. This step requires only local communication of classical results. 
Expanding the action of the unconditional map corresponding to averaging over the classical outcomes of the measurement-feedback scheme to linear order in $dt$, one derives a Lindbladian. By forcing this Lindbladian to match \eqref{rank 1 BM}, the bound on thermal noise emerges as a sufficient condition.
%
%
We emphasize that,  while the 
interaction terms are quadratic, 
Gaussianity was never required nor used anywhere in the proof. Since the result of Theorem \ref{rank 1 theorem} matches \eqref{sufficient cond} 
we have 
now shown that any non-Gaussianity will not relax this sufficiency condition.

\textit{General Interaction---}We extend the measure-and-feedback scheme to general bilinear couplings. In the notation of \eqref{general hamiltonian}, the Hamiltonian now reads
\begin{align}
    \bm G=\left(\begin{array}{cc}
    \bm M_A & \bm Q_{G} \\[4pt]
   \bm Q_{G}^{T} & \bm M_B
    \end{array}\right),\qquad
    \hat{\bm F}^{\rm th,bath}=
    \left(\begin{array}{c}
    \hat{\bm F}_A^{\rm th} \\
    \hat{\bm{F}}_B^{\rm th}
    \end{array}\right),
\end{align}
where $\bm{M}_{i}^{T}=\bm{M}_{i}$ describes the internal dynamics of subsystem $i$, $\bm{Q}_{G}$ is the bipartite coupling matrix, and each component of the vector operator $\hat{\bm F}_i^{\rm th}$ couples independently to a different bath quadrature, with block-diagonal white-noise spectral matrix $\frac{1}{2}\langle\{(\hat{F}_{i,I}^{\rm th})_a(t),(\hat{F}_{j,I}^{\rm th})_b(t')\}\rangle = \delta_{ij}\delta_{ab}\delta(t-t')\bm Q^{\rm th}_{i,ab}$. The setup is illustrated in Fig.~\ref{fig: general int}.

\begin{figure}
    \centering
    \includegraphics[width=0.63\linewidth]{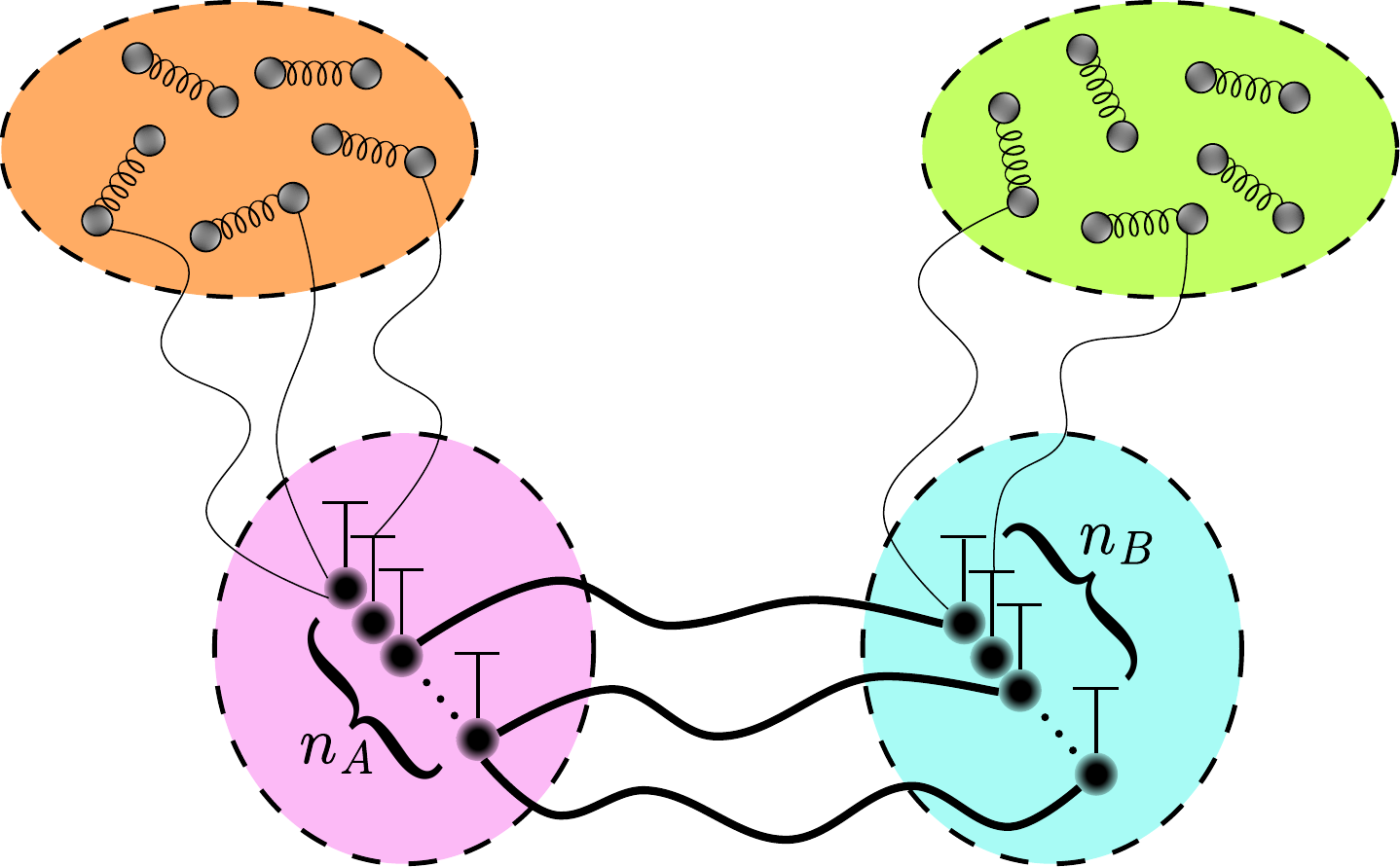}
    \caption{
    Systems with general bilinear coupling, with thicker lines representing gravity and thinner lines, system-bath couplings.
    }
    \label{fig: general int}
\end{figure}

\begin{theorem}
\label{general theorem}
 For multi-mode systems, in the regime where the Born-Markov approximation holds, and damping/relaxation may be ignored, and assuming a separable state input, a sufficient condition that a subsystem coupling interaction of the form $\hat{\bm{\xi}}^{T}\bf{G}\hat{\bm{\xi}}$ preserve separability in the presence of thermal noise is:
\begin{equation}
    \begin{pmatrix}
        \textbf{Q}^{\text{th}}_{A} & \textbf{Q}_{G}\\
        \textbf{Q}_{G}^{T} & \textbf{Q}^{\text{th}}_{B}
    \end{pmatrix} \succeq 0
\end{equation}
where $\textbf{Q}_{G}$, $\textbf{Q}_{G}^{T}$ are the strength of the bilinear coupling and $\textbf{Q}^{\text{th}}_{A},\textbf{Q}^{\text{th}}_{B}$ are the white noise spectra of the uncorrelated external baths.
\end{theorem}

The full proof is given in Section `Full GKSL Proof for General Interaction' of the Supplemental Material~\cite{SM}. The key step is to rotate the coupling matrix, $\widetilde{\bm Q}_G \equiv (\bm Q^{\rm th}_A)^{-1/2}\bm Q_G(\bm Q^{\rm th}_B)^{-1/2}$, and apply an SVD. In the rotated basis the corresponding Born-Markov master equation decouples into a sum of rank-1 equations of the same form as \eqref{rank 1 BM}, each implementable by the symmetric-feedback LOCC scheme of the previous section. A Lie--Trotter decomposition then assembles the full evolution as a product of LOCC channels, which preserves separability.

\textit{Application---}One instance where our result can be applied is in \cite{miki2026} where an identical sufficient separability condition for the case of two-mode systems was derived. This implies that neither increasing the number of modes nor allowing for non-Gaussianity could relax this condition.

\begin{figure}[t]
    \centering
    \includegraphics[width=0.63\linewidth]{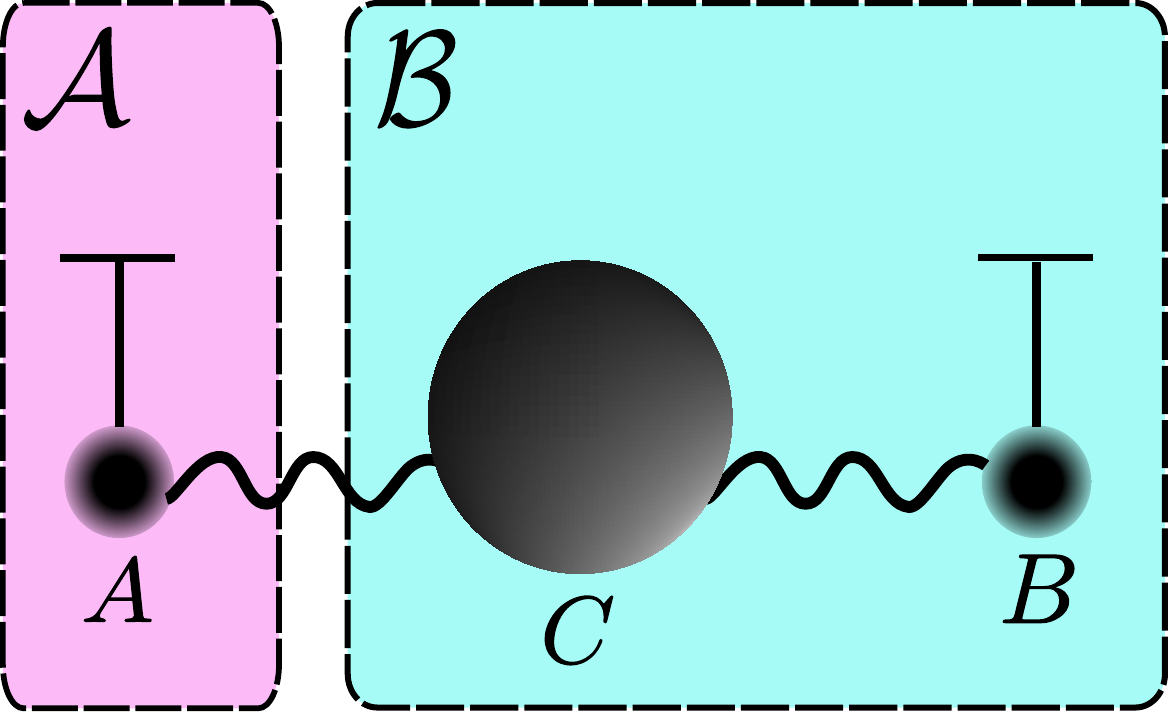}
    \caption{A tripartite system involving entanglement generation between subsystems $A$ and $B$ mediated by a massive mediator $C$. By combining subsystems $B$ and $C$ into a composite subsystem, we have reduced the problem to one of bipartite entanglement: as long as 
    thermal noise levels of $A$ and $C$ are strong enough
    there can never be any entanglement generated between subsystems $A$ and $B$.}
    \label{fig: composite subsystem}
\end{figure}

Another example is to \cite{christopher2025}, which  investigated entanglement generation between two identical harmonic oscillators (labeled systems $A$ and $B$) mediated via a third tunable oscillator (labeled system $C$) located between $A$ and $B$. This system is visualized in Fig.~\ref{fig: composite subsystem}. Suppose we start with an initial separable state $\rho(0) = \rho_{A}(0)\otimes \rho_{CB}(0)$ where $\rho_{CB}$ describes the composite system of the mediator (system $C$) and one of the oscillators (system $B$) and may or may not be entangled. Since the system is entirely Gaussian, we can directly apply our results to claim that in the presence of thermal noise, if Theorem \ref{pert-theorem} is satisfied, then $\rho(t)$ remains separable, thus implying that $\rho_{A}(t)$ and $\rho_{B}(t)$ remain separable as well. The choice of setting the bipartition between system $A$ and $C$ was arbitrary and the same result would hold if one instead chose $\rho(0) = \rho_{AC}(0)\otimes\rho_{B}(0)$. 
Once a bipartition has been selected, we may treat the composite system (say $CB$) as a single subsystem whose internal dynamics we are not interested in because they do not generate entanglement across the bipartition.

To make this explicit, take the bipartition $A\,|\,CB$ in the regime where $A$ is well separated from $B$, so that the only relevant cross-bipartition coupling is the $A$--$C$ gravitational interaction. With masses $M_A$, $M_C$, separation $d_{AC}$, and mechanical dissipation rates $\gamma_A$, $\gamma_C$, the entanglement-generation condition obtained from Theorem~\ref{pert-theorem} reads
\begin{equation}
\label{ACB criterion}
    {\hbar G M_A M_C}/{d_{AC}^3} \ge \sqrt{M_A M_C}\,\sqrt{\gamma_A\gamma_C}\,k_B T.
\end{equation}
If $C$ is a uniform spherical mediator of density $\rho_C$, $d_{AC}$ is set by the mediator radius and \eqref{ACB criterion} specializes to
\begin{equation}
  ( {4\pi \hbar G\rho_C}/{3})\sqrt{{M_A}/{M_C}} \ge \sqrt{\gamma_A\gamma_C}\,k_B T.
\end{equation}
Crucially, while a massive mediator might amplify the gravitational phase accumulated between $A$ and $C$ — and thereby enhance the magnitude of entanglement once it is generated — the requirement on thermal noise is not mitigated: the threshold remains set by the same competition between gravity and thermal decoherence, while numerically, having a massive mediator $C$ does not bring an obvious advantage unless it also has a much smaller dissipation rate than $A$.

%

Finally, our results apply to non-Gaussian input states, including
the mixed discrete-continuous variable system coupling in \cite{pedernales2022} whereby a harmonic oscillator (denoted system $C$) mediates the entanglement generation between a two-level test mass (denoted system $A$) and an ancilla qubit (denoted system $B$). As in the tripartite harmonic oscillator system, we can view this as a bipartite system by grouping $B$ and $C$ together and starting with the initial state $\rho(0) = \rho_{A}(0)\otimes\rho_{CB}(0)$. The interaction between system $A$ and system $C$ is described by a Hamiltonian of the form $H_{\text{int}} = \hat{\sigma}^{A}_{z}(\hat{a}+\hat{a}^{\dagger})$ where $\hat{\sigma}^{A}_{z}$ is the Pauli Z operator acting on the two-level test mass and $\hat{a},\hat{a}^{\dagger}$ are the canonical bosonic operators acting on the mediator. In the language of Theorem \ref{rank 1 theorem}, this is equivalent to replacing $\hat{\bm{\xi}}^{T}_{A}\textbf{u}_{A}$ by $\hat{\sigma}^{a}_{z}$ and $\hat{\bm{\xi}}^{T}_{B}\textbf{u}_{B}$ by $\hat{a}+\hat{a}^{\dagger}$. Since no part of the theorem requires a specification of the bilinear coupling, this identification is valid and thus the same separability condition in the presence of thermal noise holds.

\textit{Conclusion---} We derived a sufficient separability condition for general multimode systems coupled under a bilinear interaction and subject to white thermal noise, thus extending the results of \cite{kafri2013,kafri2014}. We first treated purely Gaussian systems and derived a bound identical to the two-mode result in \cite{miki2026}, demonstrating that increasing the number of modes does not relax the necessary entanglement criterion. Though it is the least general bound in this paper, its proof is the most explicit in showing how each component of the time-evolved system contributes to the separability condition, thus identifies the specific conditions under which the result could be boosted to a necessary and sufficient one. Additionally, to our knowledge, this is the first application of this specific separability condition to the study of entanglement in gravitating systems, thus making the proof technique of independent interest for future studies of multimode Gaussian systems. Using the techniques of \cite{diosi2017}, we 
constructed an LOCC scheme to reproduce the dynamics of the coupled systems and thus deriveed separability conditions for general multi-mode systems. Comparing with the Gaussian result shows that even allowing for non-Gaussianity does not relax the bound, further underscoring the fundamental competition between thermal noise and entanglement generation. Finally, this LOCC picture affords us full generality in applying our results to various gravity induced entanglement experiments, demonstrating the same competition between gravity and thermal noise exists ubiquitously. 
Because our treatment of the noise term is quite general, the LOCC constructions used throughout the proofs may be useful in serving as building blocks for identifying entanglement-generation thresholds in other open-system dynamics.

\textit{Acknowledgement---}D.M. is supported by the JSPS Overseas Research Fellowships and the Keck Foundation.  Y.C.\ and A.L.\ are supported by the Simons Foundation (Award Number 568762). 

\bibliography{reference}

@article{bose2017,
  title = {Spin Entanglement Witness for Quantum Gravity},
  author = {Bose, Sougato and Mazumdar, Anupam and Morley, Gavin W. and Ulbricht, Hendrik and Toro\ifmmode \check{s}\else \v{s}\fi{}, Marko and Paternostro, Mauro and Geraci, Andrew A. and Barker, Peter F. and Kim, M. S. and Milburn, Gerard},
  journal = {Phys. Rev. Lett.},
  volume = {119},
  issue = {24},
  pages = {240401},
  numpages = {6},
  year = {2017},
  month = {Dec},
  publisher = {American Physical Society},
  doi = {10.1103/PhysRevLett.119.240401},
  url = {https://link.aps.org/doi/10.1103/PhysRevLett.119.240401}
}

@article{marletto2017,
    author = "Marletto, Chiara and Vedral, Vlatko",
    title = "{Gravitationally-induced entanglement between two massive particles is sufficient evidence of quantum effects in gravity}",
    doi = "10.1103/PhysRevLett.119.240402",
    journal = "Phys. Rev. Lett.",
    volume = "119",
    number = "24",
    pages = "240402",
    year = "2017"
}

@article{marletto2025,
  title = {Quantum-information methods for quantum gravity laboratory-based tests},
  author = {Marletto, Chiara and Vedral, Vlatko},
  journal = {Rev. Mod. Phys.},
  volume = {97},
  issue = {1},
  pages = {015006},
  numpages = {29},
  year = {2025},
  month = {Mar},
  publisher = {American Physical Society},
  doi = {10.1103/RevModPhys.97.015006},
  url = {https://link.aps.org/doi/10.1103/RevModPhys.97.015006}
}

@article{miao2020,
  title = {Quantum correlations of light mediated by gravity},
  author = {Miao, Haixing and Martynov, Denis and Yang, Huan and Datta, Animesh},
  journal = {Phys. Rev. A},
  volume = {101},
  issue = {6},
  pages = {063804},
  numpages = {7},
  year = {2020},
  month = {Jun},
  publisher = {American Physical Society},
  doi = {10.1103/PhysRevA.101.063804},
  url = {https://link.aps.org/doi/10.1103/PhysRevA.101.063804}
}

@article{pedernales2022,
  title = {Enhancing Gravitational Interaction between Quantum Systems by a Massive Mediator},
  author = {Pedernales, Julen S. and Streltsov, Kirill and Plenio, Martin B.},
  journal = {Phys. Rev. Lett.},
  volume = {128},
  issue = {11},
  pages = {110401},
  numpages = {6},
  year = {2022},
  month = {Mar},
  publisher = {American Physical Society},
  doi = {10.1103/PhysRevLett.128.110401},
  url = {https://link.aps.org/doi/10.1103/PhysRevLett.128.110401}
}

@article{kaku2023,
  title = {Enhancement of quantum gravity signal in an optomechanical experiment},
  author = {Kaku, Youka and Fujita, Tomohiro and Matsumura, Akira},
  journal = {Phys. Rev. D},
  volume = {108},
  issue = {10},
  pages = {106014},
  numpages = {18},
  year = {2023},
  month = {Nov},
  publisher = {American Physical Society},
  doi = {10.1103/PhysRevD.108.106014},
  url = {https://link.aps.org/doi/10.1103/PhysRevD.108.106014}
}

@article{fujita2025,
  doi = {10.1088/1361-6382/adf0bb},
  url = {https://doi.org/10.1088/1361-6382/adf0bb},
  year = {2025},
  month = {aug},
  publisher = {IOP Publishing},
  volume = {42},
  number = {16},
  pages = {165003},
  author = {Fujita, Tomohiro and Kaku, Youka and Matsumura, Akira and Michimura, Yuta},
  title = {Inverted oscillators for testing gravity-induced quantum entanglement},
  journal = {Classical and Quantum Gravity}
}

@article{hatakeyama2026,
  title = {Theoretical study of the squeezed-light-enhanced sensitivity to gravity-induced entanglement via finite-time analysis},
  author = {Hatakeyama, Kosei and Miki, Daisuke and Yamamoto, Kazuhiro},
  journal = {Phys. Rev. D},
  volume = {113},
  issue = {2},
  pages = {024025},
  numpages = {14},
  year = {2026},
  month = {Jan},
  publisher = {American Physical Society},
  doi = {10.1103/1mfv-y24t},
  url = {https://link.aps.org/doi/10.1103/1mfv-y24t}
}

@article{fukuzumi2026,
  title = {Momentum squeezed state realized via optimal filtering in optomechanics: Implications for gravity-induced entanglement},
  author = {Fukuzumi, Ryotaro and Hatakeyama, Kosei and Miki, Daisuke and Yamamoto, Kazuhiro},
  journal = {Phys. Rev. Res.},
  volume = {8},
  issue = {2},
  pages = {023039},
  numpages = {10},
  year = {2026},
  month = {Apr},
  publisher = {American Physical Society},
  doi = {10.1103/zrs2-sk28},
  url = {https://link.aps.org/doi/10.1103/zrs2-sk28}
}

@misc{shiomatsu2025,
      title={Boosting Gravity-Induced Entanglement through Parametric Resonance}, 
      author={Yuka Shiomatsu and Youka Kaku and Akira Matsumura and Tomohiro Fujita},
      year={2025},
      eprint={2511.09169},
      archivePrefix={arXiv},
      primaryClass={gr-qc},
      url={https://arxiv.org/abs/2511.09169}, 
}

@misc{miki2026,
      title={Amplification and generation bounds of gravity-induced entanglement in pulsed optomechanical systems}, 
      author={Daisuke Miki and Alfred Li and Yanbei Chen},
      year={2026},
      eprint={in preparation},
      archivePrefix={arXiv},
      primaryClass={quant-ph}
}

@article{christopher2025,
  title = {Beyond entanglement: Diagnosing quantum mediator dynamics in gravitationally mediated experiments},
  author = {Christopher, P. George and Shankaranarayanan, S.},
  journal = {Phys. Rev. D},
  volume = {112},
  issue = {8},
  pages = {L081502},
  numpages = {7},
  year = {2025},
  month = {Oct},
  publisher = {American Physical Society},
  doi = {10.1103/njws-rs11},
  url = {https://link.aps.org/doi/10.1103/njws-rs11}
}

@article{peres1996,
  title = {Separability Criterion for Density Matrices},
  author = {Peres, Asher},
  journal = {Phys. Rev. Lett.},
  volume = {77},
  issue = {8},
  pages = {1413--1415},
  numpages = {0},
  year = {1996},
  month = {Aug},
  publisher = {American Physical Society},
  doi = {10.1103/PhysRevLett.77.1413},
  url = {https://link.aps.org/doi/10.1103/PhysRevLett.77.1413}
}

@article{horodecki1996,
  title = {Separability of mixed states: necessary and sufficient conditions},
  journal = {Physics Letters A},
  volume = {223},
  number = {1},
  pages = {1-8},
  year = {1996},
  issn = {0375-9601},
  doi = {https://doi.org/10.1016/S0375-9601(96)00706-2},
  url = {https://www.sciencedirect.com/science/article/pii/S0375960196007062},
  author = {Michał Horodecki and Paweł Horodecki and Ryszard Horodecki}
}

@article{duan2000,
  title = {Inseparability Criterion for Continuous Variable Systems},
  author = {Duan, Lu-Ming and Giedke, G. and Cirac, J. I. and Zoller, P.},
  journal = {Phys. Rev. Lett.},
  volume = {84},
  issue = {12},
  pages = {2722--2725},
  numpages = {0},
  year = {2000},
  month = {Mar},
  publisher = {American Physical Society},
  doi = {10.1103/PhysRevLett.84.2722},
  url = {https://link.aps.org/doi/10.1103/PhysRevLett.84.2722}
}

@article{simon2000,
  title = {Peres-Horodecki Separability Criterion for Continuous Variable Systems},
  author = {Simon, R.},
  journal = {Phys. Rev. Lett.},
  volume = {84},
  issue = {12},
  pages = {2726--2729},
  numpages = {0},
  year = {2000},
  month = {Mar},
  publisher = {American Physical Society},
  doi = {10.1103/PhysRevLett.84.2726},
  url = {https://link.aps.org/doi/10.1103/PhysRevLett.84.2726}
}

@book{serafini2017,
  author = {A. Serafini},
  title = "{Quantum Continuous Variables}",
  publisher = {CRC Press},
  year = {2017}
}

@article{diosi2017,
  author = {Di\'{o}si, Lajos and Tilloy, Antoine},
  title = {On GKLS Dynamics for Local Operations and Classical Communication},
  journal = {Open Systems \& Information Dynamics},
  volume = {24},
  number = {04},
  pages = {1740020},
  year = {2017},
  doi = {10.1142/S1230161217400200},
  URL = {https://doi.org/10.1142/S1230161217400200}
}

@misc{kafri2013,
      title={A noise inequality for classical forces}, 
      author={Dvir Kafri and J.M. Taylor},
      year={2013},
      eprint={1311.4558}
}

@article{kafri2014,
  doi = {10.1088/1367-2630/16/6/065020},
  url = {https://doi.org/10.1088/1367-2630/16/6/065020},
  year = {2014},
  month = {jun},
  publisher = {IOP Publishing},
  volume = {16},
  number = {6},
  pages = {065020},
  author = {Kafri, D and Taylor, J M and Milburn, G J},
  title = {A classical channel model for gravitational decoherence},
  journal = {New Journal of Physics}
}

@article{mari2016,
  title={Experiments testing macroscopic quantum superpositions must be slow},
  author={Mari, A. and De Palma, G. and Giovannetti, V.},
  journal={Scientific Reports},
  volume={6},
  pages={22777},
  year={2016}
}

@article{belenchia2018,
  title = {Quantum superposition of massive objects and the quantization of gravity},
  author = {Belenchia, Alessio and Wald, Robert M. and Giacomini, Flaminia and Castro-Ruiz, Esteban and Brukner, \ifmmode \check{C}\else \v{C}\fi{}aslav and Aspelmeyer, Markus},
  journal = {Phys. Rev. D},
  volume = {98},
  issue = {12},
  pages = {126009},
  numpages = {9},
  year = {2018},
  month = {Dec},
  publisher = {American Physical Society},
  doi = {10.1103/PhysRevD.98.126009},
  url = {https://link.aps.org/doi/10.1103/PhysRevD.98.126009}
}

@article{danielson2022,
  title = {Gravitationally mediated entanglement: Newtonian field versus gravitons},
  author = {Danielson, Daine L. and Satishchandran, Gautam and Wald, Robert M.},
  journal = {Phys. Rev. D},
  volume = {105},
  issue = {8},
  pages = {086001},
  numpages = {11},
  year = {2022},
  month = {Apr},
  publisher = {American Physical Society},
  doi = {10.1103/PhysRevD.105.086001},
  url = {https://link.aps.org/doi/10.1103/PhysRevD.105.086001}
}

@article{sugiyama2023,
  title = {Quantum uncertainty of gravitational field and entanglement in superposed massive particles},
  author = {Sugiyama, Yuuki and Matsumura, Akira and Yamamoto, Kazuhiro},
  journal = {Phys. Rev. D},
  volume = {108},
  issue = {10},
  pages = {105019},
  numpages = {13},
  year = {2023},
  month = {Nov},
  publisher = {American Physical Society},
  doi = {10.1103/PhysRevD.108.105019},
  url = {https://link.aps.org/doi/10.1103/PhysRevD.108.105019}
}

@article{sugiyama2024,
  title = {Quantumness of the gravitational field: A perspective on monogamy relation},
  author = {Sugiyama, Yuuki and Matsumura, Akira and Yamamoto, Kazuhiro},
  journal = {Phys. Rev. D},
  volume = {110},
  issue = {4},
  pages = {045016},
  numpages = {21},
  year = {2024},
  month = {Aug},
  publisher = {American Physical Society},
  doi = {10.1103/PhysRevD.110.045016},
  url = {https://link.aps.org/doi/10.1103/PhysRevD.110.045016}
}

@article{carney2022,
  title = {Newton, entanglement, and the graviton},
  author = {Carney, Daniel},
  journal = {Phys. Rev. D},
  volume = {105},
  issue = {2},
  pages = {024029},
  numpages = {17},
  year = {2022},
  month = {Jan},
  publisher = {American Physical Society},
  doi = {10.1103/PhysRevD.105.024029},
  url = {https://link.aps.org/doi/10.1103/PhysRevD.105.024029}
}

@book{schlosshauer2007decoherence,
  title={Decoherence: And the Quantum-To-Classical Transition},
  author={Schlosshauer, M.},
  isbn={9783540357735},
  lccn={2007930038},
  series={The Frontiers Collection},
  url={https://books.google.com/books?id=1qrJUS5zNbEC},
  year={2007},
  publisher={Springer}
}

@article{H_nggi_2005,
   title={Fundamental aspects of quantum Brownian motion},
   volume={15},
   ISSN={1089-7682},
   url={http://dx.doi.org/10.1063/1.1853631},
   DOI={10.1063/1.1853631},
   number={2},
   journal={Chaos: An Interdisciplinary Journal of Nonlinear Science},
   publisher={AIP Publishing},
   author={Hänggi, Peter and Ingold, Gert-Ludwig},
   year={2005},
   month=jun }

@article{Ghosh_2024,
   title={Independent-oscillator model and the quantum Langevin equation for an oscillator: a review},
   volume={2024},
   ISSN={1742-5468},
   url={http://dx.doi.org/10.1088/1742-5468/ad5711},
   DOI={10.1088/1742-5468/ad5711},
   number={7},
   journal={Journal of Statistical Mechanics: Theory and Experiment},
   publisher={IOP Publishing},
   author={Ghosh, Aritra and Bandyopadhyay, Malay and Dattagupta, Sushanta and Gupta, Shamik},
   year={2024},
   month=jul, pages={074002} }

@book{horn2012matrix,
  author    = {Roger A. Horn and Charles R. Johnson},
  title     = {Matrix Analysis},
  edition   = {2},
  year      = {2012},
  publisher = {Cambridge University Press},
  address   = {Cambridge},
  isbn      = {978-0-521-83940-2},
  doi       = {10.1017/CBO9781139020411}
}

@misc{SM,
  note = {See Supplemental Material for additional details about open system dynamics and full proof of universal bound.}
}

@misc{fabiano2026,
      title={Minimal noise in non-quantized gravity}, 
      author={Giuseppe Fabiano and Tomohiro Fujita and Akira Matsumura and Daniel Carney},
      year={2026},
      eprint={2603.26075}
}

@article{bose2025rmp,
  title = {Massive quantum systems as interfaces of quantum mechanics and gravity},
  author = {Bose, Sougato and Fuentes, Ivette and Geraci, Andrew A. and Khan, Saba Mehsar and Qvarfort, Sofia and Rademacher, Markus and Rashid, Muddassar and Toro{\v{s}}, Marko and Ulbricht, Hendrik and Wanjura, Clara C.},
  journal = {Rev. Mod. Phys.},
  volume = {97},
  issue = {1},
  pages = {015003},
  year = {2025},
  month = {Feb},
  publisher = {American Physical Society},
  doi = {10.1103/RevModPhys.97.015003},
  url = {https://link.aps.org/doi/10.1103/RevModPhys.97.015003}
}

@article{tang2025formfactor,
  title = {Optimal Form Factors for Experimental Proposals on Gravity-Induced Entanglement},
  author = {Tang, Ziqian and Xue, Hanyu and Han, Zizhao and Kan, Zikuan and Li, Zeji and Liu, Yulong},
  journal = {Phys. Rev. D},
  volume = {112},
  pages = {042004},
  year = {2025},
  doi = {10.1103/jznw-q3q8}
}

@article{miki2024feasible,
  title = {Feasible generation of gravity-induced entanglement by using optomechanical systems},
  author = {Miki, Daisuke and Matsumura, Akira and Yamamoto, Kazuhiro},
  journal = {Phys. Rev. D},
  volume = {110},
  pages = {024057},
  year = {2024},
  doi = {10.1103/PhysRevD.110.024057}
}

@article{datta2021signatures,
  title = {Signatures of the quantum nature of gravity in the differential motion of two masses},
  author = {Datta, Animesh and Miao, Haixing},
  journal = {Quantum Science and Technology},
  volume = {6},
  pages = {045014},
  year = {2021},
  doi = {10.1088/2058-9565/ac1adf}
}

@article{plato2023enhanced,
  title = {Enhanced Gravitational Entanglement via Modulated Optomechanics},
  author = {Plato, A. Douglas K. and R{\"a}tzel, Dennis and Wan, Chuanqi},
  journal = {Quantum},
  volume = {7},
  pages = {1177},
  year = {2023},
  doi = {10.22331/q-2023-11-08-1177}
}

@article{schut2022resilience,
  title = {Improving resilience of quantum-gravity-induced entanglement of masses to decoherence using three superpositions},
  author = {Schut, Martine and Tilly, Jules and Marshman, Ryan J. and Bose, Sougato and Mazumdar, Anupam},
  journal = {Phys. Rev. A},
  volume = {105},
  pages = {032411},
  year = {2022},
  doi = {10.1103/PhysRevA.105.032411}
}

@article{tilly2021qudits,
  title = {Qudits for Witnessing Quantum Gravity Induced Entanglement of Masses Under Decoherence},
  author = {Tilly, Jules and Marshman, Ryan J. and Mazumdar, Anupam and Bose, Sougato},
  journal = {Phys. Rev. A},
  volume = {104},
  pages = {052416},
  year = {2021},
  doi = {10.1103/PhysRevA.104.052416}
}

@article{rijavec2021decoherence,
  title = {Decoherence effects in non-classicality tests of gravity},
  author = {Rijavec, Simone and Carlesso, Matteo and Bassi, Angelo and Vedral, Vlatko and Marletto, Chiara},
  journal = {New Journal of Physics},
  volume = {23},
  pages = {043040},
  year = {2021},
  doi = {10.1088/1367-2630/abf3eb}
}

@article{krisnanda2020observable,
  title = {Observable quantum entanglement due to gravity},
  author = {Krisnanda, Tanjung and Tham, Guo Yao and Paternostro, Mauro and Paterek, Tomasz},
  journal = {npj Quantum Information},
  volume = {6},
  pages = {12},
  year = {2020},
  doi = {10.1038/s41534-020-0243-y},
}

@article{qvarfort2020mesoscopic,
  title = {Mesoscopic entanglement through central-potential interactions},
  author = {Qvarfort, Sofia and Bose, Sougato and Serafini, Alessio},
  journal = {Journal of Physics B: Atomic, Molecular and Optical Physics},
  volume = {53},
  number = {23},
  pages = {235501},
  year = {2020},
  doi = {10.1088/1361-6455/abbe8d}
}

@article{schmole2016micromechanical,
  title = {A micromechanical proof-of-principle experiment for measuring the gravitational force of milligram masses},
  author = {Schm{\"o}le, Jonas and Dragosits, Mathias and Hepach, Hans and Aspelmeyer, Markus},
  journal = {Classical and Quantum Gravity},
  volume = {33},
  number = {12},
  pages = {125031},
  year = {2016},
  doi = {10.1088/0264-9381/33/12/125031}
}

@article{matsumoto2019mg,
  title = {Demonstration of Displacement Sensing of a mg-Scale Pendulum for mm- and mg-Scale Gravity Measurements},
  author = {Matsumoto, Nobuyuki and Sugawara, Masakazu and Suzuki, Seiya and Abe, Naofumi and Komori, Kentaro and Michimura, Yuta and Aso, Yoichi and Cata{\~n}o-Lopez, Seth B. and Edamatsu, Keiichi},
  journal = {Phys. Rev. Lett.},
  volume = {122},
  pages = {071101},
  year = {2019},
  doi = {10.1103/PhysRevLett.122.071101}
}

@article{catano2020highq,
  title = {High-Q Milligram-Scale Monolithic Pendulum for Quantum-Limited Gravity Measurements},
  author = {Cata{\~n}o-Lopez, Seth B. and Santiago-Condori, Jordy G. and Edamatsu, Keiichi and Matsumoto, Nobuyuki},
  journal = {Phys. Rev. Lett.},
  volume = {124},
  pages = {221102},
  year = {2020},
  doi = {10.1103/PhysRevLett.124.221102}
}

@article{komori2020atto,
  title = {Attonewton-meter torque sensing with a macroscopic optomechanical torsion pendulum},
  author = {Komori, Kentaro and Enomoto, Yutaro and Ooi, Ching Pin and Miyazaki, Yuki and Matsumoto, Nobuyuki and Sudhir, Vivishek and Michimura, Yuta and Ando, Masaki},
  journal = {Phys. Rev. A},
  volume = {101},
  pages = {011802},
  year = {2020},
  doi = {10.1103/PhysRevA.101.011802}
}

@article{agafonova2026milligram,
  title = {One-milligram torsional pendulum toward experiments at the quantum-gravity interface},
  author = {Agafonova, Sofia and Rossell{\'o}, Pere and Mekonnen, Manuel and Hosten, Onur},
  journal = {Communications Physics},
  volume = {9},
  pages = {80},
  year = {2026},
  doi = {10.1038/s42005-026-02514-w}
}

@article{vinante2020ultralow,
  title = {Ultralow mechanical damping with Meissner-levitated ferromagnetic microparticles},
  author = {Vinante, A. and Falferi, P. and Gasbarri, G. and Setter, A. and Timberlake, C. and Ulbricht, H.},
  journal = {Phys. Rev. Applied},
  volume = {13},
  pages = {064027},
  year = {2020},
  doi = {10.1103/PhysRevApplied.13.064027}
}

@article{hofer2023highq,
  title = {High-Q Magnetic Levitation and Control of Superconducting Microspheres at Millikelvin Temperatures},
  author = {Hofer, Joachim and Gross, Rudolf and Higgins, Gerard and Huebl, Hans and Kieler, Oliver F. and Kleiner, Reinhold and Koelle, Dieter and Schmidt, Philip and Slater, Joshua A. and Trupke, Michael and Uhl, Kevin and Weimann, Thomas and Wieczorek, Witlef and Aspelmeyer, Markus},
  journal = {Phys. Rev. Lett.},
  volume = {131},
  pages = {043603},
  year = {2023},
  doi = {10.1103/PhysRevLett.131.043603}
}

@article{fuchs2024measuring,
  title = {Measuring gravity with milligram levitated masses},
  author = {Fuchs, Tim M. and Uitenbroek, Dennis G. and Plugge, Jaimy and van Halteren, Noud and van Soest, Jean-Paul and Vinante, Andrea and Ulbricht, Hendrik and Oosterkamp, Tjerk H.},
  journal = {Science Advances},
  volume = {10},
  pages = {eadk2949},
  year = {2024},
  doi = {10.1126/sciadv.adk2949}
}

@article{timberlake2021modified,
  title = {Probing modified gravity with magnetically levitated resonators},
  author = {Timberlake, Chris and Vinante, Andrea and Shankar, Francesco and Lapi, Andrea and Ulbricht, Hendrik},
  journal = {Phys. Rev. D},
  volume = {104},
  pages = {L101101},
  year = {2021},
  doi = {10.1103/PhysRevD.104.L101101}
}

@article{delic2020cooling,
  title = {Cooling of a levitated nanoparticle to the motional quantum ground state},
  author = {Deli{\'c}, Uro{\v{s}} and Reisenbauer, Manuel and Dare, Kahan and Grass, David and Vuleti{\'c}, Vladan and Kiesel, Nikolai and Aspelmeyer, Markus},
  journal = {Science},
  volume = {367},
  pages = {892--895},
  year = {2020},
  doi = {10.1126/science.aba3993}
}

@article{tebbenjohanns2021quantum,
  title = {Quantum control of a nanoparticle optically levitated in cryogenic free space},
  author = {Tebbenjohanns, Felix and Mattana, M. Luisa and Rossi, Massimiliano and Frimmer, Martin and Novotny, Lukas},
  journal = {Nature},
  volume = {595},
  pages = {378--382},
  year = {2021},
  doi = {10.1038/s41586-021-03617-w}
}

@article{piotrowski2023simultaneous,
  title = {Simultaneous ground-state cooling of two mechanical modes of a levitated nanoparticle},
  author = {Piotrowski, Johannes and Windey, Dominik and Vijayan, Jayadev and Gonzalez-Ballestero, Carlos and de los R{\'i}os Sommer, Andr{\'e}s and Meyer, Nadine and Quidant, Romain and Romero-Isart, Oriol and Reimann, Ren{\'e} and Novotny, Lukas},
  journal = {Nature Physics},
  volume = {19},
  pages = {1009--1013},
  year = {2023},
  doi = {10.1038/s41567-023-01956-1}
}

@article{mari2025,
  title = {Can gravity mediate the transmission of quantum information?},
  author = {Mari, Andrea and Zippilli, Stefano and Vitali, David},
  journal = {Phys. Rev. D},
  volume = {113},
  issue = {2},
  pages = {L021905},
  numpages = {8},
  year = {2026},
  month = {Jan},
  publisher = {American Physical Society},
  doi = {10.1103/pfvz-fd54},
  url = {https://link.aps.org/doi/10.1103/pfvz-fd54}
}

@article{Lami2024,
  title = {Testing the Quantumness of Gravity without Entanglement},
  author = {Lami, Ludovico and Pedernales, Julen S. and Plenio, Martin B.},
  journal = {Phys. Rev. X},
  volume = {14},
  issue = {2},
  pages = {021022},
  numpages = {47},
  year = {2024},
  month = {May},
  publisher = {American Physical Society},
  doi = {10.1103/PhysRevX.14.021022},
  url = {https://link.aps.org/doi/10.1103/PhysRevX.14.021022}
}

@article{Oshima:2023Du,
  author = "Oshima, Yuka  and  Takano, Satoru  and  Ooi, Ching Pin  and  Choi, Minseo  and  Cao, Mengdi  and  Michimura, Yuta  and  Komori, Kentaro  and  Ando, Masaki",
  title = "{Development of Torsion-Bar Antenna for Low-Frequency Gravitational-Wave Observation}",
  doi = "10.22323/1.444.1584",
  journal = "PoS",
  year = 2023,
  volume = "ICRC2023",
  pages = "1584"
}

\clearpage
\onecolumngrid
\raggedbottom
\setcounter{equation}{0}
\renewcommand{\theequation}{S\arabic{equation}}

\begin{center}
    {\large\bfseries Supplementary Material: Universal Bound for Entanglement Generation}
\end{center}

\vspace{1em}

\appendix
\section{I. Discussion on $\hat{F}_{H}(t)$ vs $\hat{F}_{I}(t)$}
\label{backaction appendix}
The claim is that:
\begin{equation}
    \hat{F}_{H}(t) = \hat{F}_{I}(t) + (\text{second-order})
\end{equation}
where as a reminder, the subscript $H$ represents Heisenberg picture and $I$ represents the interaction picture where we time-evolve operators with system-only and bath-only Hamiltonians. Here we discuss what it physically means to drop the second-order term in our analysis. To see this most clearly, consider the Caldeira-Leggett Model \cite{H_nggi_2005}:
\begin{align}
   \hat{H}_{\text{tot}} &= \hat{H}_{\text{sys}}+\hat{H}_{\text{bath}}+\hat{H}_{\text{sys-bath}},\\
        \hat{H}_{\text{sys}} &= \frac{\hat{p}^{2}}{2M}+V(\hat{q},t),\\
       \hat{H}_{\text{bath}} &= \sum_{j}\left(\frac{\hat{p}_{j}^2}{2m_{j}}+\frac{1}{2}m_{j}\omega_{j}^2\hat{q}_{j}^2\right),\\
    \hat{H}_{\text{sys-bath}} &= \hat{q}\sum_{j}c_{j}\hat{q}_{j} \equiv \hat{q}\hat{F}
\end{align}
where we have ignored the counterterm for now in lieu of the fact that it can always be absorbed into the system $V(\hat{q})$ term to define a new renormalized frequency. Additionally, we abuse notation here with the understanding that no subscripts implies system operator and operators with subscripts are bath operators. For each bath mode, we have from the Heisenberg equation of motion:
\begin{equation}
    \begin{split}
        \dot{\hat{q}}_{j,H} &= \frac{\hat{p}_{j,H}(t)}{m_{j}},\\
        \dot{\hat{p}}_{j,H} &= -m_{j}\omega_{j}^{2}q_{j,H}(t)-\hat{q}_{H}(t)c_{j}.
    \end{split}
\end{equation}
Differentiating the first one with respect to time again and plugging in the second equation gives:
\begin{equation}
    \ddot{\hat{q}}_{j,H}(t)+\omega_{j}^{2}\hat{q}_{j,H}(t)=-\frac{c_{j}}{m_{j}}\hat{q}_{H}(t).
\end{equation}
This is a driven oscillator that can be solved using Green's function:
\begin{equation}
    \hat{q}_{j,H}(t) = \left[\hat{q}_{j}(0)\cos(\omega_{j}t)+\frac{\hat{p}_{j}(0)}{m_{j}\omega_{j}}\sin(\omega_{j}t)\right]-\frac{c_{j}}{m_{j}\omega_{j}}\int_{0}^{t}\sin(\omega_{j}(t-s))\hat{q}_{H}(s)ds.
\end{equation}
Now notice that the terms in bracket correspond exactly to the Heisenberg equation of motion corresponding to evolution under the bath Hamiltonian only. In terms of our force operator, this gives:
\begin{equation}
\begin{split}
    \hat{F}_{H}(t) &= \hat{h}(t)-\int_{0}^{t}\mu(t-s)\hat{q}_{H}(s)ds, \\
    \mu(t) &\equiv \sum_{j}\frac{c_{j}^2}{m_{j}\omega_{j}}\sin(\omega_{j}t),\\
     \hat{h}(t) &\equiv \sum_{j}c_{j} \left[\hat{q}_{j}(0)\cos(\omega_{j}t)+\frac{\hat{p}_{j}(0)}{m_{j}\omega_{j}}\sin(\omega_{j}t)\right].
\end{split}
\end{equation}
Notice that $\hat{h}(t) = \hat{F}_{I}(t)$ and $\mu(t)$ encodes the effects of the system operator on the bath operator which we refer to as the second-order kernel.

Now consider the Heisenberg equations for the system operators:
\begin{equation}
    \begin{split}
        \dot{\hat{q}}_{H}&=\frac{\hat{p}_{H}(t)}{M},\\
        \dot{\hat{p}}_{H} &= -\frac{dV(\hat{q}_{H}(t),t)}{d\hat{q}}-\hat{F}_{H}(t).
    \end{split}
\end{equation}
Plugging in everything gives
\begin{equation}
    M\ddot{\hat{q}}_{H}(t)+\frac{dV(\hat{q}_{H}(t),t)}{d\hat{q}}=-\hat{h}(t)+\int_{0}^{t}ds \ \mu(t-s)\hat{q}_{H}(s).
\end{equation}
Now define the damping kernel:
\begin{equation}
    \gamma(t) \equiv \frac{1}{M}\sum_{j}\frac{c_{j}^2}{m_{j}\omega_{j}^2}\cos(\omega_{j}t)    
\end{equation}
and notice that $\mu(t) = -M\dot{\gamma}(t)$. We then have:
\begin{equation}
    \int_{0}^{t}ds \ \mu(t-s)\hat{q}_{H}(s) = M\int_{0}^{t}ds\frac{d}{ds}\gamma(t-s)\hat{q}_{H}(s)=M\gamma(0)\hat{q}_{H}(t)-M\gamma(t)\hat{q}_{H}(0)-M\int_{0}^{t}ds\gamma(t-s)\dot{\hat{q}}_{H}(s).
\end{equation}
Thus, the full Langevin equation becomes:
\begin{equation}
    M\ddot{\hat{q}}_{H}(t) + \frac{dV(\hat{q}_{H}(t),t)}{d\hat{q}}-M\gamma(0)\hat{q}_{H}(t)+M\int_{0}^{t}ds\gamma(t-s)\dot{\hat{q}}_{H}(s)=-\hat{h}(t)-M\gamma(t)\hat{q}_{H}(0).
\end{equation}
The $M\gamma(0)\hat{q}_{H}(t)$ term would be eliminated exactly by the counterterm $\frac{1}{2}\sum_{j}\frac{c_{j}^{2}}{m_{j}\omega_{j}^2}\hat{q}^2$ that we have dropped. Nevertheless, the key is that we can identify the equation of a forced harmonic oscillator from the Langevin equation and the second-order term from $\hat{F}_{H}(t)$ is exactly associated with damping (the term proportional to $\gamma \dot{\hat{q}}_{H}$). The same argument holds for the most general force operator linear in canonical variables i.e. $\hat{F}_{a} = \sum_{j}c_{aj}\hat{x}_{j}+d_{aj}\hat{p}_{j}$ with $\hat{H}_{\text{sys-bath}} = \sum_{a}\hat{Q}_{a}\hat{F}_{a}$ for $\hat{Q}_{a}$ some system operator linear in canonical variables. Thus, dropping the second-order term is equivalent to working on a timescale where frictional contributions can be neglected.

\section{II. Full Proof of Gaussian Calculation}
\label{pert appendix}
Here we provide the full proof of Theorem \ref{pert-theorem}. Our system Hamiltonian is described by \eqref{rank 1 interaction hamiltonian}, from which we derive the Langevin equation:
\begin{equation}
\begin{split}
    \frac{d}{dt}\hat{\bm{\xi}}_{H}(t) &=\left[ 
    \begin{pmatrix}
    \mathbf{M}_A & \\
    & \mathbf{M}_B
    \end{pmatrix}
        +
    K_g \begin{pmatrix}
  & \bm{\eta}\mathbf{u}_A \mathbf{u}_B^T \\
    \bm{\eta}\mathbf{u}_B \mathbf{u}_A^T &
    \end{pmatrix}\right] \hat{\bm{\xi}}_{H}(t) +     \begin{pmatrix}
    \bm{\eta}\mathbf{u}_A \hat{F}^{\rm th}_{A,H}(t) \\
    \bm{\eta}\mathbf{u}_B \hat{F}^{\rm th}_{B,H}(t)           
    \end{pmatrix}\equiv\bm{M}\hat{\bm{\xi}}_{H}(t)+K_{g}\bm{C}\hat{\bm{\xi}}_{H}(t)+\hat{\bm{n}}_{H}(t).
\end{split}
\end{equation}
Dropping the second-order terms in $\hat{F}_{H}(t)$, the solution of the Langevin equation is 
\begin{equation}
    \hat{\bm{\xi}}_{H}(t) = e^{(\bm{M}+K_{g}\bm{C})t}\hat{\bm{\xi}}_{H}(0)+\int_{0}^{t}ds\ e^{(\bm{M}+K_{g}\bm{C})(t-s)}\hat{\bm{n}}_{I}(s).
\end{equation}
and thus the covariance-matrix as a function of time is:
\begin{equation}
        \bm{V}(t) = e^{(\bm{M}+K_{g}\bm{C})t}\bm{V}(0)e^{(\bm{M}+K_{g}\bm{C})^{T}t} + \int_{0}^{t}du \ e^{(\bm{M}+K_{g}\bm{C})u}\bm{F}e^{(\bm{M}+K_{g}\bm{C})^{T}u}
\end{equation}
where 
\begin{equation}
    \bm{F} = \begin{pmatrix}
        \bm{\eta}\textbf{u}_{A}\textbf{u}_{A}^{T}\bm{\eta}^{T}S^{\text{th}}_{A} & 0\\
        0 & \bm{\eta}\textbf{u}_{B}\textbf{u}_{B}^{T}\bm{\eta}^{T}S^{\text{th}}_{B}
    \end{pmatrix}
\end{equation}
and we have used a white free-bath noise spectrum: $\frac{1}{2}\langle\{\hat{F}_{A,I}(t),\hat{F}_{B,I}(t')\}\rangle = \delta_{AB}\delta(t-t')S^{\text{th}}_{A}$.

A Gaussian state of any number of modes is separable if and only if there exists $\bm{\sigma}_{A}$, $\bm{\sigma}_{B}$ such that~\cite{serafini2017}:
\begin{equation}
    \bm{V} \succeq \bm{\sigma}_{A} \oplus\bm{\sigma}_{B}
\end{equation}
where $\bm{V}$ is the covariance matrix of the state, and $\bm{\sigma}_{A}$ and $\bm{\sigma}_{B}$ satisfy $\bm{V}_{j}+\frac{i}{2}\bm{\Omega} \geq 0$ for $j\in\{A,B\}$. In our case, this means 
\begin{equation}
\label{initial criteria}
    \bm{V}(t) = \bm{\Phi}(t)\bm{V}(0)\bm{\Phi}^{T}(t)+\int_{0}^{t}du \ \bm{\Phi}(u)\bm{F}\bm{\Phi}^{T}(u) = (\bm{\sigma}_{A}\oplus\bm{\sigma}_{B})+\bm{N}
\end{equation}
for some $\bm{N}\succeq 0$ and $\bm{\Phi}(t) \equiv e^{(\bm{M}+K_{g}\bm{C})t}$. Note that for $\bm{S}_{\text{loc}} \equiv \bm{S}_{A}\oplus\bm{S}_{B}$ and $\bm{S}_{A,B}$ symplectic, then
\begin{equation}
    \bm{V}\succeq\bm{\sigma}_{A}\oplus\bm{\sigma}_{B} \Longleftrightarrow \bm{S}_{\text{loc}}\bm{\sigma}\bm{S}^{T}_{\text{loc}} \succeq (\bm{S}_{A}\bm{\sigma}_{A}\bm{S}^{T}_{A})\oplus(\bm{S}_{B}\bm{\sigma}_{B}\bm{S}^{T}_{B})
\end{equation}

Suppose $\bm{V}(0)$ is a block-diagonal pure Gaussian state to first order in small parameters (specifically in $ S_{A,B}$), such that it admits the Williamson decomposition $\bm{V}(0) = \bm{S}^{-1}(\frac{1}{2}\mathbf{1})\bm{S}^{-T}$.
We choose this $\bm{S}$ as our canonical transformation and conjugate both sides of the equation (we assume $\bm{S} = \bm{S}_{A} \oplus \bm{S}_{B})$:
\begin{equation}
    \widetilde{\bm{V}}(t) = \widetilde{\bm{\Phi}}(t)\left(\frac{1}{2}\mathbf{1}\right)\widetilde{\bm{\Phi}}^{T}(t) + \int_{0}^{t}du \ \widetilde{\bm{\Phi}}(u)\widetilde{\bm{F}}\widetilde{\bm{\Phi}}^{T}(u)
\end{equation}
where $\widetilde{\bm{V}}(t) \equiv \bm{S}\bm{V}(t)\bm{S}^{T}$, $\widetilde{\bm{\Phi}}(t) \equiv \bm{S}\bm{\Phi}(t)\bm{S}^{-1}$, and $\widetilde{\bm{F}} \equiv \bm{S}\bm{F}\bm{S}^{T}$. Note that conjugation by symplectic matrices preserves Gaussianity and definiteness.

Expanding the first term gives:
\begin{equation}
    \begin{split}
        \bm{SMS}^{-1} &= \begin{pmatrix}
            \bm{S}_{A}\bm{M}_{A}\bm{S}^{-1}_{A} & \\
            & \bm{S}_{B}\bm{M}_{B}\bm{S}^{-1}_{B}
        \end{pmatrix}
        \equiv \begin{pmatrix}
            \bm{M}'_{A} &\\
            & \bm{M}'_{B}
        \end{pmatrix} \equiv \bm{M}'\\
        \bm{SCS}^{-1} &= \begin{pmatrix}
            & \bm{S}_{A}\bm{\eta}\textbf{u}_{A}\textbf{u}_{B}^{T}\bm{S}_{B}^{-1}\\
            \bm{S}_{B}\bm{\eta}\textbf{u}_{B}\textbf{u}_{A}^{T}\bm{S}_{A}^{-1}
        \end{pmatrix} \equiv \begin{pmatrix} &
            \bm{\eta}\textbf{w}_{A}\textbf{w}_{B}^{T}\\
            \bm{\eta}\textbf{w}_{B}\textbf{w}_{A}^{T}
        \end{pmatrix} \equiv \bm{C}'\\
        \Longrightarrow  \widetilde{\bm{\Phi}}(t) &= \bm{S}\bm{\Phi}(t)\bm{S}^{-1} =e^{\bm{S}(\bm{M}+K_{g}\bm{C})\bm{S}^{-1}t} = e^{(\bm{M}'+K_{g}\bm{C}')t}
    \end{split}
\end{equation}
where we defined $\textbf{w}_{i} \equiv \bm{S}^{-T}_{i}\textbf{u}_{i}$ for $i\in \{A,B\}$ and $e^{\bm{M}'t}$ is still block symplectic. The remaining term to expand is the thermal noise matrix:
\begin{equation}
    \begin{split}
        \widetilde{\bm{F}} \equiv \bm{SFS}^{T} = \begin{pmatrix}
            \bm{S}_{A}\bm{\eta}\textbf{u}_{A}\textbf{u}_{A}^{T}\bm{\eta}^{T}\bm{S}_{A}^{T}S^{\text{th}}_{A} &\\
            & \bm{S}_{B}\bm{\eta}\textbf{u}_{B}\textbf{u}_{B}^{T}\bm{\eta}^{T}\bm{S}_{B}^{T}S^{\text{th}}_{B}
        \end{pmatrix} \equiv \begin{pmatrix}
            \bm{\eta}\textbf{w}_{A}\textbf{w}_{A}^{T}\bm{\eta}^{T}S^{\text{th}}_{A} & \\
            & \bm{\eta}\textbf{w}_{B}\textbf{w}_{B}^{T}\bm{\eta}^{T}S^{\text{th}}_{B}
        \end{pmatrix}
    \end{split}
\end{equation}

Now we expand all terms to first order in $K_{g}t,S_{A}t,S_{B}t$:
\begin{equation}
\begin{aligned}
    \widetilde{\bm{V}}(t) &= \left(e^{\bm{M}'t}+K_{g}\int_{0}^{t}e^{\bm{M}'(t-s)}\bm{C}'e^{\bm{M}'s}ds\right)\left(\frac{1}{2}\mathbf{1} \right)\left(e^{\bm{M}'t}+K_{g}\int_{0}^{t}e^{\bm{M}'(t-s)}\bm{C}'e^{\bm{M}'s}ds\right)^{T}+ \int_{0}^{t}e^{\bm{M}'u}\widetilde{\bm{F}}e^{\bm{M}'^{T}u}du\\
    &=\frac{1}{2}\left[e^{\bm{M}'t}e^{\bm{M'^{T}}t}+K_{g}\int_{0}^{t}e^{\bm{M}'(t-s)}\bm{C}'e^{\bm{M}'s}dse^{\bm{M}'^{T}t} + K_{g}e^{\bm{M}'t}\int_{0}^{t}e^{\bm{M}'^{T}s}\bm{C}'^{T}e^{\bm{M}'^{T}(t-s)}ds\right] + \int_{0}^{t}e^{\bm{M}'u}\widetilde{\bm{F}}e^{\bm{M}'^{T}u}du
\end{aligned}
\end{equation}
Note that $e^{\bm{M}'t} = e^{\bm{M}'_{A}t} \oplus e^{\bm{M}'_{B}t}$ with $e^{\bm{M}'_{A}t}, e^{\bm{M}'_{B}t}$ symplectic so we can conjugate everything by $e^{-\bm{M}'t}(\cdot)e^{-\bm{M}'^{T}t}$:
\begin{equation}
\begin{aligned}
    \widetilde{\bm{V}}(t) &= \frac{1}{2}\left[\mathbf{1}+K_{g}\int_{0}^{t}e^{-\bm{M}'s}\bm{C}'e^{\bm{M}'s}ds + K_{g}\int_{0}^{t}e^{\bm{M}'^{T}s}\bm{C}'^{T}e^{-\bm{M}'^{T}s}ds\right] + \int_{0}^{t}ds \ e^{-\bm{M}'s}\widetilde{\bm{F}}e^{-\bm{M}'^{T}s}
\end{aligned}
\end{equation}

We now introduce a perturbative term $\delta\bm{V}$ to translate the above problem into a physically equivalent, but simpler one to analyze:
\begin{equation}
    \delta\bm{V}(t) = \begin{pmatrix}
        S_{A}^{\text{th}}\frac{1}{2}\int_{0}^{t}ds\left[\textbf{v}_{A}'(s)\textbf{v}_{A}'^{T}(s) - \bm{\eta}\textbf{v}_{A}'(s)\textbf{v}_{A}'^{T}(s)\bm{\eta}^{T}\right] & \\
        &  S_{B}^{\text{th}}\frac{1}{2}\int_{0}^{t}ds\left[\textbf{v}_{B}'(s)\textbf{v}_{B}'^{T}(s) - \bm{\eta}\textbf{v}_{B}'(s)\textbf{v}_{B}'^{T}\bm{\eta}^{T}\right]
    \end{pmatrix}
\end{equation}
where we have defined $\textbf{v}'_{i}(s) \equiv e^{\bm{M}'^{T}_{i}s}\textbf{w}_{i}$ for $i\in\{A,B\}$.

Let us now check that $\frac{1}{2}\mathbf{1} \pm \delta\bm{V}(t)$ is still a valid covariance matrix. We use the fact that all pure state covariance matrices satisfy:
\begin{equation}
    \bm{V}_{\text{pure}}\begin{pmatrix}
        \bm{\eta} & \\
        & \bm{\eta}
    \end{pmatrix}\bm{V}_{\text{pure}} = \frac{1}{4}\begin{pmatrix}
        \bm{\eta} & \\
        & \bm{\eta}
    \end{pmatrix}
\end{equation}
Plugging in gives:
\begin{equation}
\begin{aligned}
    \left(\frac{1}{2}\mathbf{1} \pm \delta\bm{V}(t)\right)\begin{pmatrix}
        \bm{\eta} & \\
        & \bm{\eta}
    \end{pmatrix}\left(\frac{1}{2}\mathbf{1} \pm \delta\bm{V}(t)\right) &= \frac{1}{4}\begin{pmatrix}
        \bm{\eta} & \\
        & \bm{\eta}
    \end{pmatrix}\pm\frac{1}{2}\left(\begin{pmatrix}
        \bm{\eta} & \\
        & \bm{\eta}
    \end{pmatrix}\delta\bm{V}+\delta\bm{V}\begin{pmatrix}
        \bm{\eta} & \\
        & \bm{\eta}
    \end{pmatrix}\right)+\mathcal{O}(\delta\bm{V}^2)
\end{aligned}
\end{equation}
Thus we need to check that $\begin{pmatrix}
        \bm{\eta} & \\
        & \bm{\eta}
    \end{pmatrix}\delta\bm{V}+\delta\bm{V}\begin{pmatrix}
        \bm{\eta} & \\
        & \bm{\eta}
    \end{pmatrix}=0$. This leads to the equation:
\begin{equation}
    \begin{split}
        &S_{i}^{\text{th}}\frac{1}{2}\left\{\int_{0}^{t}ds\ \bm{\eta}\left[\textbf{v}_{i}'(s)\textbf{v}_{i}'^{T}(s) - \bm{\eta}\textbf{v}_{i}'(s)\textbf{v}_{i}'^{T}(s)\bm{\eta}^{T}\right] +         \left[\textbf{v}_{i}'(s)\textbf{v}_{i}'^{T}(s) - \bm{\eta}\textbf{v}_{i}'(s)\textbf{v}_{i}'^{T}(s)\bm{\eta}^{T}\right]\bm{\eta}\right\}\\
        &=\int_{0}^{t}ds\ S^{\text{th}}_{i}\frac{1}{2}[\bm{\eta}\textbf{v}'_{i}(s)\textbf{v}_{i}'^{T}(s)+\textbf{v}'_{i}(s)\textbf{v}_{i}'^{T}(s)\bm{\eta}^{T}+\textbf{v}_{i}'(s)\textbf{v}_{i}'^{T}(s)\bm{\eta}-\bm{\eta}\textbf{v}'_{i}(s)\textbf{v}_{i}'^{T}(s)]\\
        &=0
    \end{split}
\end{equation}
This holds for both $i\in\{A,B\}$ so to linear order, $\frac{1}{2}\mathbf{1}\pm\delta\bm{V}(t)$ is still a valid covariance matrix.

Now we can rewrite our covariance matrix as:
\begin{equation}
\begin{aligned}
    \widetilde{\bm{V}}(t) &= \widetilde{\bm{V}}(t) -\delta\bm{V}(t)+\delta\bm{V}(t)\\
    &=\left(\frac{1}{2}\mathbf{1}-\delta\bm{V}(t)\right) + \frac{1}{2}\left[K_{g}\int_{0}^{t}e^{-\bm{M}'s}\bm{C}'e^{\bm{M}'s}ds + K_{g}\int_{0}^{t}e^{\bm{M}'^{T}s}\bm{C}'^{T}e^{-\bm{M}'^{T}s}ds\right]+ \int_{0}^{t}ds \ e^{-\bm{M}'s}\widetilde{\bm{F}}e^{-\bm{M}'^{T}s}+\delta\bm{V}(t)\\
\end{aligned}
\end{equation}

Thus, to check separability (to linear order) of $\widetilde{\bm{V}}(t)$, it is sufficient to derive a positive semi-definite (PSD) condition for:
\begin{equation}
\begin{split}
\bm{V}_{g}(t)+\bm{V}_{\text{th}}(t) + \delta\bm{V}(t) &\equiv \frac{1}{2}\left[K_{g}\int_{0}^{t}e^{-\bm{M}'s}\bm{C}'e^{\bm{M}'s}ds + K_{g}\int_{0}^{t}e^{\bm{M}'^{T}s}\bm{C}'^{T}e^{-\bm{M}'^{T}s}ds\right] + \int_{0}^{t}ds \ e^{-\bm{M}'s}\widetilde{\bm{F}}e^{-\bm{M}'^{T}s} + \delta\bm{V}(t)\\
&\equiv \int_{0}^{t}ds \ [\bm{I}_{g}(s)+\bm{I}_{\text{th}}(s)+\delta\bm{I}(s)]
\end{split}
\end{equation}
Since an integral of PSD matrices is PSD, it suffices to just study the integrand, i.e.:
\begin{equation}
\begin{split}
    &\bm{I}_{g}(s)+\bm{I}_{\text{th}}(s)+\delta\bm{I}(s) = 
    \begin{pmatrix}
       S_{A}^{\text{th}}\frac{1}{2}\left[\textbf{v}_{A}'(s)\textbf{v}_{A}'^{T}(s) + \bm{\eta}\textbf{v}_{A}'(s)\textbf{v}_{A}'^{T}(s)\bm{\eta}^{T}\right] & \frac{K_{g}}{2}\left[\bm{\eta}\textbf{v}'_{A}(s)\textbf{v}'^{T}_{B}(s)+\textbf{v}'_{A}(s)\textbf{v}'^{T}_{B}(s)\bm{\eta}^{T}\right]\\
        \frac{K_{g}}{2}\left[\bm{\eta}\textbf{v}_{B}'(s)\textbf{v}'^{T}_{A}(s) + \textbf{v}'_{B}(s)\textbf{v}'^{T}_{A}(s)\bm{\eta}^{T}\right]& S_{B}^{\text{th}}\frac{1}{2}\left[\textbf{v}_{B}'(s)\textbf{v}_{B}'^{T}(s) + \bm{\eta}\textbf{v}_{B}'(s)\textbf{v}_{B}'^{T}\bm{\eta}^{T}(s)\right]
    \end{pmatrix}\\
\end{split}
\end{equation}

Instead of proceeding via generalized Schur complement, consider the following:

\begin{lemma}
For $\bm{M} = \bm{U}\bm{S}\bm{U}^{\dagger}$ and $\bm{G} = \bm{U}^{\dagger}\bm{U}$, then $\bm{M} \succeq 0$ if and only if $\bm{G}\bm{S}\bm{G}\succeq 0$.
\end{lemma}

\begin{proof} Indeed, if $\bm{M} \succeq 0$, then for any $\textbf{z}$, we have 
\begin{equation}
    \textbf{z}^{\dagger}(\bm{G}\bm{S}\bm{G})\textbf{z} = \textbf{z}^{\dagger}\bm{U}^{\dagger}(\bm{U}\bm{S}\bm{U^{\dagger})}\bm{U}\textbf{z}=(\bm{U}\textbf{z})^{\dagger}\bm{M}(\bm{U}\textbf{z}) \geq 0
\end{equation}
so $\bm{G}\bm{S}\bm{G} \succeq 0$. 

Conversely if $\bm{GSG} \succeq 0$, for any $\textbf{y}$, let $\textbf{x} \equiv \bm{U}^{\dagger}\textbf{y}$ so that $\textbf{y}^{\dagger}\bm{M}\textbf{y} = \textbf{x}^{\dagger}\bm{S}\textbf{x}$. Because $\ker(\bm{U}) = \ker(\bm{U}^{\dagger}\bm{U})$, by Rank-Nullity, we have $\text{rank}(\bm{G}) = \text{rank}(\bm{U}) = \text{rank}(\bm{U}^{\dagger})$ and since $\text{Im}(\bm{G}) \subseteq\text{Im}(\bm{U}^{\dagger})$, we have $\text{Im}(\bm{G}) = \text{Im}(\bm{U}^{\dagger})$. Thus, $\textbf{x}\in\text{Im}(\bm{G})$ so there exists $\textbf{z}$ such that $\textbf{x}=\bm{G}\textbf{z}$ and we have:
\begin{equation}
    \textbf{y}^{\dagger}\bm{M}\textbf{y} = \textbf{x}^{\dagger}\bm{M}\textbf{x} = (\bm{G}\textbf{z})^{\dagger}\bm{M}(\bm{G}\textbf{z}) = \textbf{z}^{\dagger}(\bm{G}\bm{S}\bm{G})\textbf{z} \geq 0
\end{equation}
so we conclude that $\bm{M} \succeq 0$. 
\end{proof}

We can thus rewrite:
\begin{equation}
    \begin{split}
        &\bm{I}_{g}(s)+\bm{I}_{\text{th}}(s)+\delta\bm{I}(s) = \bm{USU}^{\dagger} =\begin{pmatrix}
            \textbf{v}_{A} & \bm{\eta}\textbf{v}_{A} & 0 & 0\\
            0 & 0 & \textbf{v}_{B} & \bm{\eta}\textbf{v}_{B}
        \end{pmatrix}
        \begin{pmatrix}
            \frac{S_{A}^{\text{th}}}{2} &0 &0 & \frac{K_{g}}{2}\\
            0 & \frac{S^{\text{th}}_{A}}{2} & \frac{K_{g}}{2} & 0\\
            0 & \frac{K_{g}}{2} & \frac{S^{\text{th}}_{B}}{2} & 0\\
            \frac{K_{g}}{2} & 0 & 0 & \frac{S^{\text{th}}_{B}}{2}
        \end{pmatrix}
        \begin{pmatrix}
            \textbf{v}_{A}^{T} & 0\\
            (\bm{\eta}\textbf{v}_{A})^{T} & 0\\
            0 & \textbf{v}_{B}^{T}\\
            0 & (\bm{\eta}\textbf{v}_{B})^{T}
        \end{pmatrix}
    \end{split}
\end{equation}
where we dropped the time dependence to avoid clutter (it will not be relevant in the following analysis).
We then have:
\begin{equation}
    \begin{split}
        \bm{G} &= \bm{U}^{\dagger}\bm{U} = \text{diag}(||\textbf{a}||^2,||\textbf{a}||^2,||\textbf{b}||^2,||\textbf{b}||^2),\\
        \bm{GSG} &= \begin{pmatrix}
            \frac{S^{\text{th}}_{A}}{2}||\textbf{a}||^{4} & 0 & 0 & \frac{K_{g}}{2}||\textbf{a}||^2||\textbf{b}||^2\\
            0 & \frac{S_{A}^{\text{th}}}{2}||\textbf{a}||^4 & \frac{K_{g}}{2}||\textbf{a}||^2||\textbf{b}||^2 & 0\\
            0 & \frac{K_{g}}{2}||\textbf{a}||^2||\textbf{b}||^2 & \frac{S^{\text{th}}_{B}}{2}||\textbf{b}||^4 & 0 \\
            \frac{K_{g}}{2}||\textbf{a}||^2||\textbf{b}||^2 & 0 & 0 &\frac{S^{\text{th}}_{B}}{2}||\textbf{b}||^4
        \end{pmatrix}
    \end{split}
\end{equation}
where we defined $||\textbf{i}||^2=\textbf{v}_{i}^{T}\textbf{v}_{i} = (\bm{\eta}\textbf{v}_{i})^{T}(\bm{\eta}\textbf{v}_{i})$ for $\textbf{i}\in\{\textbf{a},\textbf{b}\}$. PSD is preserved by conjugation so conjugation by permutation matrix $\bm{P} = \begin{pmatrix}
    1 & 0 & 0 & 0\\
    0 & 0 & 0 & 1\\
    0 & 1 & 0 & 0\\
    0 & 0 & 1 & 0
\end{pmatrix}$ gives:
\begin{equation}
    \bm{P}(\bm{GSG})\bm{P}^{T} = \left(
\begin{array}{cccc}
 \dfrac{\|\mathbf a\|^4 S^{\text{th}}_A}{2}
 & \dfrac{\|\mathbf a\|^2 \|\mathbf b\|^2 K_g}{2}
 & 0 & 0 \\[6pt]
 \dfrac{\|\mathbf a\|^2 \|\mathbf b\|^2 K_g}{2}
 & \dfrac{\|\mathbf b\|^4 S^{\text{th}}_B}{2}
 & 0 & 0 \\[6pt]
 0 & 0
 & \dfrac{\|\mathbf a\|^4 S^{\text{th}}_A}{2}
 & \dfrac{\|\mathbf a\|^2 \|\mathbf b\|^2 K_g}{2} \\[6pt]
 0 & 0
 & \dfrac{\|\mathbf a\|^2 \|\mathbf b\|^2 K_g}{2}
 & \dfrac{\|\mathbf b\|^4 S^{\text{th}}_B}{2}
\end{array}
\right)
\end{equation}
A block diagonal matrix is PSD if and only if the blocks are individually PSD and both blocks are identical so we just focus on PSD of one of the $2\times 2$ blocks. From this we easily derive $S_{A}^{\text{th}}S_{B}^{\text{th}} \geq K_{g}^2$. 

Now if we define $\frac{1}{2}\textbf{I}-\delta\bm{V}(t) \equiv \bm{\sigma}_{\text{Final},A}\oplus \bm{\sigma}_{\text{Final},B}$ where we know both  $\bm{\sigma}_{\text{Final},A}, \bm{\sigma}_{\text{Final},B}$ are valid covariance matrices to linear order in $S^{\text{th}}_{A}t,S^{\text{th}}_{B}t$, and pick:
\begin{equation}
    \bm{\sigma}_{i} = \bm{S}_{i}^{-1}e^{\bm{M}'_{i}t}\bm{\sigma}_{\text{Final},i}e^{\bm{M}_{i}'^{T}t}\bm{S}_{i}^{-T}
\end{equation}
for $i \in\{A,B\}$, \eqref{initial criteria} is satisfied if $\eqref{sufficient cond}$ is satisfied, thus concluding the argument.

\subsection{II.a. Including Correlated Noise}
The thermal noise matrix now becomes:
\begin{equation}
    \bm{F} = \begin{pmatrix}
        \bm{\eta}\textbf{u}_{A}\textbf{u}_{A}^{T}\bm{\eta}^{T}S^{\text{th}}_{A} & \bm{\eta}\textbf{u}_{A}\textbf{u}^{T}_{B}\bm{\eta}^{T} S^{\text{th}}_{AB}\\
        \bm{\eta}\textbf{u}_{B}\textbf{u}^{T}_{A}\bm{\eta}^{T}S^{\text{th}}_{AB} & \bm{\eta}\textbf{u}_{B}\textbf{u}_{B}^{T}\bm{\eta}^{T}S^{\text{th}}_{B}
    \end{pmatrix}
\end{equation}
Since this is a physical covariance matrix, it must be PSD, so the correlated noise terms must satisfy:
\begin{equation}
    (S^{\text{th}}_{AB})^2 \leq S^{\text{th}}_{A}S^{\text{th}}_{B}
\end{equation}
this is the assumption we will be making going forward whenever correlated noise is discussed.

As before, define 
\begin{equation}
    \delta\bm{V}(t) = \begin{pmatrix}
        \delta\bm{V}_{A}(t) & \bm{0}\\
        \bm{0} & \delta\bm{V}_{B}(t)
    \end{pmatrix}
\end{equation}
with 
\begin{equation}
    \delta\bm{V}_{i}(t) = \int_{0}^{t}ds \bm{R}_{i}(s)\bm{D}_{i}\bm{R}_{i}^{T}(s), \hspace{0.5cm}\bm{R}_{i}(s) = (\textbf{v}'_{i}(s), \bm{\eta}\textbf{v}'_{i}(s)), \hspace{0.5cm} \bm{D}_{i} = \bm{L}_{i}-\begin{pmatrix}
        0 & 0\\
        0 & S^{\text{th}}_{i}
    \end{pmatrix}
\end{equation}
where $\bm{L}_{i}$ is any matrix that satisfies $\bm{L}_{i} = \bm{L}_{i}^{T}$, $\bm{L}_{i} \succeq 0$, $\text{Tr}(\bm{L}_{i}) = S^{\text{th}}_{i}$ such that $\bm{D}_{i} = \bm{D}_{i}^{T}$ and $\text{Tr}(\bm{D}_{i}) = 0$. Notice now that:
\begin{equation}
    \bm{\eta}\bm{R}_{i}(s)\bm{D}_{i}\bm{R}_{i}^{T}(s) + \bm{R}_{i}(s)\bm{D}_{i}\bm{R}_{i}^{T}(s)\bm{\eta} = \bm{R}_{i}(s)(\bm{J}\bm{D}_{i}+\bm{D}_{i}\bm{J})\bm{R}^{T}_{i}(s) = \bm{R}_{i}(s)\text{Tr}(\bm{D}_{i})\bm{J}\bm{R}_{i}^{T}(s) = 0
\end{equation}
where we defined $\bm{J} = \begin{pmatrix}
    0 & -1 \\
    1 & 0
\end{pmatrix}$. This then implies $\begin{pmatrix}
        \bm{\eta} & \\
        & \bm{\eta} 
    \end{pmatrix}\delta\bm{V}+\delta\bm{V}\begin{pmatrix}
        \bm{\eta} & \\
        & \bm{\eta}
    \end{pmatrix}=0$ to linear order and our proof from before still follows. Our original proof had $\bm{L}_{i} = \frac{S^{\text{th}}_{i}}{2}\bm{1}_{2}$. Put together, we now have:
    \begin{equation}
        \bm{I}_{g}(s) + \bm{I}_{\text{th}}(s)+\delta\bm{I}(s) = \begin{pmatrix}
            \bm{R}_{A}(s) & \bm{0}\\
            \bm{0} & \bm{R}_{B}(s)
        \end{pmatrix}\begin{pmatrix}
            \bm{L}_{A} & \bm{X}\\
            \bm{X}^{T} & \bm{L}_{B}
        \end{pmatrix}
        \begin{pmatrix}
            \bm{R}^{T}_{A}(s) & \bm{0}\\
            \bm{0} & \bm{R}^{T}_{B}(s)
        \end{pmatrix}
    \end{equation}
    where $\bm{X} = \begin{pmatrix}
        0 & K_{g}/2\\
        K_{g}/2 & S^{\text{th}}_{AB}
    \end{pmatrix}$.
    
Conjugation preserves PSD, so it suffices to choose $\bm{L}_{A},\bm{L}_{B}$ such that $\begin{pmatrix}
            \bm{L}_{A} & \bm{X}\\
            \bm{X}^{T} & \bm{L}_{B}
        \end{pmatrix} \succeq 0$. Since $\bm{X} = \bm{X}^{T}$, we can diagonalize it via $\bm{X} = \bm{O}\bm{\Lambda}\bm{O}^{T}$ with:
\begin{equation}
    \bm{\Lambda} = \begin{pmatrix}
        \lambda_{+} & 0\\
        0 & \lambda_{-}
    \end{pmatrix}, \hspace{0.5cm} \lambda_{\pm} = \frac{S^{\text{th}}_{AB}\pm\sqrt{(S^{\text{th}}_{AB})^2+K_{g}^2}}{2}
\end{equation}
Defining $|\bm{X}| = \sqrt{\bm{X}^{T}\bm{X}} = \bm{O}|\bm{\Lambda}|\bm{O}^{T}$ with $|\bm{\Lambda}| = \text{diag}(|\lambda_{+}|,|\lambda_{-}|)$, we choose:
\begin{equation}
    \bm{L}_{A} = \frac{S^{\text{th}}_{A}}{\sqrt{K_{g}^{2}+(S^{\text{th}}_{AB})^2}}|\bm{X}|, \hspace{0.5cm} \bm{L}_{B} = \frac{S^{\text{th}}_{B}}{\sqrt{K_{g}^{2}+(S^{\text{th}}_{AB})^2}}|\bm{X}|
\end{equation}
These are valid choices of $\bm{L}_{i}$ as they satisfy symmetry, PSD, and $\text{Tr}(\bm{L}_{i}) = S^{\text{th}}_{i}$. We can now perform orthogonal conjugation to get:
\begin{equation}
    \begin{pmatrix}
        \bm{L}_{A} & \bm{X}\\
        \bm{X}^{T} & \bm{L}_{B}
    \end{pmatrix}\succeq 0 \Longleftrightarrow \begin{pmatrix}
        \bm{O}^{T} & \bm{0}\\
        \bm{0} & \bm{O}^{T}
    \end{pmatrix}
\begin{pmatrix}
        \bm{L}_{A} & \bm{X}\\
        \bm{X}^{T} & \bm{L}_{B}
    \end{pmatrix}
    \begin{pmatrix}
        \bm{O} & \bm{0}\\
        \bm{0} & \bm{O}
    \end{pmatrix} = \begin{pmatrix}
        \frac{S^{\text{th}}_{A}}{\sqrt{K_{g}^{2}+(S^{\text{th}}_{AB})^2}}|\bm{\Lambda}| & \bm{\Lambda}\\
        \bm{\Lambda} &  \frac{S^{\text{th}}_{B}}{\sqrt{K_{g}^{2}+(S^{\text{th}}_{AB})^2}}|\bm{\Lambda}|
    \end{pmatrix} \succeq 0
\end{equation}
Permuting the second and third rows and columns gives:
\begin{equation}
    \begin{pmatrix}
        \frac{S^{\text{th}}_{A}}{\sqrt{K_{g}^{2}+(S^{\text{th}}_{AB})^2}}|\bm{\Lambda}| & \bm{\Lambda}\\
        \bm{\Lambda} &  \frac{S^{\text{th}}_{B}}{\sqrt{K_{g}^{2}+(S^{\text{th}}_{AB})^2}}|\bm{\Lambda}|
    \end{pmatrix} \Longrightarrow |\lambda_{+}|\begin{pmatrix}
        \frac{S^{\text{th}}_{A}}{\sqrt{K_{g}^{2}+(S^{\text{th}}_{AB})^2}} & 1\\
       1 &  \frac{S^{\text{th}}_{B}}{\sqrt{K_{g}^{2}+(S^{\text{th}}_{AB})^2}}
    \end{pmatrix} \oplus |\lambda_{-}|\begin{pmatrix}
        \frac{S^{\text{th}}_{A}}{\sqrt{K_{g}^{2}+(S^{\text{th}}_{AB})^2}} & -1 \\
      -1  &  \frac{S^{\text{th}}_{B}}{\sqrt{K_{g}^{2}+(S^{\text{th}}_{AB})^2}}
    \end{pmatrix}
\end{equation}
Checking PSD reduces to checking PSD of these blocks, equivalently, that:
\begin{equation}
     \frac{S^{\text{th}}_{A}}{\sqrt{K_{g}^{2}+(S^{\text{th}}_{AB})^2}} \frac{S^{\text{th}}_{B}}{\sqrt{K_{g}^{2}+(S^{\text{th}}_{AB})^2}} -1\geq 0
\end{equation}
Simplifying then gives:
\begin{equation}
    S^{\text{th}}_{A}S^{\text{th}}_{B} \geq K_{g}^{2} + (S^{\text{th}}_{AB})^2
\end{equation}

\subsection{II.b. Stringent Necessary and Sufficient With Correlated Noise}
By restricting the dynamics of our time evolution, it is possible to derive a necessary and sufficient separability condition in the Gaussian setting even in the presence of correlated noise. We start back at the equation:
\begin{equation}
    \widetilde{\bm{V}}(t) = \left(\frac{1}{2}\mathbf{1}-\delta\bm{V}(t)\right) + [\bm{V}_{g}(t)+\bm{V}_{\text{th}}(t)+\delta\bm{V}(t)]
\end{equation}
Now define the transformation:
\begin{equation}
    \bm{L}(t) \equiv \mathbf{I}+\delta\bm{V}(t)
\end{equation}
To linear order, $L(t)$ is symplectic:
\begin{equation}
    \bm{L}\begin{pmatrix}
        \bm{\eta} & \\
        & \bm{\eta}
    \end{pmatrix}\bm{L}^{T} = \left(\mathbf{I}+\delta\bm{V}(t)\right)\begin{pmatrix}
        \bm{\eta} & \\
        & \bm{\eta}
    \end{pmatrix}\left(\mathbf{I}+\delta\bm{V}(t)\right) = \begin{pmatrix}
        \bm{\eta} & \\
        & \bm{\eta}
    \end{pmatrix} + \mathcal{O}(\delta\bm{V}^2(t))
\end{equation}
where we used the fact that $\begin{pmatrix}
        \bm{\eta} & \\
        & \bm{\eta}
    \end{pmatrix}\delta\bm{V}+\delta\bm{V}\begin{pmatrix}
        \bm{\eta} & \\
        & \bm{\eta}
    \end{pmatrix}=0$ and that $\delta\bm{V}(t)$ is symmetric.
Then, acting upon $\widetilde{\bm{V}}(t)$ and expanding to linear order again, we have:
\begin{equation}
    \overline{\bm{V}}(t)\equiv \bm{L}(t)\widetilde{\bm{V}}(t)\bm{L}^{T}(t) = \frac{1}{2}\mathbf{1}+[\bm{V}_{g}(t)+\bm{V}_{\text{th}}(t)+\delta\bm{V}(t)]
\end{equation}
Note that $\bm{L}(t)$ is a block-diagonal matrix, so conjugation by it would not change separability or entanglement between subsystems. That is, $\overline{\bm{V}}(t)$ is separable if and only if $\widetilde{\bm{V}}(t)$ is separable. The main reason for this step is to map the locally deformed vacuum contribution
$\frac12\mathbf 1-\delta \bm V(t)$ back to $\frac12\mathbf 1$, so that all
first-order corrections are collected in the additive perturbation
$\bm V_g(t)+\bm V_{\rm th}(t)+\delta\bm V(t)$.

Now we impose the first extra condition, using the same notation as the previous section:
\begin{equation}
    \textbf{v}'_{A}(s)  = f_{A}(s)\textbf{w}_{A}, \hspace{0.5cm} \textbf{v}'_{B}(s) = f_{B}(s)\textbf{w}_{B}
\end{equation}
 for some shape functions capturing the dynamics $f_{A}(s),f_{B}(s)$ and where recall $\textbf{w}_{i} \equiv \textbf{S}^{-T}_{i}\textbf{u}_{i}$. That is, the system dynamics remains parallel to only one canonical mode sector (defined by $\textbf{w}_{i}$). The reason for this restriction is that we can now study the full integral equation rather than just the integrand. Defining:
 \begin{equation}
     I_{A} \equiv \int_{0}^{t}f_{A}(s)^2ds, \hspace{0.3cm} I_{B} \equiv \int_{0}^{t}f_{B}(s)^2ds, \hspace{0.3cm} I_{AB} \equiv \int_{0}^{t}f_{A}(s)f_{B}(s)ds
 \end{equation}
we can equivalently write:
 \(\bm R_i(s)=f_i(s)\bm R_i^{(0)}\), where
\begin{equation}
    \bm R_i^{(0)}\equiv
    \begin{pmatrix}
        \mathbf w_i & \bm\eta\mathbf w_i
    \end{pmatrix}.
\end{equation}
We thus obtain
\begin{equation}
\begin{aligned}
    \bm V_g(t)+\bm V_{\rm th}(t)+\delta\bm V(t)
    &=
    \begin{pmatrix}
        \bm R_A^{(0)}&0\\
        0&\bm R_B^{(0)}
    \end{pmatrix}
    \begin{pmatrix}
        I_A\bm L_A&I_{AB}\bm X\\
        I_{AB}\bm X^T&I_B\bm L_B
    \end{pmatrix}
    \begin{pmatrix}
        \bm R_A^{(0)T}&0\\
        0&\bm R_B^{(0)T}
    \end{pmatrix},
\end{aligned}
\end{equation}
where
\begin{equation}
    \bm X=
    \begin{pmatrix}
        0&K_g/2\\
        K_g/2&S_{AB}^{\rm th}
    \end{pmatrix}.
\end{equation}
The choice of $
    \bm{L}_{A}^{T} = \frac{S^{\text{th}}_{A}}{\sqrt{K_{g}^{2}+(S^{\text{th}}_{AB})^2}}|\bm{X}|, \hspace{0.5cm} \bm{L}_{B}^{T} = \frac{S^{\text{th}}_{B}}{\sqrt{K_{g}^{2}+(S^{\text{th}}_{AB})^2}}|\bm{X}|$ is still identical to before.
    
Defining $\alpha \equiv ||\textbf{w}_{A}||, \beta \equiv ||\textbf{w}_{B}||$, define the local basis vectors on subsystems $A,B$:
\begin{equation}
    \textbf{e}_{1,A} \equiv \frac{\textbf{w}_{A}}{\alpha}, \hspace{0.2cm} \textbf{e}_{2,A} \equiv -\frac{\bm{\eta}\textbf{w}_{A}}{\alpha}, \hspace{0.5cm} \textbf{e}_{1,B} \equiv \frac{\textbf{w}_{B}}{\beta}, \hspace{0.2cm} \textbf{e}_{2,B} \equiv -\frac{\bm{\eta}\textbf{w}_{B}}{\beta}
\end{equation}
Within their respective subspaces, the satisfy the definition of an orthonomal basis vector and additionally form a canonical symplectic pair in the sense that $\textbf{e}_{i}^{T}\bm{\eta}\textbf{e}_{j} = 1\delta_{i,j}$. We can now extend this via a symplectic Gram-Schmidt procedure to a full orthonormal symplectic basis within each subsystem. 

Now define local orthogonal symplectic matrices $\bm{O}_{A},\bm{O}_{B}$ whose rows are the orthogonal symplectic basis set, and ordered such that the first two rows are $\textbf{e}^{T}_{1,A},\textbf{e}^{T}_{2,A}$ and $\textbf{e}^{T}_{1,B},\textbf{e}^{T}_{2,B}$ respectively. Defining $\bm{O} \equiv \bm{O}_{A}\oplus\bm{O}_{B}$, we are interested in the matrix $\bm{O}\overline{\bm{V}}\bm{O}^{T}$. Trivially, $\bm{O}\mathbf{1}\bm{O}^{T} = \mathbf{1}$. The remaining terms can be studied block by block:
\begin{equation}
    \begin{aligned}  \bm{O}_{A}\left[\bm{R}^{(0)}_{A}I_{A}\bm{L}_{A}\bm{R}^{(0)T}_{A}\right]\bm{O}_{A}^{T} &= I_{A}\alpha^2
        \begin{pmatrix}
            \bm{Q}\bm{L}_{A}\bm{Q} & \bm{0}\\
            \bm{0} & \bm{0}
        \end{pmatrix}\\
            \bm{O}_{B}\left[\bm{R}^{(0)}_{B}I_{B}\bm{L}_{B}\bm{R}^{(0)T}_{B}\right]\bm{O}_{B}^{T} &= I_{B}\beta^2
        \begin{pmatrix}
            \bm{Q}\bm{L}_{B}\bm{Q} & \bm{0}\\
            \bm{0} & \bm{0}
        \end{pmatrix}\\
        \bm{O}_{A}\left[\bm{R}^{(0)}_{A}I_{AB}\bm{X}\bm{R}^{(0)T}_{B}\right]\bm{O}_{B}^{T} &= I_{AB}\alpha\beta
        \begin{pmatrix}
            \bm{Q}\bm{X}\bm{Q} & \bm{0}\\
            \bm{0} & \bm{0}
        \end{pmatrix}
    \end{aligned}
\end{equation}
where $\bm{Q} = \begin{pmatrix}
    1 & 0\\
    0 & -1
\end{pmatrix}$.
Note that the only nonzero elements of this transformed matrix is only in the subspace spanned by $\{\textbf{e}_{1,A}, \textbf{e}_{2,A},\textbf{e}_{1,B},\textbf{e}_{2,B}\}$ so we can rearrange the basis elements such that:
\begin{equation}
    \bm{O}\overline{\bm{V}}(t)\bm{O}^{T} = \left\{\frac{1}{2}\bm{1}_{4} + 
    \begin{pmatrix}
        I_{A}\alpha^2\bm{Q}\bm{L}_{A}\bm{Q} & I_{AB}\alpha\beta\bm{QXQ}\\
        I_{AB}\alpha\beta\bm{Q}\bm{X}^{T}\bm{Q} & I_{B}\beta^2\bm{Q}\bm{L}_{B}\bm{Q}
    \end{pmatrix}
    \right\} \oplus \frac{1}{2}\bm{1}_{2(n_{A}+n_{B})-4} \equiv \bm{V}_{\text{nontrivial}} \oplus \frac{1}{2}\bm{1}_{2(n_{A}+n_{B})-4}
\end{equation}
Again, the entanglement structure of the original problem has not changed because we are conjugating by a block diagonal symplectic matrix. 

Notice how the structure of entanglement is very clear now, as the $\frac{1}{2}\bm{1}_{2(n_{A}+n_{B})-4}$ term is just a product state of $A,B$ vacuum states, so the overall state is separable across the $A,B$ partition if and only if the other term in the direct sum is separable. However, we have now reduced the problem to a $1\times 1$ system, where the PPT condition is sufficient and necessary. For a general two-mode covariance matrix written in block form:
\begin{equation}
    \bm{V} = \begin{pmatrix}
        \bm{A} & \bm{C}\\
        \bm{C}^{T} & \bm{B}
    \end{pmatrix}
\end{equation}
the symplectic eigenvalue after partial transpose can be calculated via:
\begin{equation}
    \tilde{\nu}_{\pm} \equiv \sqrt{\frac{\tilde{\Delta}\pm\sqrt{\tilde{\Delta}^2-4\det V}}{2}}, \hspace{0.5cm} \tilde{\Delta} \equiv \det A+ \det B-2 \det C
\end{equation}
and the PPT condition is:
\begin{equation}
    \tilde{\nu}_{-} \geq \frac{1}{2}
\end{equation}
This can equivalently be rewritten as:
\begin{equation}
    \det V-\frac{\tilde{\Delta}}{4}+\frac{1}{16}\geq 0
\end{equation}
Applying this to $V = \bm{V}_{\text{nontrivial}}$
In our case, we calculate:
\begin{equation}
\label{PPT polynomial}
\frac{\alpha^2\beta^2}{4}
\left[
\frac{S_A^{\rm th}S_B^{\rm th}I_A I_B}{\tau^2}
-
I_{AB}^2
\right]
\left[
K_{g}^2+(S^{\text{th}}_{AB})^2
+
\frac{K_g^2}{2}
\left(
\alpha^2 I_A S_A^{\rm th}
+
\beta^2 I_B S_B^{\rm th}
\right)
+
\frac{K_g^4\alpha^2\beta^2}{4}
\left(
\frac{S_A^{\rm th}S_B^{\rm th}I_A I_B}{\tau^2}
-
I_{AB}^2
\right)
\right]
\geq 0 .
\end{equation}
The second factor in brackets is $\geq 0$, and we justify this by noting that since $\bm{V}_{\text{nontrivial}}$ is a valid covariance matrix, it is PSD, so we can perform a series of conjugations (first conjugating by $\bm{Q}\oplus\bm{Q}$ and then diagonalizing $\bm{X}$) to simplify it to:
\begin{equation}
\bm V_{\text{nontrivial}}
=
\begin{pmatrix}
\dfrac{1}{2}
+
\dfrac{\alpha^2 I_A S_A^{\rm th}}{2\tau}
\left(\tau+S_{AB}^{\rm th}\right)
&
0
&
\dfrac{\alpha\beta I_{AB}}{2}
\left(S_{AB}^{\rm th}+\tau\right)
&
0
\\[8pt]
0
&
\dfrac{1}{2}
+
\dfrac{\alpha^2 I_A S_A^{\rm th}}{2\tau}
\left(\tau-S_{AB}^{\rm th}\right)
&
0
&
\dfrac{\alpha\beta I_{AB}}{2}
\left(S_{AB}^{\rm th}-\tau\right)
\\[8pt]
\dfrac{\alpha\beta I_{AB}}{2}
\left(S_{AB}^{\rm th}+\tau\right)
&
0
&
\dfrac{1}{2}
+
\dfrac{\beta^2 I_B S_B^{\rm th}}{2\tau}
\left(\tau+S_{AB}^{\rm th}\right)
&
0
\\[8pt]
0
&
\dfrac{\alpha\beta I_{AB}}{2}
\left(S_{AB}^{\rm th}-\tau\right)
&
0
&
\dfrac{1}{2}
+
\dfrac{\beta^2 I_B S_B^{\rm th}}{2\tau}
\left(\tau-S_{AB}^{\rm th}\right)
\end{pmatrix},
\end{equation}
where $\tau \equiv \sqrt{K_g^2+\left(S_{AB}^{\rm th}\right)^2}$. By permuting rows and columns two and three, we can rewrite this as:
\begin{equation}
\begin{aligned}
    \bm{V}_{\text{nontrivial}}
    &=
    \begin{pmatrix}
    \dfrac{1}{2}
    +
    \dfrac{\alpha^2 I_A S_A^{\rm th}}{2\tau}
    \left(\tau+S_{AB}^{\rm th}\right)
    &
    \dfrac{\alpha\beta I_{AB}}{2}
    \left(S_{AB}^{\rm th}+\tau\right)
    \\[8pt]
    \dfrac{\alpha\beta I_{AB}}{2}
    \left(S_{AB}^{\rm th}+\tau\right)
    &
    \dfrac{1}{2}
    +
    \dfrac{\beta^2 I_B S_B^{\rm th}}{2\tau}
    \left(\tau+S_{AB}^{\rm th}\right)
    \end{pmatrix}
    \oplus
    \begin{pmatrix}
    \dfrac{1}{2}
    +
    \dfrac{\alpha^2 I_A S_A^{\rm th}}{2\tau}
    \left(\tau-S_{AB}^{\rm th}\right)
    &
    \dfrac{\alpha\beta I_{AB}}{2}
    \left(S_{AB}^{\rm th}-\tau\right)
    \\[8pt]
    \dfrac{\alpha\beta I_{AB}}{2}
    \left(S_{AB}^{\rm th}-\tau\right)
    &
    \dfrac{1}{2}
    +
    \dfrac{\beta^2 I_B S_B^{\rm th}}{2\tau}
    \left(\tau-S_{AB}^{\rm th}\right)
    \end{pmatrix}.
\end{aligned}
\end{equation}
PSD of $\bm{V}_{\text{nontrivial}}$ implies PSD of the individual blocks, i.e. that
\begin{equation}
\begin{aligned}
&
\left[
\frac{1}{2}
+
\frac{\alpha^2 I_A S_A^{\rm th}}{2\tau}
\left(\tau+S_{AB}^{\rm th}\right)
\right]
\left[
\frac{1}{2}
+
\frac{\beta^2 I_B S_B^{\rm th}}{2\tau}
\left(\tau+S_{AB}^{\rm th}\right)
\right]
-
\frac{\alpha^2\beta^2I_{AB}^2}{4}
\left(\tau+S_{AB}^{\rm th}\right)^2
\geq0,\\
&
\left[
\frac{1}{2}
+
\frac{\alpha^2 I_A S_A^{\rm th}}{2\tau}
\left(\tau-S_{AB}^{\rm th}\right)
\right]
\left[
\frac{1}{2}
+
\frac{\beta^2 I_B S_B^{\rm th}}{2\tau}
\left(\tau-S_{AB}^{\rm th}\right)
\right]
-
\frac{\alpha^2\beta^2I_{AB}^2}{4}
\left(\tau-S_{AB}^{\rm th}\right)^2
\geq0.
\end{aligned}
\end{equation}
Dividing the first inequality by $[(\tau+S^{\text{th}}_{AB})/2]^2$ gives:
\begin{equation}
\left[
\frac{\alpha^2 I_A S_A^{\rm th}}{\tau}
+
\frac{1}{\tau+S_{AB}^{\rm th}}
\right]
\left[
\frac{\beta^2 I_B S_B^{\rm th}}{\tau}
+
\frac{1}{\tau+S_{AB}^{\rm th}}
\right]
-
\alpha^2\beta^2I_{AB}^2
\geq0.
\end{equation}
Similarly, dividing the second inequality by $[(\tau-S^{\text{th}}_{AB})/2]^2$ gives:
\begin{equation}
\left[
\frac{\alpha^2 I_A S_A^{\rm th}}{\tau}
+
\frac{1}{\tau-S_{AB}^{\rm th}}
\right]
\left[
\frac{\beta^2 I_B S_B^{\rm th}}{\tau}
+
\frac{1}{\tau-S_{AB}^{\rm th}}
\right]
-
\alpha^2\beta^2I_{AB}^2
\geq0.
\end{equation}
Since:
\begin{equation}
    \frac{2\tau}{K_{g}^2} = \frac{1}{\tau+S^{\text{th}}_{AB}}+\frac{1}{\tau-S^{\text{th}}_{AB}}
\end{equation}
and the function:
\begin{equation}
    f(x) = \left[\frac{\alpha^2I_{A}S^{\text{th}}_{A}}{\tau}+x\right]\left[\frac{\beta^2I_{B}S^{\text{th}}_{B}}{\tau}+x\right]-\alpha^2\beta^2I^{2}_{AB} 
\end{equation}
is monotonic increasing for $x \geq 0$, we can write:
\begin{equation}
\left[
\frac{\alpha^2 I_A S_A^{\rm th}}{\tau}
+
\frac{2\tau}{K_g^2}
\right]
\left[
\frac{\beta^2 I_B S_B^{\rm th}}{\tau}
+
\frac{2\tau}{K_g^2}
\right]
-
\alpha^2\beta^2I_{AB}^2
\geq0.
\end{equation}
Multiplying by $K_{g}^4/4$ and expanding recovers the second factor in brackets in \eqref{PPT polynomial}. Thus, the sign of the expression is entirely decided by the first term in bracketes, i.e.:
\begin{equation}
    S^{\text{th}}_{A}S^{\text{th}}_{B}I_{A}I_{B} \geq I_{AB}^{2}\left[K_{g}^{2}+(S^{\text{th}}_{AB})^2\right]
\end{equation}
By Cauchy-Schwarz, we have that:
\begin{equation}
    I_{AB}^{2} \equiv \left(\int_{0}^{t}f_{A}(s)f_{B}(s)ds\right)^2 \leq \left(\int_{0}^{t}f_{A}(s)^2ds\right)\left(\int_{0}^{t}f_{B}(s)^2ds\right) \equiv I_{A}I_{B}
\end{equation}
so the PPT condition is equivalent to:
\begin{equation}
    S^{\text{th}}_{A}S^{\text{th}}_{B} \geq \rho^2\left[K_{g}^{2}+(S^{\text{th}}_{AB})^2\right], \hspace{0.5cm} \rho^2 \equiv \frac{I_{AB}^{2}}{I_{A}I_{B}}, \hspace{0.5cm} 0 \leq \rho^2 \leq 1
\end{equation}
If we further impose the stronger condition that $f_{A}(s) \propto f_{B}(s)$, then Cauchy Schwarz saturates and we are left with:
\begin{equation}
    S^{\text{th}}_{A}S^{\text{th}}_{B} \geq K_{g}^{2}+(S^{\text{th}}_{AB})^2
\end{equation}
The same condition as before, except now it is necessary and sufficient for separability under the restricted dynamics.

\section{III. Derivation of Rank-1 Interaction Born Markov Equation}
\label{rank 1 bm appendix}

We start with the Hamiltonian in \eqref{rank 1 interaction hamiltonian} and proceed as in \cite{schlosshauer2007decoherence}. Define  $\hat{\bm{S}}_{\alpha} = \hat{\bm{\xi}}^{T}_{\alpha}\textbf{u}_{\alpha}$, so that the interaction picture operator can be expressed as:
\begin{equation}
    \hat{\bm{S}}_{\alpha,I}(t) = e^{i\hat{H}_{\text{sys}}t}\hat{\bm{S}}_{\alpha}e^{-i\hat{H}_{\text{sys}}t} = \textbf{u}_{\alpha}^{T}e^{i\hat{H}_{\text{sys}}t}(\hat{\bm{\xi}}_{\alpha})e^{-i\hat{H}_{\text{sys}}t} = \textbf{u}^{T}_{\alpha}\mathcal{P}_{\alpha}e^{(\bm{M}+K_{g}\bm{C})t}\hat{\bm{\xi}}(0)
\end{equation}
where $\mathcal{P}_{\alpha}$ projects onto the $\alpha$ subspace. We next evaluate the environment self-correlation function:
\begin{equation}
    \mathcal{C}_{\alpha\beta}(\tau) = \langle \hat{F}^{\text{th}}_{\alpha,I}(\tau)\hat{F}^{\text{th}}_{\beta}\rangle=\frac{1}{2}\langle\{\hat{F}_{\alpha,I}(\tau),\hat{F}_{\beta}(0)\}\rangle+\frac{1}{2}\langle[\hat{F}_{\alpha,I}(\tau),\hat{F}_{\beta}(0)]\rangle\equiv\nu_{\alpha\beta}(\tau) +i\mu_{\alpha\beta}(\tau) 
\end{equation}

We then define:
\begin{equation}
    \begin{split}
        \hat{\bm{B}}_{\alpha}&\equiv\int_{0}^{\infty}d\tau\sum_{\beta}\mathcal{C}_{\alpha\beta}(\tau)\hat{\bm{S}}_{\beta,I}(-\tau)\\
        \hat{\bm{C}}_{\alpha} &\equiv \int_{0}^{\infty}d\tau\sum_{\beta}\mathcal{C}_{\beta\alpha}(-\tau)\hat{\bm{S}}_{\beta,I}(-\tau)
    \end{split}
\end{equation}
Plugging in gives:
\begin{equation}
    \begin{split}
        \hat{\bm{B}}_{\alpha} &= \int_{0}^{\infty}d\tau\sum_{\beta}[(\nu_{\alpha\beta}(\tau)+i\mu_{\alpha\beta}(\tau))\hat{\bm{S}}_{\beta,I}(-\tau)]\\
        \hat{\bm{C}}_{\alpha} &= \int_{0}^{\infty}d\tau\sum_{\beta}[(\nu_{\beta\alpha}(-\tau)+i\mu_{\beta\alpha}(-\tau))\hat{\bm{S}}_{\beta,I}(-\tau)]
    \end{split}
\end{equation}
We assume stationarity of the bath state (i.e. $[\hat{H}_{\text{bath}},\rho_{\text{bath}}]=0$), we have that:
\begin{equation}
    \langle \hat{F}_{\alpha,I}(t)\hat{F}_{\beta}(0)\rangle = \langle \hat{F}_{\alpha,I}(t+s)\hat{F}_{\beta,I}(s)\rangle \Longrightarrow \langle \hat{F}_{\alpha,I}(t)\hat{F}_{\beta}(0) \rangle = \langle \hat{F}_{\alpha}(0)\hat{F}_{\beta,I}(-t)\rangle 
\end{equation}
Now notice that 
\begin{equation}
    \nu_{\beta\alpha}(-\tau) \equiv \frac{1}{2}\langle\{\hat{F}_{\beta,I}(-\tau),\hat{F}_{\alpha}(0)\}\rangle = \frac{1}{2}\langle\{\hat{F}_{\alpha}(0),\hat{F}_{\beta,I}(-\tau)\}\rangle = \frac{1}{2}\langle\{\hat{F}_{\alpha,I}(\tau),\hat{F}_{\beta}(0)\}\rangle = \nu_{\alpha\beta}(\tau)
\end{equation}
where the second equality follows by symmetry of anticommutator and the third equality follows by stationarity. Assuming a white noise spectrum gives:
\begin{equation}
    \nu_{\alpha\beta}(\tau) = \delta(\tau)\delta_{\alpha\beta}S^{\text{th}}_{\alpha}
\end{equation}
A similar derivation yields:
\begin{equation}
    \mu_{\beta\alpha}(-\tau) = -\mu_{\alpha\beta}(\tau)
\end{equation}
Now define:
\begin{equation}
    \begin{split}
        \hat{\bm{N}}_{\alpha} &\equiv \frac{\hat{\bm{B}}_{\alpha}+\hat{\bm{C}}_{\alpha}}{2}=\int_{0}^{\infty}d\tau\sum_{\beta}\nu_{\alpha\beta}(\tau)\hat{\bm{S}}_{\beta,I}(-\tau)\\
        \hat{\bm{D}}_{\alpha} &\equiv \frac{\hat{\bm{B}}_{\alpha}-\hat{\bm{C}}_{\alpha}}{2i} = \int_{0}^{\infty}d\tau\sum_{\beta}\mu_{\alpha\beta}(\tau)\hat{\bm{S}}_{\beta,I}(-\tau)
    \end{split}
\end{equation}
so that $\hat{\bm{B}}_{\alpha} = \hat{\bm{N}}_{\alpha}+i\hat{\bm{D}}_{\alpha}, \hat{\bm{C}}_{\alpha} = \hat{\bm{N}}_{\alpha}-i\hat{\bm{D}}_{\alpha}$. We thus have in the Schrodinger picture:
\begin{equation}
\label{full rank 1 BM}
    \begin{split}
        \frac{d}{dt}\hat{\bm{\rho}}_{S}(t) &= -i[\hat{H}_{\text{sys}},\hat{\bm{\rho}}_{S}(t)]-\sum_{\alpha}\left\{[\hat{\bm{S}}_{\alpha},\hat{\bm{B}}_{\alpha}\hat{\bm\rho}_{S}]+[\hat{\bm{\rho}}_{S}\hat{\bm{C}}_{\alpha},\hat{\bm{S}}_{\alpha}]\right\}\\
        &=-i[\hat{H}_{\text{sys}},\hat{\bm{\rho}}_{S}(t)]-\sum_{\alpha}\left([\hat{\bm{S}}_{\alpha},\hat{\bm{N}}_{\alpha}\hat{\bm{\rho}}_{S}]+[\hat{\bm{\rho}}_{S}\hat{\bm{N}}_{\alpha},\hat{\bm{S}}_{\alpha}]\right)-i\sum_{\alpha}\left([\hat{\bm{S}}_{\alpha},\hat{\bm{D}}_{\alpha}\hat{\bm{\rho}}_{S}]-[\hat{\bm{\rho}}_{S}\hat{\bm{D}}_{\alpha},\hat{\bm{S}}_{\alpha}]\right)\\
        &=-i[\hat{H}_{\text{sys}},\hat{\bm{\rho}}_{S}(t)]-\sum_{\alpha}[\hat{\bm{S}}_{\alpha},[\hat{\bm{N}}_{\alpha},\hat{\bm{\rho}}_{S}]]-i\sum_{\alpha}[\hat{\bm{S}}_{\alpha},\{\hat{\bm{D}}_{\alpha},\hat{\bm{\rho}}_{S}\}]\\
        &=-i[\hat{H}_{HO},\hat{\bm{\rho}}_{S}(t)]-iK_{g}[\hat{\bm{\xi}}^{T}_{A}\textbf{u}_{A}\textbf{u}_{B}^{T}\hat{\bm{\xi}}_{B},\hat{\bm{\rho}}_{S}(t)]+\sum_{\alpha=A,B}S^{\text{th}}_{\alpha}\mathcal{D}[\hat{\bm{\xi}}^{T}_{\alpha}\textbf{u}_{\alpha}]\hat{\bm{\rho}}_{S}(t) -i\sum_{\alpha}[\hat{\bm{S}}_{\alpha},\{\hat{\bm{D}}_{\alpha},\hat{\bm{\rho}}_{S}(t)\}]
    \end{split}
\end{equation}
where $\hat{H}_{HO}=\frac{1}{2}\hat{\bm{\xi}}^{T}\bm{\Omega}^{-1}\bm{M}\hat{\bm{\xi}}$ and we have defined the dissipator:
\begin{equation}
    \mathcal{D}[\hat{\bm{\xi}}^{T}_{\alpha}\textbf{u}_{\alpha}]\hat{\bm{\rho}}_{S}(t) = \hat{\bm{\xi}}^{T}_{\alpha}\textbf{u}_{\alpha}\hat{\bm{\rho}}_{S}(t)\hat{\bm{\xi}}^{T}_{\alpha}\textbf{u}_{\alpha}-\frac{1}{2}\{(\hat{\bm{\xi}}^{T}_{\alpha}\textbf{u}_{\alpha})^2,\hat{\bm{\rho}}_{S}(t)\}
\end{equation}
The mixed commutator term is referred to as the damping term and will be dropped in all of our analysis. Further details are given in Appendix V.

\section{IV. Full GKSL Proof for Rank 1 Interaction}
\label{rank 1 GKSL proof appendix}

Here we provide the full proof of Theorem \ref{rank 1 theorem}. The goal is to construct an LOCC protocol whose Lindbladian reproduces the Born-Markov Master equation \eqref{rank 1 BM}. To that end, consider a symmetric measurement-and-feedback protocol where Alice and Bob individually implements local Gaussian weak measurements of their respective local operators $\hat{X}_{A}$ and $\hat{X}_{B}$ with measurement strengths $\gamma_{A},\gamma_{B} > 0$ and obtain classical measurement outcomes $y_{A},y_{B}\in\mathbb{R}$.
Define the corresponding Kraus operators by $\hat{K}_{A}(y_{A}), \hat{K}_{B}(y_{B})$:
\begin{equation}
    \begin{split}
        \hat{K}_{A}(y_{A}) &\equiv \left(\frac{2\gamma_{A}dt}{\pi}\right)^{1/4}\exp\left[-\gamma_{A}dt(y_{A}-\hat{X}_{A})^2\right] \equiv N_{A}\exp\left[-\gamma_{A}dt(y_{A}-\hat{X}_{A})^2\right]\\
        \hat{K}_{B}(y_{B}) &\equiv \left(\frac{2\gamma_{B}dt}{\pi}\right)^{1/4}\exp\left[-\gamma_{B}dt(y_{B}-\hat{X}_{B})^2\right]\equiv N_{B}\exp\left[-\gamma_{B}dt(y_{B}-\hat{X}_{B})^2\right]
    \end{split}
\end{equation}
Given the classical outcomes, we apply local unitaries $\hat{U}_{A}(y_{B}), \hat{U}_{B}(y_{A})$ conditioned on the other party's outcome with feedback strengths $\lambda_{A},\lambda_{B}\in\mathbb{R}$:
\begin{equation}
    \begin{split}
            \hat{U}_{A}(y_{B}) &\equiv \exp\left(-i\lambda_{A}y_{B}\hat{X}_{A}dt\right)\\
        \hat{U}_{B}(y_{A}) &\equiv \exp\left(-i\lambda_{B}y_{A}\hat{X}_{B}dt\right)\\
    \end{split}
\end{equation}
This step requires only local communication of classical results. The overall Kraus operator conditioned on $(y_{A},y_{B})$ is then defined via:
\begin{equation}
       \hat{K}(y_{A},y_{B}) \equiv [\hat{U}_{A}(y_{B})\hat{K}_{A}(y_{A})]\otimes[\hat{U}_{B}(y_{A})\hat{K}_{B}(y_{B})]
\end{equation}

Now consider the following unconditional map corresponding to averaging over the classical outcomes of the measurement-feedback scheme:
\begin{equation}
    \mathcal{E}(\hat{\rho}) = \int dy_{A}dy_{B}\hat{K}(y_{A},y_{B})\hat{\rho} \hat{K}^{\dagger}(y_{A},y_{B})
\end{equation}
The necessary Gaussian integrals we need are:
\begin{equation}
    \int_{-\infty}^{\infty}N_{A}^2e^{-2\gamma_{A}dt y^2}dy =N_{A}^2\sqrt{\frac{\pi}{2\gamma_{A}dt}}=1, \hspace{0.5cm} \int_{-\infty}^{\infty}N_{A}^2y^2e^{-2\gamma_{A}dt y^2}dy = N_{A}^2\frac{1}{2}\sqrt{\frac{\pi}{(2\gamma_{A} dt)^3}}=\frac{1}{4\gamma_{A}dt}
\end{equation}
Notably, this means a term proportional to $y^2dt^2$ can still be $\mathcal{O}(dt)$ after integration so we must keep those terms in our series expansions.

Now perform the following decomposition and expansions:
\begin{equation}
\begin{split}
    \hat{K}_{A}(y_{A}) &= N_{A}e^{-\gamma_{A}dty_{A}^2}e^{2\gamma_{A}dty_{A}\hat{X}_{A}}e^{-\gamma_{A}dt\hat{X}_{A}^2}\\
    &\approx N_{A}e^{-\gamma_{A}dty_{A}^2}\left(1+2\gamma_{A}dty_{A}\hat{X}_{A}+\frac{(2\gamma_{A}dty_{A}\hat{X}_{A})^2}{2}\right)\left(1-\gamma_{A}dt\hat{X}_{A}^2\right)\\
    &=N_{A}e^{-\gamma_{A}dt y_{A}^2}\left(1-\gamma_{A}dt\hat{X}_{A}^2+2\gamma_{A}dt y_{A}\hat{X}_{A}+2\gamma_{A}^{2}dt^2y_{A}^{2}\hat{X}_{A}^2\right)\\
    \hat{U}_{A}(y_{B}) &\approx 1-i\lambda_{A}y_{B}\hat{X}_{A}dt-\frac{1}{2}\lambda_{A}^{2}y_{B}^2\hat{X}_{A}^2dt^2
\end{split}
\end{equation}
Taking the product gives:
\begin{equation}
    \begin{split}
       &[\hat{U}_{A}(y_{B})\hat{K}_{A}(y_{A})] \\
       &\approx N_{A}e^{-\gamma_{A}dt y_{A}^{2}}\left[1+(2\gamma_{A}dty_{A}\hat{X}_{A}-i\lambda_{A}y_{B}\hat{X}_{A}dt)+\left(-\gamma_{A}dt\hat{X}_{A}^{2}+2\gamma_{A}^{2}dt^{2}y_{A}^{2}\hat{X}_{A}^{2}-\frac{1}{2}\lambda_{A}^2y_{B}^{2}\hat{X}_{A}^{2}dt^2\right)\right]\\
        &\equiv N_{A}e^{-\gamma_{A}dt y_{A}^{2}}\left[1+\hat{A}^{(1)}+\hat{A}^{(2)}\right]
    \end{split}
\end{equation}
Similarly, we have:
\begin{equation}
    \begin{split}
       &[\hat{U}_{B}(y_{A})\hat{K}_{B}(y_{B})] \\
       &\approx N_{B}e^{-\gamma_{B}dt y_{B}^{2}}\left[1+(2\gamma_{B}dty_{B}\hat{X}_{B}-i\lambda_{B}y_{A}\hat{X}_{B}dt)+\left(-\gamma_{B}dt\hat{X}_{B}^{2}+2\gamma_{B}^{2}dt^{2}y_{B}^{2}\hat{X}_{B}^{2}-\frac{1}{2}\lambda_{B}^2y_{A}^{2}\hat{X}_{B}^{2}dt^2\right)\right]\\
        &\equiv N_{B}e^{-\gamma_{B}dt y_{B}^{2}}\left[1+\hat{B}^{(1)}+\hat{B}^{(2)}\right]
    \end{split}
\end{equation}
Note that the cross terms in the products either integrate to $0$ or are of a higher order in $dt$.

Thus the full Kraus operator is:
\begin{equation}
    \hat{K}(y_{A},y_{B}) = N_{A}N_{B}e^{-\gamma_{A}dty_{A}^{2}-\gamma_{B}dty_{B}^{2}}\left(1+\hat{A}^{(1)}+\hat{A}^{(2)}\right)\otimes\left(1+\hat{B}^{(1)}+\hat{B}^{(2)}\right)
\end{equation}
so that the integrand in the channel definition can be rewritten as:
\begin{equation}
\begin{split}
    &\hat{K}(y_{A},y_{B})\hat{\rho} \hat{K}^{\dagger}(y_{A},y_{B}) \\
    &= N_{A}^{2}N_{B}^{2}e^{-2\gamma_{A}dty_{A}^{2}-2\gamma_{B}dty_{B}^{2}}\left(1+\hat{A}^{(1)}+\hat{A}^{(2)}\right)\otimes\left(1+\hat{B}^{(1)}+\hat{B}^{(2)}\right)\hat{\rho}\left(1+\hat{A}^{(1)}+\hat{A}^{(2)}\right)^{\dagger}\otimes\left(1+\hat{B}^{(1)}+\hat{B}^{(2)}\right)^{\dagger}
\end{split}
\end{equation}

From normalization, the product of identities immediately gives:
\begin{equation}
    \int dy_{A}dy_{B} N_{A}^{2}N_{B}^{2}e^{-2\gamma_{A}dty_{A}^{2}-2\gamma_{B}dty_{B}^{2}}\hat{\rho} = \hat{\rho}
\end{equation}
We then have the contribution from the first term of $\hat{A}^{(1)}\hat{\rho} \hat{A}^{(1)\dagger}$:
\begin{equation}
     \int dy_{A}dy_{B} N_{A}^{2}N_{B}^{2}e^{-2\gamma_{A}dty_{A}^{2}-2\gamma_{B}dty_{B}^{2}}4\gamma_{A}^{2}dt^{2}y_{A}^{2}\hat{X}_{A}\hat{\rho}\hat{X}_{A} = \gamma_{A}dt\hat{X}_{A}\hat{\rho}\hat{X}_{A}
\end{equation}
Meanwhile, from the $-\gamma_{A}dt\hat{X}_{A}^{2}+2\gamma_{A}^{2}dt^2y_{A}^2\hat{X}_{A}^{2}$ terms in $\hat{A}^{(2)}$, we get: 
\begin{equation}
    \hat{A}^{(2)}\hat{\rho}+\hat{\rho} \hat{A}^{(2)\dagger} = -\gamma_{A}dt(\hat{X}_{A}^{2}\hat{\rho}+\hat{\rho}\hat{X}_{A}^{2})+\frac{1}{2}\gamma_{A}dt\left(\hat{X}_{A}^{2}\hat{\rho}+\hat{\rho}\hat{X}_{A}^{2}\right) = -\frac{\gamma_{A}}{2}dt\{\hat{X}_{A}^{2},\hat{\rho}\}
\end{equation}
Combining all the above contributions gives $\gamma_{A}dt\mathcal{D}[\hat{X}_{A}]\hat{\rho}$. There are analogous $B$ subsystem terms which contribute $\gamma_{B}dt\mathcal{D}[\hat{X}_{B}]\hat{\rho}$.

Now pay attention to the contribution of $-i\lambda_{A}y_{B}\hat{X}_{A}dt$ from the $\hat{A}^{(1)}\hat{\rho} \hat{A}^{(1)\dagger}$ term:
\begin{equation}
    \int dy_{A}dy_{B}N_{A}^2N_{B}^{2}e^{-2\gamma_{A}dty_{A}^{2}-2\gamma_{B}dty_{B}^{2}}(-i\lambda_{A}y_{B}dt)(i\lambda_{A}y_{B}dt)\hat{X}_{A}\hat{\rho}\hat{X}_{A} = \frac{\lambda_{A}^{2}}{4\gamma_{B}}dt\hat{X}_{A}\hat{\rho}\hat{X}_{A}
\end{equation}
Meanwhile, from the $\hat{A}^{(2)}\hat{\rho}+\hat{\rho} \hat{A}^{(2)\dagger}$ piece, we have
\begin{equation}
    -\frac{1}{8}\frac{\lambda_{A}^{2}}{\gamma_{B}}dt(\hat{X}_{A}^{2}\hat{\rho}+\hat{\rho}\hat{X}_{A}^{2})
\end{equation}
Combining gives $(\lambda_{A}^{2}/(4\gamma_{B}))dt\mathcal{D}[\hat{X}_{A}]\hat{\rho}$ and similarly from the $B$ pieces $(\lambda_{B}^{2}/(4\gamma_{A}))dt\mathcal{D}[\hat{X}_{B}]\hat{\rho}$.

Finally, we have the cross terms. From $(\hat{A}^{(1)}\otimes \hat{B}^{(1)})\hat{\rho}$ we get:
\begin{equation}
    (2\gamma_{A}dty_{A}\hat{X}_{A})(-i\lambda_{B}y_{A}\hat{X}_{B}dt)+(-i\lambda_{A}y_{B}\hat{X}_{A}dt)(2\gamma_{B}dty_{B}\hat{X}_{B})\overset{\text{Integrate}}{\Longrightarrow}-i\frac{\lambda_{B}}{2}dt\hat{X}_{A}\hat{X}_{B}\hat{\rho}-i\frac{\lambda_{A}}{2}dt\hat{X}_{A}\hat{X}_{B}\hat{\rho}
\end{equation}
Adding the Hermitian conjugate yields $-i\frac{\lambda_{A}+\lambda_{B}}{2}[\hat{X}_{A}\hat{X}_{B},\hat{\rho}]$. Finally, from the $\hat{A}^{(1)}\hat{\rho} \hat{B}^{(1)\dagger}$ terms, we have:
\begin{equation}
    (2\gamma_{A}dty_{A}\hat{X}_{A})\hat{\rho}(i\lambda_{B}y_{A}\hat{X}_{B}dt) + (-i\lambda_{A}y_{B}\hat{X}_{A}dt)\hat{\rho}(2\gamma_{B}dty_{B}\hat{X}_{B})\overset{\text{Integrate}}{\Longrightarrow} \frac{i}{2}dt\lambda_{B}\hat{X}_{A}\hat{\rho}\hat{X}_{B}-\frac{i}{2}dt\lambda_{A}\hat{X}_{A}\hat{\rho}\hat{X}_{B}
\end{equation}
Adding the other symmetric contribution with $A\leftrightarrow B$ gives $\frac{i}{2}(\lambda_{B}-\lambda_{A})(\hat{X}_{A}\hat{\rho}\hat{X}_{B}-\hat{X}_{B}\hat{\rho}\hat{X}_{A})$.

Putting everything together gives:
\begin{equation}
    \mathcal{E}(\hat{\rho}) = \hat{\rho} + dt\left\{-\frac{i}{2}(\lambda_{A}+\lambda_{B})[\hat{X}_{A}\hat{X}_{B},\hat{\rho}]+\frac{i}{2}(\lambda_{B}-\lambda_{A})(\hat{X}_{A}\hat{\rho}\hat{X}_{B}-\hat{X}_{B}\hat{\rho}\hat{X}_{A})+\Gamma'_{A}\mathcal{D}[\hat{X}_{A}]\hat{\rho}+\Gamma'_{B}\mathcal{D}[\hat{X}_{B}]\hat{\rho}\right\}
\end{equation}
where
\begin{equation}
    \begin{split}
        \Gamma'_{A} &\equiv \gamma_{A}+\frac{\lambda_{A}^2}{4\gamma_{B}}\\
        \Gamma'_{B} &\equiv \gamma_{B}+\frac{\lambda_{B}^2}{4\gamma_{A}}
    \end{split}
\end{equation}

Now pick $\lambda_{A} = \lambda_{B} \equiv \lambda$, this gives the master equation:
\begin{equation}
\label{rank one LOCC res}
    \dot{\hat{\rho}} = -i\lambda[\hat{X}_{A}\hat{X}_{B},\hat{\rho}]+\left(\gamma_{A}+\frac{\lambda^2}{4\gamma_{B}}\right)\mathcal{D}[\hat{X}_{A}]\hat{\rho} + \left(\gamma_{B}+\frac{\lambda^2}{4\gamma_{A}}\right)\mathcal{D}[\hat{X}_{B}]\hat{\rho}
\end{equation}
Matching to \eqref{rank 1 BM}, we require that $\lambda = K_{g}$ and solve the system:
\begin{equation}
    \gamma_{A}+\frac{K_{g}^2}{4\gamma_{B}} = S^{\text{th}}_{A}, \hspace{0.4cm} \gamma_{B}+\frac{K_{g}^2}{4\gamma_{A}} = S^{\text{th}}_{{B}}
\end{equation}
These admit the solutions:
\begin{equation}
    \begin{split}
        \gamma_{A} &= \frac{S^{\text{th}}_{A}}{2}\left(1\pm\sqrt{1-\frac{K_{g}^2}{S^{\text{th}}_{A}S^{\text{th}}_{B}}}\right)\\
        \gamma_{B} &= \frac{S^{\text{th}}_{B}}{2}\left(1\pm\sqrt{1-\frac{K_{g}^2}{S^{\text{th}}_{A}S^{\text{th}}_{B}}}\right)
    \end{split}
\end{equation}
We see that these admit $\gamma_{A},\gamma_{B} > 0$ as long as $S^{\text{th}}_{A}S^{\text{th}}_{B} \geq K_{g}^{2}$ which is exactly our separbility condition. The same result holds if we replace $\hat{X}_{A} = (\hat{\bm{\xi}}^{T}_{A}\textbf{u}_{A})$ and $\hat{X}_{B} = (\hat{\bm{\xi}}^{T}_{B}\textbf{u}_{B})$. The remaining terms of \eqref{rank 1 BM} are local in the $A,B$ subsystems, so we can implement them via local unitaries independent of this LOCC scheme, thus completing the construction.

\subsection{IV.a. Including Correlated Noise}
Assuming a white noise spectrum of correlated noise gives:
\begin{equation}
\langle \hat{F}^{\text{th}}_{\alpha,I}(\tau)\hat{F}^{\text{th}}_{\beta}\rangle\approx\frac{1}{2}\langle\{\hat{F}_{\alpha,I}(\tau),\hat{F}_{\beta}(0)\}\rangle = \delta(\tau)\Sigma^{\text{th}}_{\alpha,\beta}
\end{equation}
where
\begin{equation}
    \Sigma^{\text{th}} = \begin{pmatrix}
        S^{\text{th}}_{A} & S^{\text{th}}_{AB}\\
        S^{\text{th}}_{AB} & S^{\text{th}}_{B}
    \end{pmatrix}
\end{equation}
and we dropped the damping term as per our general assumptions. Following the same notation and calculation as Appendix III,
the only change now is that:
\begin{equation}
    \hat{\bm{N}}_{A} = \frac{1}{2}(S^{\text{th}}_{A}\hat{\bm{S}}_{A}+S^{\text{th}}_{AB}\hat{\bm{S}}_{B}), \hat{\bm{N}}_{B} = \frac{1}{2}(S^{\text{th}}_{AB}\hat{\bm{S}}_{A}+S^{\text{th}}_{B}\hat{\bm{S}}_{B})
\end{equation}
The Born Markov Master equation now becomes:
\begin{equation}
\label{correlated rank 1 BM}
    \frac{d}{dt}\hat{\bm{\rho}}_{S}(t) = -i[\hat{H}_{HO},\hat{\bm{\rho}}_{S}(t)]-iK_{g}[\hat{\bm{\xi}}^{T}_{A}\textbf{u}_{A}\textbf{u}_{B}^{T}\hat{\bm{\xi}}_{B},\hat{\bm{\rho}}_{S}(t)]+\sum_{\alpha=A,B}S^{\text{th}}_{\alpha}\mathcal{D}[\hat{\bm{\xi}}^{T}_{\alpha}\textbf{u}_{\alpha}]\hat{\bm{\rho}}_{S}(t) - S^{\text{th}}_{AB}[\hat{\bm{\xi}}^{T}_{A}\textbf{u}_{A},[\hat{\bm{\xi}}^{T}_{B}\textbf{u}_{B},\hat{\bm{\rho}}_{S}(t)]]
\end{equation}
where we have dropped the damping term to have a cleaner/more general discussion about the effects of correlated noise. Compared to the uncorrelated case, we see the appearance of a new $S^{\text{th}}_{AB}$ term that will need to be accounted for with a modified LOCC scheme.

It turns out the only modification to our LOCC scheme required is to allow both parties to also feedback on their own measurements, i.e.:
\begin{equation}
    \begin{split}
            \hat{U}_{A}(y_{A},y_{B}) &\equiv \exp\left(-i(\kappa_{A}y_{A} +\lambda y_{B})\hat{X}_{A}dt\right)\\
        \hat{U}_{B}(y_{A},y_{B}) &\equiv \exp\left(-i(\kappa_{B}y_{B}+\lambda y_{A})\hat{X}_{B}dt\right)\\
    \end{split}
\end{equation}
where $\kappa_{A},\kappa_{B}$ are the self-feedback strengths, and we have apriori set $\lambda_{A} = \lambda_{B} = \lambda$ since this was already a requirement in the uncorrelated problem. The local weak measurement Kraus operators $\hat{K}_{A}(y_{A}), \hat{K}_{A}(y_{B})$ stay the same as before. Going through the same calculation procedure as before, we ultimately find that:
\begin{equation}
    \mathcal{E}(\hat{\rho}) = \hat{\rho}+dt\left\{-i\lambda[\hat{X}_{A}\hat{X}_{B},\hat{\rho}]-i\frac{\kappa_{A}}{2}[\hat{X}^{2}_{A},\hat{\rho}]-i\frac{\kappa_{B}}{2}[\hat{X}^{2}_{B},\hat{\rho}]+\Gamma_{A}\mathcal{D}[\hat{X}_{A}]\hat{\rho}+\Gamma_{B}\mathcal{D}[\hat{X_{B}}]\hat{\rho}-\Gamma_{AB}[\hat{X}_{A},[\hat{X}_{B},\hat{\rho}]]\right\}
\end{equation}
with
\begin{equation}
    \Gamma_{A} = \gamma_{A}+\frac{\lambda^2}{4\gamma_{B}}+\frac{\kappa_{A}^2}{4\gamma_{A}}, \hspace{0.5cm} \Gamma_{B} = \gamma_{B}+\frac{\lambda^2}{4\gamma_{A}}+\frac{\kappa_{B}^2}{4\gamma_{B}}, \hspace{0.5cm} \Gamma_{AB} = \frac{\lambda\kappa_{A}}{4\gamma_{A}}+\frac{\lambda\kappa_{B}}{4\gamma_{B}}
\end{equation}
Matching to \eqref{correlated rank 1 BM} now requires $\lambda = K_{g}$ and:
\begin{equation}
\begin{aligned}
    S^{\text{th}}_{A} &=  \gamma_{A}+\frac{K_{g}^2}{4\gamma_{B}}+\frac{\kappa_{A}^2}{4\gamma_{A}}\\
    S^{\text{th}}_{B} &= \gamma_{B}+\frac{K_{g}^2}{4\gamma_{A}}+\frac{\kappa_{B}^2}{4\gamma_{B}}\\
    S^{\text{th}}_{AB} &= \frac{K_{g}\kappa_{A}}{4\gamma_{A}}+\frac{K_{g}\kappa_{B}}{4\gamma_{B}}
\end{aligned}
\end{equation}

Note that there are additional $[\hat{X}^2,\hat{\rho}]$ terms that result from this LOCC scheme. These can be cancelled separately via the local unitaries $\hat{U}_{i} = \exp\left(i\frac{\kappa_{i}}{2}\hat{X}_{i}^{2}dt\right)$ and can thus be disregarded in our analysis.

Now pick
\begin{equation}
    \kappa_{A} = 2\gamma_{A}\frac{S^{\text{th}}_{AB}}{K_{g}}, \hspace{0.5cm} \kappa_{B} = 2\gamma_{B}\frac{S^{\text{th}}_{AB}}{K_{g}}
\end{equation}
The $S^{\text{th}}_{AB}$ equation is trivially satisfied and the remaining two equations become:
\begin{equation}
    \begin{aligned}
        S^{\text{th}}_{A} &= \gamma_{A}\left(1+\frac{(S^{\text{th}}_{AB})^2}{K_{g}^{2}}\right)+\frac{K_{g}^{2}}{4\gamma_{B}} = \tilde{\gamma}_{A}+\frac{K_{g}^{2}+(S^{\text{th}}_{AB})^2}{4\tilde{\gamma}_{B}}\\
        S^{\text{th}}_{B} &= \gamma_{B}\left(1+\frac{(S^{\text{th}}_{AB})^2}{K_{g}^{2}}\right)+\frac{K_{g}^{2}}{4\gamma_{A}} = \tilde{\gamma}_{B}+\frac{K_{g}^{2}+(S^{\text{th}}_{AB})^2}{4\tilde{\gamma}_{A}}
    \end{aligned}
\end{equation}
where we defined $\tilde{\gamma}_{A} = \gamma_{A}\left(1+\frac{(S^{\text{th}}_{AB})^2}{K_{g}^{2}}\right), \tilde{\gamma}_{B} = \gamma_{B}\left(1+\frac{(S^{\text{th}}_{AB})^2}{K_{g}^{2}}\right)$. We can solve these simultaneous equations:
\begin{equation}
    \begin{aligned}
                \tilde{\gamma}_{A} &= \frac{S^{\text{th}}_{A}}{2}\left(1\pm\sqrt{1-\frac{K_{g}^2+(S^{\text{th}}_{AB})^2}{S^{\text{th}}_{A}S^{\text{th}}_{B}}}\right)\\
        \tilde{\gamma}_{B} &= \frac{S^{\text{th}}_{B}}{2}\left(1\pm\sqrt{1-\frac{K_{g}^2+(S^{\text{th}}_{AB})^2}{S^{\text{th}}_{A}S^{\text{th}}_{B}}}\right)
    \end{aligned}
\end{equation}
Since $\gamma_{i}$ and $\tilde{\gamma}_{i}$ differ by a positive constant of proportionality, $\gamma_{A},\gamma_{B}$ admit nonnegative real solutions as long as:
\begin{equation}
    S^{\text{th}}_{A}S^{\text{th}}_{B} \geq K_{g}^{2}+(S^{\text{th}}_{AB})^2
\end{equation}
This exactly matches our result from the Gaussian correlated case.

\section{V. Discussion on Including Damping/Relaxation Term in GKSL Arguments}
\label{damping GKSL appendix}

The question now is if it is possible to construct another LOCC scheme that only reproduces the mixed commutator (friction) term in \eqref{full rank 1 BM} (i.e. we assume uncorrelated noise to keep the discussion of damping clean). If this is possible, we would then be able to string together a sequence of different local operations to reproduce the dynamics of the full master equation. It turns out that this is impossible without introducing extra dissipator terms that cannot be combined with the dissipator terms of the symmetric LOCC scheme unless a very strict criteria is met. As such, it is not possible, using our framework, to derive a separability condition in the presence of friction in full generality.

This time, instead of a symmetric LOCC scheme, we use a one-way scheme which consists of a local Gaussian weak measurement of $\hat{D}$ with measurement strength $\gamma> 0$ and obtain measurement outcomes $y\in\mathbb{R}$.
Define the corresponding Kraus operator by
\begin{equation}
       K(y) = \left(\frac{2\gamma dt}{\pi}\right)^{1/4}\exp\left[-\gamma dt(y-\hat{D})^2\right]
\end{equation}
We apply a local unitary $\hat{U}(y)$ with feed-forward gain $\lambda$ conditioned on the classical outcome via:
\begin{equation}
    \hat{U}(y) = \exp(-i\lambda y\hat{S} dt)
\end{equation}
where $\hat{S}$ generates the feed-forward unitary and will be identified specifically later.
The overall map describing this LOCC + unconditioned projection scheme is then:
\begin{equation}
    \mathcal{E}(\hat{\rho}) = \int dy[\hat{U}(y)\hat{K}(y)]\hat{\rho}[\hat{U}(y)\hat{K}(y)]^{\dagger}
\end{equation}
We can expand the product as:
\begin{equation}
    \begin{split}
        [\hat{U}(y)\hat{K}(y)]&\approx Ne^{-\gamma dt y^{2}}\Bigg[1+(2\gamma dt y\hat{D}-i\lambda y\hat{S}dt)+\left(-\gamma dt \hat{D}^2+2\gamma^2 dt^2y ^{2}\hat{D}^{2}-\frac{1}{2}\lambda^{2}y^2\hat{S}^{2}dt^2-2i\gamma\lambda y^{2}dt^2\hat{S}\hat{D}\right)\Bigg]\\
        & \equiv N_{B}e^{-\gamma_{B}dt y_{B}^{2}}[1+\hat{A}+\hat{B}]
    \end{split}
\end{equation}
Notice the appearance of a new term proportional to $\gamma\lambda$ in $B$ here. This was not present in the symmetric measurement-feedback process and is unique to the one-sided process. Thus, the relevant terms that will integrate to linear order in $dt$ are:
\begin{equation}
[\hat{U}(y)\hat{K}(y)]\hat{\rho}[\hat{U}(y)\hat{K}(y)]^{\dagger}
= N^{2}e^{-2\gamma dty^{2}}\left(\hat{\rho}+\hat{A}\hat{\rho}+\hat{\rho}\hat{A}^{\dagger}+\hat{A}\hat{\rho}\hat{A}^{\dagger}+\hat{B}\hat{\rho}+\hat{\rho}\hat{B}^{\dagger}\right)
\end{equation}
The odd in $y$ terms vanish after the Gaussian integral, so the only nontrivial pieces come from $\hat{A}\hat{\rho} \hat{A}^{\dagger}$ and $\hat{B}\hat{\rho} + \hat{\rho} \hat{B}^{\dagger}$:
\begin{equation}
    \begin{split}
      \hat{A}\hat{\rho} \hat{A}^{\dagger}&: \gamma dt\hat{D}\hat{\rho}\hat{D} + \frac{\lambda^{2}}{4\gamma}dt \hat{S}\hat{\rho} \hat{S}+\frac{i\lambda}{2}dt(\hat{D}\hat{\rho} \hat{S}- \hat{S}\hat{\rho} \hat{D})\\
       \hat{B}\hat{\rho} + \hat{\rho} \hat{B}^{\dagger}&: -\frac{i\lambda}{2}dt \hat{S}\hat{D}\hat{\rho}+\frac{i\lambda }{2}dt \hat{\rho} \hat{D}\hat{S} +(-\gamma dt+\frac{\gamma}{2}dt)(\hat{D}^2\hat{\rho}+\hat{\rho} \hat{D}^2)-\frac{\lambda^2}{8\gamma}dt(\hat{S}^2\hat{\rho}+\hat{\rho} \hat{S}^2)
    \end{split}
\end{equation}
Now add everything together gives:
\begin{equation}
    \mathcal{E}(\hat{\rho}) =\hat{\rho}+dt\left[ -i\frac{\lambda}{2}[\hat{S},\{\hat{D},\hat{\rho}\}]+\gamma \mathcal{D}[\hat{D}]\hat{\rho} +  \frac{\lambda^2}{4\gamma}\mathcal{D}[\hat{S}]\hat{\rho}\right]
\end{equation}
where $\mathcal{D}[\hat{O}]\hat{\rho} \equiv \hat{O}\hat{\rho}\hat{O}-\frac{1}{2}\{\hat{O}^2,\hat{\rho}\}$. Thus, using the one sided LOCC process produces dynamics that contain anti-commutators.

However, we still cannot directly match this to the friction term in \eqref{full rank 1 BM} yet because $\hat{\bm{D}}_{\alpha}$ is a non-local operator that spans across subsystems $A$ and $B$. More explicitly, write:
\begin{equation}
    e^{-(\bm{M}+K_{g}\bm{C})\tau} \equiv \begin{pmatrix}
        \bm{\Phi}_{AA}(\tau) & \bm{\Phi}_{AB}(\tau)\\
        \bm{\Phi}_{BA}(\tau) & \bm{\Phi}_{BB}(\tau)
    \end{pmatrix}
\end{equation}
Then we have:
\begin{equation}
    \begin{split}
        \hat{\bm{S}}_{A}(-\tau) &= \textbf{u}^{T}_{A}\bm{\Phi}_{AA}(\tau)\hat{\bm{\xi}}_{A} + \textbf{u}^{T}_{A}\bm{\Phi}_{AB}(\tau)\hat{\bm{\xi}}_{B}\\
        \hat{\bm{S}}_{B}(-\tau) &= \textbf{u}^{T}_{B}\bm{\Phi}_{BA}(\tau)\hat{\bm{\xi}}_{A}+\textbf{u}^{T}_{B}\bm{\Phi}_{BB}(\tau)\hat{\bm{\xi}}_{B}
    \end{split}
\end{equation}
Plugging back in gives:
\begin{equation}
    \begin{split}
        \hat{\bm{D}}_{A} &= \int_{0}^{\infty}d\tau\mu_{AA}(\tau)(\textbf{u}^{T}_{A}\bm{\Phi}_{AA}(\tau)\hat{\bm{\xi}}_{A} + \textbf{u}^{T}_{A}\bm{\Phi}_{AB}(\tau)\hat{\bm{\xi}}_{B})+\int_{0}^{\infty}d\tau \mu_{AB}(\tau)(\textbf{u}^{T}_{B}\bm{\Phi}_{BA}(\tau)\hat{\bm{\xi}}_{A}+\textbf{u}^{T}_{B}\bm{\Phi}_{BB}(\tau)\hat{\bm{\xi}}_{B})
    \end{split}
\end{equation}
Thus, we can write $\hat{\bm{D}}_{A} = \hat{\bm{D}}_{AA}+\hat{\bm{D}}_{AB}$ where:
\begin{equation}
    \begin{split}
        \hat{\bm{D}}_{AA} &= \left[\int_{0}^{\infty}d\tau(\mu_{AA}(\tau)\textbf{u}^{T}_{A}\bm{\Phi}_{AA}(\tau)+\mu_{AB}(\tau)\textbf{u}^{T}_{B}\bm{\Phi}_{BA}(\tau))\right]\hat{\bm{\xi}}_{A}\\
        \hat{\bm{D}}_{AB} &=\left[\int_{0}^{\infty}d\tau(\mu_{AA}(\tau)\textbf{u}^{T}_{A}\bm{\Phi}_{AB}(\tau)+\mu_{AB}(\tau)\textbf{u}^{T}_{B}\bm{\Phi}_{BB}(\tau))\right]\hat{\bm{\xi}}_{B}
    \end{split}
\end{equation}
Similarly, $\hat{\bm{D}}_{B} = \hat{\bm{D}}_{BA}+\hat{\bm{D}}_{BB}$ where:
\begin{equation}
    \begin{split}
        \hat{\bm{D}}_{BA} &=\left[\int_{0}^{\infty}d\tau(\mu_{BA}(\tau)\textbf{u}^{T}_{A}\bm{\Phi}_{AA}(\tau)+\mu_{BB}(\tau)\textbf{u}^{T}_{B}\bm{\Phi}_{BA}(\tau))\right]\hat{\bm{\xi}}_{A}\\
        \hat{\bm{D}}_{BB} &=\left[\int_{0}^{\infty}d\tau(\mu_{BA}(\tau)\textbf{u}^{T}_{A}\bm{\Phi}_{AB}(\tau)+\mu_{BB}(\tau)\textbf{u}^{T}_{B}\bm{\Phi}_{BB}(\tau))\right]\hat{\bm{\xi}}_{B}
    \end{split}
\end{equation}
Thus, we have written each non-local term as a sum of local terms so the dissipative term in the Born-Markov (BM) equation can now be written as:
\begin{equation}
    -i\sum_{\alpha,\beta}[\hat{\bm{S}}_{\alpha},\{\hat{\bm{D}}_{\alpha\beta},\hat{\bm{\rho}}_{S}\}]
\end{equation}
which we can identify as being produced by a composition of one-sided LOCC processes with different associated parameters $\lambda_{\alpha\beta},\gamma_{\alpha\beta}$ for $\alpha,\beta\in\{A,B\}$.

Thus, if we put together the local unitary associated with $H_{HO}$, the symmetric feedback scheme producing the entangling pairwise interaction, and then this composition of one-sided schemes producing the mixed commutator, we arrive at the full master equation:
\begin{equation}
\begin{split}
    \dot{\hat{\rho}} &= -i[H_{HO},\hat{\rho}]-iK_{g}[\hat{\bm{S}}_{A}\hat{\bm{S}}_{B},\hat{\rho}]-i\sum_{\alpha}[\hat{\bm{S}}_{\alpha},\{\hat{\bm{D}}_{\alpha},\hat{\rho}\}]\\
    &+\left(\gamma_{A}+\frac{K_{g}^2}{4\gamma_{B}}\right)\mathcal{D}[\hat{\bm{S}}_{A}]\hat{\rho} + \left(\gamma_{B}+\frac{K_{g}^2}{4\gamma_{A}}\right)\mathcal{D}[\hat{\bm{S}}_{B}]\hat{\rho}+\sum_{\alpha \beta}\gamma_{\alpha\beta}\mathcal{D}[\hat{\bm{D}}_{\alpha\beta}]\hat{\rho} + \sum_{\alpha\beta}\frac{1}{\gamma_{\alpha\beta}}\mathcal{D}[\hat{\bm{S}}_{\alpha}]\hat{\rho}
\end{split}
\end{equation}
where we have already set $\lambda_{\alpha\beta} = 2$ to match the mixed commutator coefficient (we have abused notation here, parameters with single subscripts $\gamma_{\alpha}$ are associated with the symmetric LOCC scheme while parameters with two subscripts $\gamma_{\alpha\beta}$ correspond to the one-way scheme.

The only way to match this to \eqref{full rank 1 BM} is if we can somehow absorb the $\mathcal{D}[\hat{\bm{D}}_{\alpha\beta}]\hat{\rho}$ into the $\mathcal{D}[\hat{\bm{S}}_{\beta}]\hat{\rho}$ terms. Under the Markovian high-temperature limit of the Caldeira Leggett model with Ohmic damping, one has the following \cite{Ghosh_2024}:
\begin{equation}
    \nu_{\alpha\beta} \equiv \frac{1}{2}\langle\{\hat{F}_{\alpha,I}(\tau),\hat{F}_{\beta}(0)\}\rangle = C_{1}\delta_{\alpha\beta}\delta(\tau), \hspace{0.3cm} \mu_{\alpha\beta} \equiv \frac{1}{2i}\langle[\hat{F}_{\alpha,I}(\tau),\hat{F}_{\beta}(0)]\rangle = C_{2}\delta_{\alpha\beta}\frac{d}{d\tau}\delta(\tau)
\end{equation}
for $C_{1},C_{2}$ prefactors associated with the choice of Ohmic damping. We can then evaluate:
\begin{equation}
    \begin{split}
                \hat{\bm{D}}_{AA} \propto -\textbf{u}^{T}_{A}\bm{\dot{\Phi}}_{AA}(0)\hat{\bm{\xi}}_{A}, \hspace{0.2cm} \hat{\bm{D}}_{AB} \propto -\textbf{u}^{T}_{A}\bm{\dot{\Phi}}_{AB}(0)\hat{\bm{\xi}}_{B}, \hspace{0.2cm} \hat{\bm{D}}_{BA}\propto -\textbf{u}^{T}_{B}\bm{\dot{\Phi}}_{BA}(0)\hat{\bm{\xi}}_{A}, \hspace{0.2cm} \hat{\bm{D}}_{BB} \propto -\textbf{u}^{T}_{B}\bm{\dot{\Phi}}_{BB}(0)\hat{\bm{\xi}}_{B}
    \end{split}
\end{equation}
By expanding in matrix notation, we then require:
\begin{equation}
    \mathcal{D}[\textbf{v}^{T}_{A}\hat{\bm{\xi}}_{A}]\hat{\rho} +    \mathcal{D}[\textbf{u}^{T}_{A}\hat{\bm{\xi}}_{A}]\hat{\rho} = \mathcal{D}[\textbf{r}^{T}_{A}\hat{\bm{\xi}}_{A}]\hat{\rho}, \hspace{0.2cm} \textbf{r}_{A}\textbf{r}_{A}^{T} = \textbf{v}_{A}\textbf{v}_{A}^{T} + \textbf{u}_{A}\textbf{u}^{T}_{A}
\end{equation}
where $\textbf{v}_{i},\textbf{v}_{j}$ are the coefficient vectors in $\hat{\bm{D}}_{ij}$ for $i,j\in\{A,B\}$. 

Suppose that this parallel condition holds, so that $\hat{\bm{D}}_{\alpha\beta} = d_{\alpha\beta}\hat{\bm{S}}_{\beta}$ for some constants $d_{\alpha\beta}$. Matching to \eqref{full rank 1 BM} then requires the following equations hold:
\begin{equation}
    \begin{split}
        S^{\text{th}}_{A} &= \gamma_{A}+\frac{K_{g}^{2}}{4\gamma_{B}}+\frac{1}{\gamma_{AA}}+\frac{1}{\gamma_{AB}} + \gamma_{AA}d^{2}_{AA}+\gamma_{BA}d^{2}_{BA}\\
          S^{\text{th}}_{B} &= \gamma_{B}+\frac{K_{g}^{2}}{4\gamma_{A}}+\frac{1}{\gamma_{BB}}+\frac{1}{\gamma_{BA}} + \gamma_{BB}d^{2}_{BB}+\gamma_{AB}d^{2}_{AB}
    \end{split}
\end{equation}
Now note that:
\begin{equation}
    \gamma_{AA}d_{AA}^2+\frac{1}{\gamma_{AA}}\geq 2|d_{AA}|, \hspace{0.5cm}\gamma_{BB}d^{2}_{BB}+\frac{1}{\gamma_{BB}} \geq 2|d_{BB}|
\end{equation}
Defining $\tilde{S}^{\text{th}}_{A} \equiv S^{\text{th}}_{A}-2|d_{AA}|$ and $\tilde{S}^{\text{th}}_{B} \equiv S^{\text{th}}_{B}-2|d_{BB}|$. We can thus rewrite:
\begin{equation}
\tilde{S}^{\text{th}}_{A} \geq (\sqrt{\gamma_{A}})^2+\left(\frac{K_{g}}{2\sqrt{\gamma_{B}}}\right)^2+\left(\frac{1}{\sqrt{\gamma_{AB}}}\right)^2+\left(|d_{BA}|\sqrt{\gamma_{BA}}\right)^2, \hspace{0.5cm} \tilde{S}^{\text{th}}_{B} \geq (\sqrt{\gamma_{B}})^2+\left(\frac{K_{g}}{2\sqrt{\gamma_{A}}}\right)^2+\left(\frac{1}{\sqrt{\gamma_{BA}}}\right)^2+\left(|d_{AB}|\sqrt{\gamma_{AB}}\right)^2
\end{equation}
Under Cauchy-Schwarz, we then have:
\begin{equation}
    \tilde{S}_{A}\tilde{S}_{B} \geq (K_{g}+|d_{AB}|+|d_{BA}|)^2
\end{equation}
Plugging back in gives the final bound:
\begin{equation}
\label{damping LOCC condition}
    (S^{\text{th}}_{A}-2|d_{AA}|)(S^{\text{th}}_{B}-2|d_{BB}|) \geq (K_{g}+|d_{AB}|+|d_{BA}|)^2
\end{equation}
Under the case of no damping, $d_{\alpha\beta} = 0$, which reduces to our usual bound.

We can now use this result to do an exact parameter matching like in the no-damping case. Suppose we choose:
\begin{equation}
    \gamma_{AA} = \frac{1}{|d_{AA}|}, \hspace{0.3cm}\gamma_{BB} = \frac{1}{|d_{BB}|}
\end{equation}
so that
\begin{equation}
    \gamma_{AA}d_{AA}^2 + \frac{1}{\gamma_{AA}} = 2|d_{AA}|,\hspace{0.3cm} \gamma_{BB}d_{BB}^2+\frac{1}{\gamma_{BB}} = 2|d_{BB}|
\end{equation} Using the same notation as before (defining $\tilde{S}^{\text{th}}_{i} = S^{\text{th}}_{i} - 2|d_{ii}|$, this would then give:
\begin{equation}
    \tilde{S}^{\text{th}}_{A} = \gamma_{A}+\frac{K_{g}^{2}}{4\gamma_{B}}+\frac{1}{\gamma_{AB}}+\gamma_{BA}d_{BA}^{2}, \hspace{0.5cm}     \tilde{S}^{\text{th}}_{B} = \gamma_{B}+\frac{K_{g}^{2}}{4\gamma_{A}}+\frac{1}{\gamma_{BA}}+\gamma_{AB}d_{AB}^{2}
\end{equation}
Now define $t = \sqrt{\tilde{S}^{\text{th}}_{A}/\tilde{S}^{\text{th}}_{B}}$ and make the following choices:
\begin{equation}
    \gamma_{A} = \frac{K_{g}}{2}t, \quad \gamma_{B} = \frac{K_{g}}{2t}, \quad \gamma_{AB} = \frac{1}{|d_{AB}|t}, \quad \gamma_{BA} = \frac{t}{|d_{BA}|}
\end{equation}
we have:
\begin{equation}
    \tilde{S}^{\text{th}}_{A} = (K_{g}+|d_{AB}|+|d_{BA}|)t, \hspace{0.5cm} \tilde{S}^{\text{th}}_{B} = (K_{g}+|d_{BA}|+|d_{AB}|)\frac{1}{t}
\end{equation}
These choice of parameters are valid as long as:
\begin{equation}
    \tilde{S}^{\text{th}}_{A}\tilde{S}^{\text{th}}_{B} = (K_{g}+|d_{AB}|+|d_{BA}|)^2
\end{equation}
but from the Cauchy-Schwarz argument, we know that this is only the instance where the inequality is saturated. More generally, this implies that:
\begin{equation}
    \tilde{S}^{\text{th}}_{A} = \gamma_{A}+\frac{K_{g}^{2}}{4\gamma_{B}}+\frac{1}{\gamma_{AB}}+\gamma_{BA}d_{BA}^{2}+L_{A}, \hspace{0.5cm}     \tilde{S}^{\text{th}}_{B} = \gamma_{B}+\frac{K_{g}^{2}}{4\gamma_{A}}+\frac{1}{\gamma_{BA}}+\gamma_{AB}d_{AB}^{2}+L_{B}
\end{equation}
for some local noise terms $L_{A},L_{B}$. It remains to check that these noise terms are indeed positive and can thus be reproduced using an LOCC protocol. Using the same choice of parameters as before gives:
\begin{equation}
    L_{A} = t\left(\sqrt{\tilde{S}^{\text{th}}_{A}\tilde{S}^{\text{th}}_{B}}-(K_{g}+|d_{AB}|+|d_{BA}|)\right), \quad L_{B} = \frac{1}{t}\left(\sqrt{\tilde{S}^{\text{th}}_{A}\tilde{S}^{\text{th}}_{B}}-(K_{g}+|d_{AB}|+|d_{BA}|)\right)
\end{equation}
Since we already assumed $\tilde{S}^{\text{th}}_{A}\tilde{S}^{\text{th}}_{B} \geq (K_{g}+|d_{AB}|+|d_{BA}|)^2$, both $L_{A},L_{B} \geq 0$ and are thus purely local noise terms that can be produced using just local weak measurements. Thus, we have shown that whenever $\eqref{damping LOCC condition}$ is satisfied, we are able to construct a genuine LOCC scheme that preserves separability. Of course, this bound only holds in the case where the stringent parallel condition $\hat{\bm{D}}_{\alpha\beta} = d_{\alpha\beta}\hat{\bm{S}}_{\beta}$ holds. This is not generally true and in that case, our specific LOCC scheme does not hold.

\section{VI. Derivation of General Interaction Born Markov Equation}
\label{general BM appendix}

The general interaction Hamiltonian is now:
\begin{equation}
\begin{split}
        \hat{H}_{\text{sys}} &= \frac{1}{2}\hat{\bm{\xi}}^{T}(\bm{M}+\bm{G})\hat{\bm{\xi}}\\
       \hat{H}_{\text{bath}} &= \sum_{j}\left(\frac{\hat{p}_{j}^2}{2m_{j}}+\frac{1}{2}m_{j}\omega_{j}^2\hat{q}_{j}^2\right)\\
    \hat{H}_{\text{sys-bath}} &= \hat{\bm{\xi}}_{A}^{T}\hat{\bm{F}}_{A}+\hat{\bm{\xi}}^{T}_{B}\hat{\bm{F}}_{B} = \sum_{\alpha=\{A,B\}}\sum_{i=1}^{2n_{\alpha}}(\hat{\bm{\xi}}^{T}_{\alpha})_{i}\otimes\hat{\bm{F}}_{\alpha,i} = \sum_{\mu=1}^{2(n_{A}+n_{B})}\hat{S}_{\mu}\otimes \hat{E}_{\mu}
\end{split}
\end{equation}
where $\hat{\bm{\xi}}^{T} = (\hat{\bm{\xi}}_{A}\otimes\mathbf{1}_{B}$, $\mathbf{1}_{A}\otimes \hat{\bm{\xi}}_{B})$, $\bm{M}^{T}=\bm{M}$ is the self interactions within each subsystem, $\bm{G}^{T} = \bm{G}$ is the coupling matrix between system $A$ and $B$ such that $\hat{\bm{\xi}}^{T}\bm{G}\hat{\bm{\xi}} = \hat{\bm{\xi}}^{T}_{A}\bm{Q}_{G}\hat{\bm{\xi}}_{B}+\hat{\bm{\xi}}^{T}_{B}\bm{Q}^{T}_{G}\hat{\bm{\xi}}_{A}$, $n_{A},n_{B}$ are the number of modes in the $A,B$ subsystems respectively, and:
\begin{equation}
    \begin{split}
        \hat{S}_{\mu} &\equiv \begin{cases}
        (\hat{\bm{\xi}}_{A})_{i}, \hspace{0.2cm}\mu = i \leq 2n_{A}\\
        (\hat{\bm{\xi}}_{B})_{i}, \hspace{0.2cm} \mu = 2n_{A}+i, 1\leq i\leq2n_{B}
        \end{cases}\\
        \hat{E}_{\mu} &\equiv \begin{cases}
            \hat{\bm{F}}_{A,i}, \hspace{0.2cm}\mu = i \leq 2n_{A}\\
            \hat{\bm{F}}_{B,i},\hspace{0.2cm}\mu=2n_{A}+i
        \end{cases}
    \end{split}
\end{equation}
We also define the indexing map:
\begin{equation}
    \alpha(\mu) = \begin{cases}
        A, \mu= 1,...,2n_{A}\\
        B,\mu = 2n_{A}+1,...,2n_{A}+2n_{B}
    \end{cases}, \hspace{0.2cm}i(\mu) = \begin{cases}
        \mu,\alpha(\mu) = A,\\
        \mu-2n_{A}, \alpha(\mu) = B
    \end{cases}
\end{equation}
Adopting the same notation from before, we now have:
\begin{equation}
    \mathcal{C}_{\mu\nu}(\tau) = \langle \hat{E}_{\mu}(\tau)\hat{E}_{\nu}\rangle \approx \frac{1}{2}\langle \{\hat{\bm{F}}_{\alpha(\mu),i(\mu)}(\tau)\hat{\bm{F}}_{\alpha(\nu),i(\nu)}\}\rangle = \delta(\tau)\delta_{\alpha(\mu),\alpha(\nu)}\bm{Q}^{\text{th}}_{\alpha(\mu),i(\mu)i(\nu)}
\end{equation}
where we assumed white noise spectrum $\frac{1}{2}\langle\{\hat{F}_{\alpha,i}(\tau),\hat{F}_{\beta,j}\}\rangle = \delta(\tau) \delta_{\alpha,\beta}\bm{Q}^{\text{th}}_{\alpha,ij}$ and have dropped the commutator (friction) term.

We then have:
\begin{equation}
    \begin{split}
        \hat{\bm{B}}_{\mu}&\equiv\int_{0}^{\infty}d\tau\sum_{\nu}\mathcal{C}_{\mu\nu}(\tau)\hat{\bm{S}}_{\nu}(-\tau) = \frac{1}{2}\sum_{\nu}\delta_{\alpha(\mu),\alpha(\nu)}\bm{Q}^{\text{th}}_{\alpha(\mu),i(\mu)i(\nu)}(\hat{\bm{\xi}}_{\alpha(\nu)})_{i(\nu)}\\
        \hat{\bm{C}}_{\mu} &\equiv \int_{0}^{\infty}d\tau\sum_{\nu}\mathcal{C}_{\nu\mu}(-\tau)\hat{\bm{S}}_{\nu}(-\tau)=\frac{1}{2}\sum_{\nu}\delta_{\alpha(\mu),\alpha(\nu)}\bm{Q}^{\text{th}}_{\alpha(\nu),i(\nu)i(\mu)}(\hat{\bm{\xi}}_{\alpha(\nu)})_{i(\nu)}
    \end{split}
\end{equation}
Now plug into the Born-Markov Master equation:
\begin{equation}
\begin{split}
    \dot{\hat{\bm{\rho}}}_{S}(t) &= -i[\hat{H}_{\text{sys}},\hat{\bm{\rho}}_{S}(t)]-\sum_{\mu}[(\hat{\bm{\xi}}_{\alpha(\mu)})_{i(\mu)},\frac{1}{2}\sum_{\nu}\delta_{\alpha(\mu),\alpha(\nu)}\bm{Q}^{\text{th}}_{\alpha(\mu),i(\mu)i(\nu)}(\hat{\bm{\xi}}_{\alpha(\nu)})_{i(\nu)}\hat{\bm{\rho}}_{S}(t)]\\
    &-\sum_{\mu}[\hat{\bm{\rho}}_{S}(t)\frac{1}{2}\sum_{\nu}\delta_{\alpha(\mu),\alpha(\nu)}\bm{Q}^{\text{th}}_{\alpha(\nu),i(\nu)i(\mu)}(\hat{\bm{\xi}}_{\alpha(\nu)})_{i(\nu)},(\hat{\bm{\xi}}_{\alpha(\mu)})_{i(\mu)}]\\
    &=-i[\hat{H}_{\text{sys}},\hat{\bm{\rho}}_{S}(t)]-\sum_{\alpha=\{A,B\},ij}[\hat{\bm{\xi}}_{\alpha,i},\frac{1}{2}\bm{Q}^{\text{th}}_{\alpha,ij}\hat{\bm{\xi}}_{\alpha,j}\hat{\bm{\rho}}_{S}(t)]-\sum_{\alpha=\{A,B\},ij}[\hat{\bm{\rho}}_{S}(t)\frac{1}{2}\bm{Q}^{\text{th}}_{\alpha,ij}\hat{\bm{\xi}}_{\alpha,i},\hat{\bm{\xi}}_{\alpha,j}]\\
    &=-i[\hat{H}_{HO},\hat{\bm{\rho}}_{S}(t)]-\frac{i}{2}[\hat{\bm{\xi}}^{T}_{A}\bm{Q}_{G}\hat{\bm{\xi}}_{B}+\hat{\bm{\xi}}^{T}_{B}\bm{Q}^{T}_{G}\hat{\bm{\xi}}_{A},\hat{\bm{\rho}}_{S}(t)]+\sum_{\alpha,ij}\bm{Q}^{\text{th}}_{\alpha,ij}\left[\hat{\bm{\xi}}_{\alpha,i}\hat{\bm{\rho}}_{S}(t)\hat{\bm{\xi}}_{\alpha,j}-\frac{1}{2}\{\hat{\bm{\xi}}_{\alpha,j}\hat{\bm{\xi}}_{\alpha,i},\hat{\bm{\rho}}_{S}(t)\}\right]
\end{split}
\end{equation}
where in the final line we used $\bm{Q}^{\text{th}} = \bm{Q}^{T,\text{th}}$ and in the final line we use:
\begin{equation}
    \sum_{\mu,\nu} = \sum_{\mu\in A,\nu\in A}+\sum_{\mu\in A,\nu\in B}+\sum_{\mu\in B,\nu \in A}+\sum_{\mu\in B,\nu\in B} \overset{\delta_{\alpha(\mu),\alpha(\nu)}}{\Longrightarrow} \sum_{\mu\in A,\nu\in A}+\sum_{\mu\in B,\nu\in B} = \sum_{\alpha=\{A,B\}}\sum_{ij=1}^{2n_{\alpha}}
\end{equation}

\section{VII. Full GKSL Proof for General Interaction}
\label{general GKSL appendix}

The target Born-Markov Master equation to reproduce is:
\begin{equation}
    \dot{\hat{\bm{\rho}}}_{S}(t)=-i[\hat{H}_{HO},\hat{\bm{\rho}}_{S}(t)]-\frac{i}{2}[\hat{\bm{\xi}}^{T}_{A}\bm{Q}_{G}\hat{\bm{\xi}}_{B}+\hat{\bm{\xi}}^{T}_{B}\bm{Q}^{T}_{G}\hat{\bm{\xi}}_{A},\hat{\bm{\rho}}_{S}(t)]+\sum_{\alpha,ij}\bm{Q}^{\text{th}}_{\alpha,ij}\left[\hat{\bm{\xi}}_{\alpha,i}\hat{\bm{\rho}}_{S}(t)\hat{\bm{\xi}}_{\alpha,j}-\frac{1}{2}\{\hat{\bm{\xi}}_{\alpha,j}\hat{\bm{\xi}}_{\alpha,i},\hat{\bm{\rho}}_{S}(t)\}\right] 
\end{equation}
First assume $\bm{Q}^{\text{th}}_{\alpha} \succ0$ is real symmetric (recall correlation matrices are necessarily PSD), so that there exist a (possibly rectangular) real matrix $\bm{R}_{A}$ such that $\bm{Q}^{\text{th}}_{A} = \bm{R}^{T}_{A}\bm{R}_{A}$. We then define local Hermitian operators:
\begin{equation}
    \hat{Z}_{A,k} \equiv (\bm{R}_{A}\hat{\bm{\xi}}_{A})_{k} = \sum_{\mu}(\bm{R}_{A})_{k\mu}\hat{\bm{\xi}}_{A\mu}
\end{equation}
which is Hermitian because it is a linear combination of Hermitian operators.

Now note the following:
\begin{equation}
    \begin{split}
        \sum_{k}(\hat{Z}_{A,k}\hat{\bm{\rho}}_{S}(t)\hat{Z}_{A,k}-\frac{1}{2}\{\hat{Z}^{2}_{A,k},\hat{\bm{\rho}}_{S}(t)\}) &= \sum_{\mu\nu}\left(\sum_{k}(\bm{R}_{A})_{k\mu}(\bm{R}_{A})_{k\nu}\right)\left(\hat{\bm{\xi}}_{A\mu}\hat{\bm{\rho}}_{S}(t)\hat{\bm{\xi}}_{A\nu}-\frac{1}{2}\{\hat{\bm{\xi}}_{A\mu}\hat{\bm{\xi}}_{A\nu},\hat{\bm{\rho}}_{S}(t)\}\right)\\
        &=\sum_{\mu\nu}\left(\bm{R}^{T}_{A}\bm{R}_{A}\right)_{\mu\nu}\left(\hat{\bm{\xi}}_{A\mu}\hat{\bm{\rho}}_{S}(t)\hat{\bm{\xi}}_{A\nu}-\frac{1}{2}\{\hat{\bm{\xi}}_{A\mu}\hat{\bm{\xi}}_{A\nu},\hat{\bm{\rho}}_{S}(t)\}\right)\\
        &=\sum_{\mu\nu}(\bm{Q}^{\text{th}}_{A})_{\mu\nu}\left(\hat{\bm{\xi}}_{A\mu}\hat{\bm{\rho}}_{S}(t)\hat{\bm{\xi}}_{A\nu}-\frac{1}{2}\{\hat{\bm{\xi}}_{A\mu}\hat{\bm{\xi}}_{A\nu},\hat{\bm{\rho}}_{S}(t)\}\right)
    \end{split}
\end{equation}
The same holds for the $B$ equation with $\bm{Q}^{\text{th}}_{B} = \bm{R}_{B}^{T}\bm{R}_{B}$. 

Now define the whitened coupling matrix:
\begin{equation}
    \widetilde{\bm{Q}}_{G} \equiv (\bm{Q}^{\text{th}}_{A})^{-1/2}\bm{Q}_{G}(\bm{Q}^{\text{th}}_{B})^{-1/2}
\end{equation}
Since $\widetilde{\bm{Q}}_{G}$ is real, use the real SVD:
\begin{equation}
    \widetilde{\bm{Q}}_{{G}} = \bm{U\Sigma V}^{T}, \hspace{0.2cm}\bm{U}\in O(2n_{A}), \hspace{0.2cm} \bm{V}\in O(2n_{B}), \hspace{0.2cm} \bm{\Sigma} = \text{diag}(s_{1},...,s_{r}), s_{k}\geq0, r\leq\min(2n_{A},2n_{B})
\end{equation}

We then define the (still Hermitian) operator vectors:
\begin{equation}
    \hat{\bm{Y}}_{A} \equiv (\bm{Q}^{\text{th}}_{A})^{1/2}\hat{\bm{\xi}}_{A}, \hspace{0.2cm} \hat{\bm{Y}}_{B} \equiv (\bm{Q}^{\text{th}}_{B})^{1/2}\hat{\bm{\xi}}_{B}, \hspace{0.2cm}\hat{\bm{Z}}_{A} \equiv \bm{U}^{T}\hat{\bm{Y}}_{A}, \hspace{0.2cm} \hat{\bm{Z}}_{B} \equiv \bm{V}^{T}\hat{\bm{Y}}_{B} 
\end{equation}
Now the interaction terms become:
\begin{equation}
    \begin{split}
\hat{\bm{\xi}}^{T}_{A}\bm{Q}_{G}\hat{\bm{\xi}}_{B}+\hat{\bm{\xi}}^{T}_{B}\bm{Q}_{G}^{T}\hat{\bm{\xi}}_{A} &= \hat{\bm{Y}}^{T}_{A}\widetilde{\bm{Q}}_{G}\hat{\bm{Y}}_{B}+\hat{\bm{Y}}^{T}_{B}\widetilde{\bm{Q}}^{T}_{G}\hat{\bm{Y}}_{A} = \hat{\bm{Z}}^{T}_{A}\bm{\Sigma}\hat{\bm{Z}}_{B}+\hat{\bm{Z}}^{T}_{B}\bm{\Sigma}\hat{\bm{Z}}_{A} = \sum_{k=1}^{r}2s_{k}\hat{Z}_{A,k}\hat{Z}_{B,k}
    \end{split}
\end{equation}
where we use the fact that since $\bm{Q}^{\text{th}}_{A,B}$ are positive definite, there exist symmetric square root such that $\bm{Q}^{\text{th}}_{A,B} = (\bm{Q}^{\text{th}}_{A,B})^{1/2}(\bm{Q}^{\text{th}}_{A,B})^{1/2} = (\bm{Q}^{\text{th}}_{A,B})^{1/2,T}(\bm{Q}^{\text{th}}_{A,B})^{1/2}$.
Now note that:
\begin{equation}
    \begin{split}
        &\sum_{\mu\nu}(\bm{Q}^{\text{th}}_{A})_{\mu\nu}\left(\hat{\bm{\xi}}_{A\mu}\hat{\bm{\rho}}_{S}(t)\hat{\bm{\xi}}_{A\nu}-\frac{1}{2}\{\hat{\bm{\xi}}_{A\mu}\hat{\bm{\xi}}_{A\nu},\hat{\bm{\rho}}_{S}(t)\}\right) \\
        &= \sum_{\mu\nu}(\bm{Q}^{\text{th}}_{A})_{\mu\nu}\sum_{ij}(\bm{Q}^{\text{th}}_{A})^{-1/2}_{\mu i}(\bm{Q}^{\text{th}}_{A})^{-1/2}_{\nu j}\left(\hat{\bm{Y}}_{Ai}\hat{\bm{\rho}}_{S}(t)\hat{\bm{Y}}_{A j}-\frac{1}{2}\{\hat{\bm{Y}}_{A i}\hat{\bm{Y}}_{A j},\hat{\bm{\rho}}_{S}(t)\}\right)\\
        &=\sum_{ij}((\bm{Q}^{\text{th}}_{A})^{-1/2}\bm{Q}^{\text{th}}_{A}(\bm{Q}^{\text{th}}_{A})^{-1/2})_{ij}\left(\hat{\bm{Y}}_{Ai}\hat{\bm{\rho}}_{S}(t)\hat{\bm{Y}}_{A j}-\frac{1}{2}\{\hat{\bm{Y}}_{A i}\hat{\bm{Y}}_{A j},\hat{\bm{\rho}}_{S}(t)\}\right)\\
        &=\sum_{ij}\delta_{ij}\left(\hat{\bm{Y}}_{Ai}\hat{\bm{\rho}}_{S}(t)\hat{\bm{Y}}_{A j}-\frac{1}{2}\{\hat{\bm{Y}}_{A i}\hat{\bm{Y}}_{A j},\hat{\bm{\rho}}_{S}(t)\}\right)\\
        &=\sum_{i}\sum_{kl}(\bm{U}_{A})_{ik}(\bm{U}_{A})_{il}\left(\hat{\bm{Z}}_{Ak}\hat{\bm{\rho}}_{S}(t)\hat{\bm{Z}}_{A l}-\frac{1}{2}\{\hat{\bm{Z}}_{A k}\hat{\bm{Z}}_{A l},\hat{\bm{\rho}}_{S}(t)\}\right)\\
        &=\sum_{k}\left(\hat{\bm{Z}}_{Ak}\hat{\bm{\rho}}_{S}(t)\hat{\bm{Z}}_{A k}-\frac{1}{2}\{\hat{\bm{Z}}_{A k}^{2},\hat{\bm{\rho}}_{S}(t)\}\right)
    \end{split}
\end{equation}
where in the second to last line, we used $\sum_{i}(\bm{U}_{A})_{ik}(\bm{U}_{A})_{il} = \sum_{i}(\bm{U}^{T}_{A})_{ki}(\bm{U}_{A})_{il} = (\bm{U}^{T}_{A}\bm{U}_{})_{kl} = \delta_{kl}$.
An identical calculation holds for the $B$ diffusion term. Thus, the Born-Markov Master equation can be represented as a decoupled sum:
\begin{equation}
\label{general int to match BM}
    \mathcal{L}^{\text{tot}}(\rho) = -i[\hat{H}_{HO},\hat{\bm{\rho}}_{S}(t)]-i\sum_{k=1}^{r}s_{k}[\hat{Z}_{A,k}\hat{Z}_{B,k},\hat{\bm{\rho}}_{S}(t)]+\sum_{\alpha=\{A,B\}}\sum_{k=1}^{2n_{\alpha}}\left(\hat{\bm{Z}}_{\alpha k}\hat{\bm{\rho}}_{S}(t)\hat{\bm{Z}}_{\alpha k}-\frac{1}{2}\{\hat{\bm{Z}}_{\alpha k}^{2},\hat{\bm{\rho}}_{S}(t)\}\right)
\end{equation}

Now consider local Gaussian weak measurements of $\hat{Z}_{A,k}$ and $\hat{Z}_{B,k}$ with strengths $\gamma_{A,k}>0$ and $\gamma_{B,k}>0$ with outcomes $y_{A,k},y_{B,k}\in\mathbb{R}$ via the Kraus operators
\begin{equation}
    \begin{split}
           K_{A,k}(y_{A,k}) &\equiv \left(\frac{2\gamma_{A,k}dt}{\pi}\right)^{1/4}\exp\left[-\gamma_{A,k}dt(y_{A,k}-\hat{Z}_{A,k})^2\right] \\
        K_{B,k}(y_{B,k}) &\equiv \left(\frac{2\gamma_{B,k}dt}{\pi}\right)^{1/4}\exp\left[-\gamma_{B,k}dt(y_{B,k}-\hat{Z}_{B,k})^2\right]
    \end{split}
\end{equation}
From the scalar equations we know this satisfies $\int dy_{A,k}K^{\dagger}_{A,k}K_{A,k}=1$ and same for B. Given the outcomes, we then apply symmetric feedback local unitaries:
\begin{equation}
    \begin{split}
          U_{A,k}(y_{B,k}) &\equiv \exp\left(-i\lambda_{A,k}y_{B,k}\hat{Z}_{A,k}dt\right)\\
        U_{B,k}(y_{A,k}) &\equiv \exp\left(-i\lambda_{B,k}y_{A,k}\hat{Z}_{B,k}dt\right)
    \end{split}
\end{equation}
with real gains $\lambda_{A,k},\lambda_{B,k}\in\mathbb{R}$. The overall Kraus operator conditioned on $(y_{A,k},y_{B,k})$ is:
\begin{equation}
            K_{k}(y_{A,k},y_{B,k}) \equiv [U_{A,k}(y_{B,k})K_{A,k}(y_{A,k})]\otimes[U_{B,k}(y_{A,k})K_{B,k}(y_{B,k})]
\end{equation}
which is manifestly a product Kraus operator. Thus, the unconditional maps is:
\begin{equation}
    \mathcal{E}^{(k)}(\rho) = \int dy_{A,k}dy_{B,k}K_{k}(y_{A,k},y_{B,k})\rho K_{k}^{\dagger}(y_{A,k},y_{B,k})
\end{equation}
is thus LOCC and the dynamics are separable. Note that this is the same setup as the Rank-1 procedure. Porting over the results from \eqref{rank one LOCC res}, we get the master equation:
\begin{equation}
    \mathcal{L}^{(k)}(\rho) = -i\lambda_{k}[\hat{Z}_{A,k}\hat{Z}_{B,k},\rho]+\left(\gamma_{A,k}+\frac{\lambda_{k}^2}{4\gamma_{B,k}}\right)\mathcal{D}[\hat{Z}_{A,k}]\rho + \left(\gamma_{B,k}+\frac{\lambda_{k}^2}{4\gamma_{A,k}}\right)\mathcal{D}[\hat{Z}_{B,k}]\rho
\end{equation}
Now match to the total BM equation \eqref{general int to match BM}, we require:
\begin{equation}
    \lambda_{k} = s_{k}, \hspace{0.2cm} \gamma_{A,k}+\frac{s_{k}^{2}}{4\gamma_{B,k}} = 1, \hspace{0.2cm} \gamma_{B,k}+\frac{s_{k}^{2}}{4\gamma_{A,k}}=1
\end{equation}
Solving gives:
\begin{equation}
    \gamma_{B,k}=\gamma_{A,k} =\frac{1}{2}\left(1\pm\sqrt{1-s_{k}^{2}}\right)
\end{equation}

Now consider the matrix:
\begin{equation}\bm{M}=
    \begin{pmatrix}
        \bm{Q}^{\text{th}}_{A} & \bm{Q}_{G}\\
        \bm{Q}_{G}^{T} & \bm{Q}^{\text{th}}_{B}
    \end{pmatrix}
\end{equation}
Suppose $M \succ 0$. Then we know the Schur complement satisfies:
\begin{equation}
    \bm{Q}^{\text{th}}_{B}-\bm{Q}^{T}_{G}(\bm{Q}^{\text{th}}_{A})^{-1}\bm{Q}_{G}\succ 0
\end{equation}
Now multiply on both sides by $(\bm{Q}^{\text{th}}_{B})^{-1/2}$:
\begin{equation}
    1-(\bm{Q}^{\text{th}}_{B})^{-1/2}\bm{Q}^{T}_{G}(\bm{Q}^{\text{th}}_{A})^{-1}\bm{Q}_{G}(\bm{Q}^{\text{th}}_{B})^{-1/2}\succ0
\end{equation}
This can be rewritten as
\begin{equation}
    1-[(\bm{Q}^{\text{th}}_{A})^{-1/2}\bm{Q}_{G}(\bm{Q}^{\text{th}}_{B})^{-1/2}]^{T}[(\bm{Q}^{\text{th}}_{A})^{-1/2}\bm{Q}_{G}(\bm{Q}^{\text{th}}_{B})^{-1/2}]=1-\widetilde{\bm{Q}}_{G}^{T}\widetilde{\bm{Q}}_{G}=1-\bm{V}(\bm{\Sigma})^2\bm{V}^{T}\succ0
\end{equation}
Thus we know $1-s_{k}^2 > 0$ so $s_{k} < 1 \forall k$. If this is satisfied then the $\gamma_{A,k},\gamma_{B,k}$ are real valued. We now just need to sum over all $k$ to get:
\begin{equation}
    \mathcal{L}^{\text{tot}}(\rho) = -i[\hat{H}_{HO},\hat{\bm{\rho}}_{S}(t)]+\sum_{k=1}^{r}\mathcal{L}^{(k)}(\rho)+\sum_{\alpha=\{A,B\}}\sum_{k=r+1}^{2n_{\alpha}}\left(\hat{\bm{Z}}_{\alpha k}\hat{\bm{\rho}}_{S}(t)\hat{\bm{Z}}_{\alpha k}-\frac{1}{2}\{\hat{\bm{Z}}_{\alpha k}^{2},\hat{\bm{\rho}}_{S}(t)\}\right)
\end{equation}
Comparing with \eqref{general int to match BM}, we see that the only terms left are only local (the local unitary with $\hat{H}_{HO}$ and extra $\mathcal{D}[\hat{Z}_{\alpha k}]\hat{\bm{\rho}}$ terms) so these can be implemented individually and independently via LOCC schemes. More explicitly, the extra $\mathcal{D}[\hat{Z}_{\alpha k}]\hat{\bm{\rho}}$ terms can be implemented via pure weak measurement (picking feedback strenght $\lambda = 0$ in the measurement-feedback schemes). Finally, note that because $[\hat{Z}_{A,k},\hat{Z}_{A,j}] \neq 0$ for $k\neq j$, a Lie-Trotter decomposition would be needed to implement $e^{t\sum_{k}\mathcal{L}^{(k)}}$. Nevertheless, each infinitesimal Trotter factor is LOCC, so the dynamics at every finite Trotter step is LOCC, and so will the continuous-time limit of the full procedure.

The above argument can be generalized to $\bm{Q}^{\text{th}}_{A,B}$ PSD. Since $\bm{Q}^{\text{th}}_{A,B}$ are PSD, there exists square roots $(\bm{Q}^{\text{th}}_{A,B})^{1/2}$ such that$\bm{Q}^{\text{th}}_{A,B} = [(\bm{Q}^{\text{th}}_{A,B})^{1/2}]^{T}(\bm{Q}^{\text{th}}_{A,B})^{1/2}=(\bm{Q}^{\text{th}}_{A,B})^{1/2}(\bm{Q}^{\text{th}}_{A,B})^{1/2}$. Now the key fact is the following \cite{horn2012matrix}:
\begin{equation}
    \begin{pmatrix}
        \bm{Q}^{\text{th}}_{A} & \bm{Q}_{G}\\
        \bm{Q}_{G}^{T} & \bm{Q}^{\text{th}}_{B}
    \end{pmatrix} \succeq 0 \Longleftrightarrow \exists \bm{X}\in \mathbb{R}^{2n_{A}\times 2n_{B}}\Big|\bm{Q}_{G} = (\bm{Q}^{\text{th}}_{A})^{1/2}\bm{X}(\bm{Q}^{\text{th}}_{B})^{1/2} \ \ \text{and} \ \ ||\bm{X}||_{\text{op}} \leq 1
\end{equation}
Since $\bm{X}$ is real, we can perform SVD on it.  $\bm{X}$ then plays the same role as the whitened coupling matrix $\widetilde{\bm{Q}}_{G}$ and the same argument follows.

Unlike in the rank-one scalar case, this same technique cannot be adapted for the correlated noise case where $\bm{Q}^{\text{th}}_{i}$ is no longer block diagonal. Indeed in that case, one finds that the total Lindbladian is no longer fully decoupled even after performing the proper matrix transformations. As such, it is likely that a more advanced LOCC scheme is required to reproduce the most general dynamics.

\end{document}